\documentclass[12pt, a4paper]{article}

\usepackage{graphicx}
\usepackage{bm}
\usepackage{amsmath, amssymb,theorem,verbatim}
\usepackage[round]{natbib} 
\usepackage{setspace}
\usepackage{multirow} 
\usepackage{hyperref}
\usepackage{algpseudocode}
\usepackage{algorithm}
\usepackage{textcomp} 
\usepackage{placeins}

\hypersetup{
  hypertexnames=false,
  pdftitle={LABS: Extending the scope of binary segmentation via a look-ahead device},
  pdfauthor={Piotr Fryzlewicz},
  pdfkeywords={change-point detection, piecewise-linear signal, recursive algorithm, trend estimation}
}

\long\def\comment#1{}

\newcommand{\qed}{\hfill$\square$}
\newenvironment{proof}[1][Proof]{\par\noindent\textit{#1.}\space}{\qed\par\medskip}

\newtheorem{theorem}{Theorem}[section]
\newtheorem{prop}{Proposition}[section] 
\newtheorem{lemma}{Lemma}[section]
\newtheorem{cor}{Corollary}[section] 
\newtheorem{defin}{Definition}[section] 
\newtheorem{assumption}{Assumption}[section]

\newtheorem{remark}{Remark}[section]

\newcommand{\mainthm}{Theorem~4.1}
\newcommand{\assterm}{Assumption~4.1}
\newcommand{\assnondeg}{Assumption~4.2}
\newcommand{\defpop}{Definition~4.1}
\newcommand{\proppop}{Proposition~4.1}
\newcommand{\algmain}{Algorithm~2}

\begin{document}

\title{LABS: Extending the scope of binary segmentation via a look-ahead device} 
\author{Piotr Fryzlewicz\thanks{Department of Statistics, London School of Economics and Political Science, London WC2A 2AE, UK. Correspondence: \texttt{p.fryzlewicz@lse.ac.uk}. ORCID: 0000-0002-9676-902X.}}
\date{}
\maketitle

\begin{abstract}
Binary segmentation is widely used for multiple change-point detection because it is fast, simple to describe, and simple to implement. Its validity rests on the requirement that, at each recursive stage, the procedure identifies one of the true change-points when several are present in the current interval. This holds for detecting changes in mean using the CUSUM statistic, but fails in some other settings, in particular in slope change detection for continuous piecewise-linear signals. We propose Look-Ahead Binary Segmentation (LABS), a modification in which the change-points returned by the two child recursions define a narrower interval on which the parent estimate is re-evaluated. LABS inherits the computational speed of standard binary segmentation, but achieves the near-optimal consistency rate of $O\{(n\log n)^{1/2}\}$ in the slope-change signal setting when the LABS model is chosen via either thresholding or a Schwarz-like information criterion. Simulations show that LABS is fast and achieves state-of-the-art performance.
\vspace{5pt}

\end{abstract}

\section{Introduction}

Binary segmentation is a frequently used method for segmenting ordered data, such as time series, into multiple regions of perceived homogeneity. In its simplest form, it recursively subdivides each current section of the data into two maximally contrasting subsections, starting from the entire dataset, and stops on a given section when that section is either too short to be divided further or considered homogeneous according to a certain criterion. There is a 1-to-1 correspondence between the detected segments and the estimated change-points that separate them.

In a time series setting, binary segmentation is often used to identify regions with approximately the same expectation, either of the time series itself or of a transformation of it; see for example \cite{v81}, \cite{venkatraman1993}, \cite{b97}, \cite{ccs11}, \cite{cf12} and \cite{fsr11}, and \cite{f07a} for the link with Unbalanced Haar wavelets. \cite{kfe12} describe it as ``arguably the most widely used changepoint search method''. \cite{brodsky1993nonparametric} observe that many types of change, including distributional ones, can be seen as changes in the mean of certain diagnostic sequences, which makes binary segmentation applicable beyond mean shifts. \cite{f14a} proposes Wild Binary Segmentation, which preserves the binary segmentation logic but evaluates the contrast on many deterministically or randomly drawn sub-intervals, an idea also present in \cite{kbm23}. Both of these methods apply to the mean-shift model, but do not generalize beyond it.

Outside the classical time series setting, binary segmentation logic appears in Automated Interaction Detection \citep{ms63} and Classification and Regression Trees (CART; \cite{bfso84}), in which the best variable on which to subdivide is chosen at each stage in a greedy way from the set of available predictive features, and is inherited by modern descendants of CART offering improvements via bagging (random forests; \cite{b01}) or boosting \citep{f01, c16, kmfwcmyl17}. In these contexts the mechanism is alternatively referred to as recursive binary partitioning. In contrast to regression trees, which typically use constant models in each terminal node, \cite{hhz06} introduce partitioned regression, in which terminal nodes contain full regression models.

Binary segmentation is popular because it is fast (typically $O(n \log n)$, where $n$ is the data length), easy to explain, straightforward to implement, and usable with a range of contrast functions across a range of stochastic models. Its success, however, relies on whether it is capable, at each recursive stage, of identifying one of the change-points if multiple are present in the current interval. When this condition holds, binary segmentation can recursively isolate and estimate all change-points. When it fails, the algorithm may detect spurious change-points far from any true locations.
\cite{venkatraman1993} shows that, in the piecewise-constant signal setting, binary segmentation used with the CUSUM statistic (a maximum-likelihood-based contrast for additive i.i.d.\ Gaussian noise) does have this property, which underpins the consistency of binary segmentation for detecting changes in mean.
By contrast, \cite{bcf16} show that in the piecewise-linear signal model the corresponding maximum-likelihood-based contrast can achieve its maximum far from any of the true change-points when multiple are present in the current interval, a point reiterated by \cite{mfl19}. This is the obstacle to extending binary segmentation to this and similar settings.

Several methods address the limitations of classical binary segmentation in settings where it fails; we review only those applicable to the piecewise-linear signal model. In the Narrowest-Over-Threshold (NOT) method, \cite{bcf16} sample many intervals and select the narrowest one on which the contrast exceeds a threshold, the idea being that sufficiently narrow intervals are likely to contain at most one change-point. A related idea, based on interval expansion rather than sampling, appears in Isolate-Detect \citep[ID;][]{af22}.
\cite{mfl19} introduce CPOP, a dynamic programming method for detecting changes in slope using an $L_0$ penalty. This approach is optimal for the criterion it minimizes, but the authors find empirically that it has $O(n^2)$ computational cost when the number of change-points is fixed, making it orders of magnitude slower in practice than the other techniques discussed here, a point we illustrate later.
\cite{mf23} propose TrendSegment, a bottom-up adaptive wavelet-based approach that focuses on local features in the early stages of merging. 
\cite{koc24} propose a moving sum (MOSUM) procedure that scans for changes using local parameter estimates from adjacent sliding windows, with $O(n)$ cost for a fixed collection of bandwidths.

We propose Look-Ahead Binary Segmentation (LABS), a modification of classical binary segmentation that extends its scope to settings where the standard algorithm fails. It employs a look-ahead mechanism that uses information from the recursive calls on the child intervals to refine the change-point estimate in the parent interval, at each recursive step. This localizes the estimation problem adaptively, as the recursion proceeds. The modification is confined to a single refinement step within the recursive flow, so the algorithm remains straightforward to implement and to describe. It is naturally expressed through recursion, and it has the same typical $O(n \log n)$ cost as classical binary segmentation under balanced recursive splits.

The basic version of LABS requires only the threshold parameter, which is also needed by classical binary segmentation. We also describe a grid-based extension. It uses an $M$-point grid for the proposal and an $R$-point grid for the re-test, providing additional localization within each call at a fixed multiple of the cost. We prove consistency for every fixed $M,R \geq 2$, and for selection from a threshold solution path by a strengthened Schwarz criterion. This mechanism differs from the one in NOT: LABS is consistent even without grid refinement, corresponding to $M=R=2$, whereas NOT requires the number of sampled sub-intervals to grow with $n$.

The name look-ahead follows related uses in recursive search and decision-tree algorithms, where tentative decisions are assessed through their child problems; Section~S6 of the supplement gives details.

The remainder of this paper is organized as follows. Section~\ref{sec:motivation} presents the motivation and problem statement, focusing on the piecewise-linear signal model, reviews classical binary segmentation, and explains why it fails for detecting changes in slope. Section~\ref{sec:labs} introduces the LABS algorithm and its grid-based extension, and compares it with existing methods. Section~\ref{sec:consistency} establishes consistency for the piecewise-linear model. Section~\ref{sec:simulations} reports a simulation study. Section~\ref{sec:data} applies LABS to real data. Proofs and further discussion are in the supplement appended to this preprint.

\section{Motivation and problem statement}
\label{sec:motivation}

Although LABS applies to a variety of change-point models, including changes in mean, changes in variance, and changes in higher-order polynomial structure, we present it in the context of detecting changes in slope in a continuous piecewise-linear signal. This is a canonical setting in which classical binary segmentation is known to fail \citep{bcf16, mfl19}. The continuous piecewise-linear model arises in applications such as the analysis of tropospheric ozone trends \citep{csks23}, COVID-19 infection curves \citep{jzs23}, vehicle state estimation \citep{hkpck25}, and intracellular transport \citep{ddcm25}.
Our application is to a global temperature time series; see Section \ref{sec:data}.

We observe $X_1, \ldots, X_n$ generated by
\begin{equation}
\label{eq:model}
    X_t = f_t + \varepsilon_t, \qquad t = 1, \ldots, n,
\end{equation}
where $\varepsilon_t$ are zero-mean random variables with common variance $\sigma^2$, and the signal $f_t$ is continuous and piecewise-linear. Specifically, we assume there exist $N$ change-points $1 < \tau_1 < \cdots < \tau_N < n$ and segment parameters $(\theta_{j,1}, \theta_{j,2})$ for $j = 1, \ldots, N\!+\!1$, such that $f_t = \theta_{j,1} + \theta_{j,2}\, t$ for $t = \tau_{j-1}+1, \ldots, \tau_j$ (with $\tau_0 = 0$, $\tau_{N+1} = n$), subject to the continuity constraint
\begin{equation}
\label{eq:continuity}
    \theta_{j,1} + \theta_{j,2}\, \tau_j \;=\; \theta_{j+1,1} + \theta_{j+1,2}\, \tau_j, \qquad j = 1, \ldots, N,
\end{equation}
which requires the left and right linear pieces to agree at each change-point $\tau_j$. The change-points are the indices at which the slope changes, that is, $\theta_{j,2} \neq \theta_{j+1,2}$, and our task is to estimate $N$ and the locations $\tau_1, \ldots, \tau_N$. The contrast functions below are derived under independent Gaussian errors, and the simulations of Section~\ref{sec:simulations} use independent Gaussian noise; the algorithm itself does not depend on this and only requires a contrast function to be supplied.

We first review classical binary segmentation. Given a contrast function $\mathcal{C}_{s,e}(b)$ that quantifies the evidence for a change-point at location $b$ within the interval $(s, e]$, binary segmentation proceeds as follows. $\mathcal{C}_{s,e}(b)$ is taken to be non-negative, as is the case for the CUSUM \eqref{eq:cusum} and piecewise-linear GLR \eqref{eq:lin_contrast} contrasts used here, which are both defined as absolute values; for a general signed contrast, replace $\mathcal{C}$ by $|\mathcal{C}|$ throughout.

\begin{algorithm}[H]
\caption{Classical Binary Segmentation}
\label{alg:bs}
\begin{algorithmic}[1]
\Function{BinSeg}{$X, s, e, \lambda$}
    \State $\hat{b} \gets \arg\max_{s < b < e} \mathcal{C}_{s,e}(b)$
    \If{$\mathcal{C}_{s,e}(\hat{b}) > \lambda$}
        \State Declare a change-point at $\hat{b}$
        \State \Call{BinSeg}{$X, s, \hat{b}, \lambda$} \Comment{Recurse on $(s, \hat{b}]$}
        \State \Call{BinSeg}{$X, \hat{b}, e, \lambda$} \Comment{Recurse on $(\hat{b}, e]$}
    \EndIf
\EndFunction
\end{algorithmic}
\end{algorithm}

\noindent The algorithm is launched by calling $\textsc{BinSeg}(X, 0, n, \lambda)$, where $\lambda > 0$ is a threshold parameter. In the piecewise-constant signal model with additive Gaussian noise, the natural contrast is the CUSUM statistic
\begin{equation}
\label{eq:cusum}
    \mathcal{C}^{\mathrm{CUSUM}}_{s,e}(b) = \sqrt{\frac{(e-b)(b-s)}{e-s}} \left| \frac{1}{b-s}\sum_{t=s+1}^{b} X_t - \frac{1}{e-b}\sum_{t=b+1}^{e} X_t \right|.
\end{equation}
A result of \cite{venkatraman1993} establishes that the CUSUM achieves its maximum near one of the true change-points even when multiple are present in $(s, e]$, which underpins the consistency of binary segmentation for detecting changes in mean.

In the continuous piecewise-linear model, the contrast derived from the generalized likelihood ratio \citep{bcf16} takes a different form. For each candidate location $b \in \{s+2, \ldots, e-1\}$, define $\ell = e - s$ and the normalizing quantities
{\small
\begin{equation}
\label{eq:alpha_beta}
\begin{aligned}
\alpha^b_{(s,e]}
&= \left(\frac{6}{\ell(\ell^2 - 1)\big(1 + (e\!-\!b\!+\!1)(b\!-\!s) + (e\!-\!b)(b\!-\!s\!-\!1)\big)}\right)^{\!1/2}\!\!\!, \\
\beta^b_{(s,e]}
&= \left(\frac{(e-b+1)(e-b)}{(b-s-1)(b-s)}\right)^{\!1/2}\!\!\!.
\end{aligned}
\end{equation}
}
The contrast vector $\phi^b_{(s,e]} = (\phi^b_{(s,e]}(1), \ldots, \phi^b_{(s,e]}(n))'$ has components
{\footnotesize
\begin{equation}
\label{eq:contrast_vector}
\phi^b_{(s,e]}(t) = \begin{cases}
\alpha^b_{(s,e]} \, \beta^b_{(s,e]} \Big[\big(3(b\!-\!s)+(e\!-\!b)-1\big)\, t \,-\, \big(b(\ell-1)+2(s+1)(b-s)\big)\Big], & t = s\!+\!1, \ldots, b, \\[6pt]
-\,\dfrac{\alpha^b_{(s,e]}}{\beta^b_{(s,e]}} \Big[\big(3(e\!-\!b)+(b\!-\!s)+1\big)\, t \,-\, \big(b(\ell-1)+2e(e-b+1)\big)\Big], & t = b\!+\!1, \ldots, e, \\[6pt]
0, & \text{otherwise},
\end{cases}
\end{equation}
}
and the contrast function is
\begin{equation}
\label{eq:lin_contrast}
    \mathcal{C}^{\mathrm{lin}}_{s,e}(b) = \left| \langle X, \, \phi^b_{(s,e]} \rangle \right|.
\end{equation}
The vector $\phi^b_{(s,e]}$ is the unit-norm continuous piecewise-linear function on $(s,e]$ with a single knot at $b$ that is orthogonal to the constant and linear functions on that interval. Two observations to the left of $b$ and one to its right are needed for the knot to be identified, which gives the candidate set $\{s+2, \ldots, e-1\}$.

When the interval $(s, e]$ contains exactly one change-point, maximizing $\mathcal{C}^{\mathrm{lin}}_{s,e}(b)$ over $b$ yields a consistent estimator of its location with near-optimal guarantees \citep{bcf16}. When multiple change-points are present, however, the maximizer can be far from all of them, unlike for CUSUM in the piecewise-constant model. To illustrate this, consider the ``trap'' signal
\begin{equation}
\label{eq:trap}
    f_t = \begin{cases}
        t/\tau_1, & t = 1, \ldots, \tau_1, \\
        1, & t = \tau_1 + 1, \ldots, \tau_2, \\
        (\tau_3 - t)/(\tau_3 - \tau_2), & t = \tau_2 + 1, \ldots, \tau_3,
    \end{cases}
\end{equation}
where $0 < \tau_1 < \tau_2 < \tau_3$, illustrated in Figure~\ref{fig:trap} with $\tau_1 = 200$, $\tau_2 = 400$, $\tau_3 = 600$, alongside the contrast $\mathcal{C}^{\mathrm{lin}}_{0,n}(b)$ computed on the full interval. The contrast profile is unimodal, with its maximum near $b = 300$, the midpoint between the two true change-points, rather than near either $\tau_1$ or $\tau_2$ (neither $\tau_1$ nor $\tau_2$ is even a local maximum). Binary segmentation would therefore detect a spurious change-point at the midpoint before recursing to the left and to the right. This failure \citep{bcf16, mfl19} is the problem that LABS is designed to address. Its cause is that the contrast is evaluated on an interval containing multiple change-points, so that localization to narrower intervals is needed. As shown in the next section, LABS achieves the localization adaptively as the recursion progresses, by passing change-point information from children to parents and using it to localize the estimation for the parent.

\begin{figure}[htbp]
\centering
\includegraphics[width=0.65\textwidth]{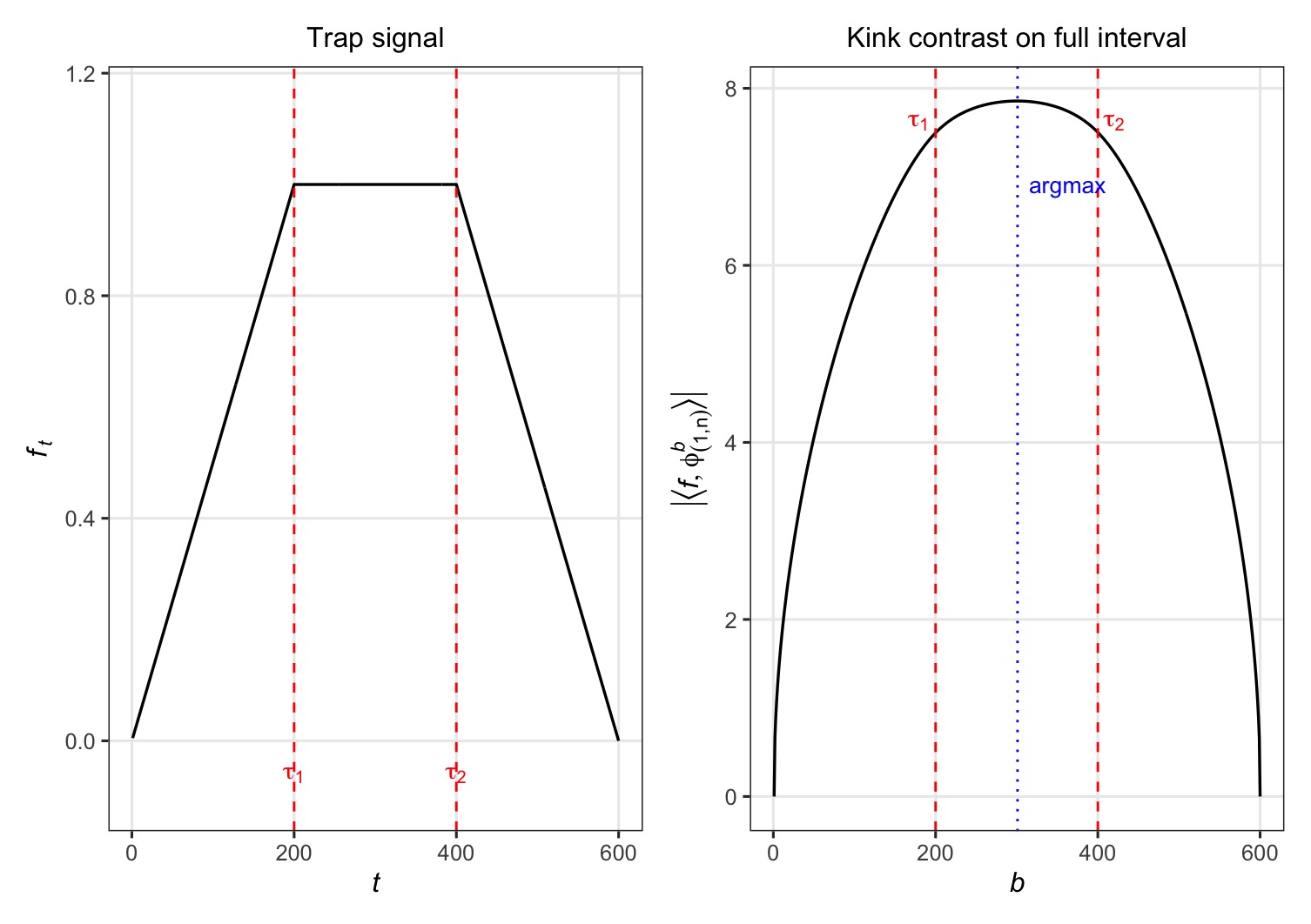}
\caption{Left: the trap signal~\eqref{eq:trap} with $\tau_1 = 200$, $\tau_2 = 400$, $\tau_3 = 600$; dashed red lines mark the true change-points. Right: the absolute kink contrast $|\langle f, \phi^b_{(0,n]} \rangle|$ computed over the full interval; its maximum (dotted blue) is near $b = 300$, away from either true change-point.}
\label{fig:trap}
\end{figure}

\section{Look-Ahead Binary Segmentation (LABS)}
\label{sec:labs}

\subsection{The algorithm}
\label{sec:labs_algorithm}

The failure of classical binary segmentation in the piecewise-linear setting stems from estimating change-point locations on intervals that contain multiple change-points. The methods reviewed in Section~\ref{sec:motivation}, most notably ID and NOT, address this by changing which intervals are considered. LABS instead preserves the recursive structure of binary segmentation: after detecting a candidate change-point and recursing into the left and right child intervals, the change-points detected in those children are used to define a narrower interval on which the parent change-point is re-evaluated. We refer to this device as look-ahead.

Suppose a candidate change-point $\hat{b}$ is detected in $(s, e]$, and on recursing let $\mathcal{L}$ and $\mathcal{R}$ be the sets of change-points detected in $(s, \hat{b}]$ and $(\hat{b}, e]$. When non-empty, $\max(\mathcal{L})$ provides a refined left boundary and $\min(\mathcal{R})$ a refined right boundary. Under the assumptions of Section~\ref{sec:consistency}, the resulting interval contains at most one previously undetected change-point, so re-evaluating the contrast on it gives a more accurate estimate.

Let $\mathcal{C}_{s,e}(b)$ denote a contrast function and $\lambda > 0$ a threshold. The algorithm is launched by calling $\textsc{LABS}(X, 0, n, \lambda)$. Here $B_{s,e}$ denotes the set of valid candidate locations for the contrast on $(s,e]$: $B_{s,e} = \{s+2, \ldots, e-1\}$ for the piecewise-linear GLR \eqref{eq:lin_contrast}, and $B_{s,e} = \{s+1, \ldots, e-1\}$ for the CUSUM \eqref{eq:cusum}. The contrast is evaluable on $(s,e]$ when $B_{s,e} \neq \emptyset$. We adopt the leftmost-maximizer convention throughout: $\arg\max$ is interpreted as $\min \arg\max$ when the maximizer is not unique.

\begin{algorithm}[H]
{\scriptsize
\caption{Look-Ahead Binary Segmentation (LABS)}
\label{alg:labs}
\begin{algorithmic}[1]
\Function{LABS}{$X, s, e, \lambda$}
    \If{$B_{s,e} = \emptyset$ (interval too short)}
        \State \Return $\emptyset$
    \EndIf
    \State $\hat{b} \gets \arg\max_{b \in B_{s,e}} \mathcal{C}_{s,e}(b)$
    \If{$\mathcal{C}_{s,e}(\hat{b}) > \lambda$} \Comment{Proposal phase}
        \State $\mathcal{L} \gets$ \Call{LABS}{$X, s, \hat{b}, \lambda$} \Comment{Recurse on $(s, \hat{b}]$}
        \State $\mathcal{R} \gets$ \Call{LABS}{$X, \hat{b}, e, \lambda$} \Comment{Recurse on $(\hat{b}, e]$}
        \State \Comment{Look-ahead refinement phase}
        \If{$\mathcal{L} \neq \emptyset$ \textbf{or} $\mathcal{R} \neq \emptyset$}
            \State $s' \gets \begin{cases} \max(\mathcal{L}) & \text{if } \mathcal{L} \neq \emptyset \\ s & \text{otherwise} \end{cases}$
            \State $e' \gets \begin{cases} \min(\mathcal{R}) & \text{if } \mathcal{R} \neq \emptyset \\ e & \text{otherwise} \end{cases}$
            \If{$B_{s',e'} \neq \emptyset$}
                \State $\hat{b}' \gets \arg\max_{b \in B_{s',e'}} \mathcal{C}_{s',e'}(b)$ \Comment{Re-test on $(s', e']$}
                \If{$\mathcal{C}_{s',e'}(\hat{b}') > \lambda$}
                    \State \Return $\mathcal{L} \cup \{\hat{b}'\} \cup \mathcal{R}$ \Comment{Accept refined estimate}
                \EndIf
            \EndIf
            \State \Return $\mathcal{L} \cup \mathcal{R}$ \Comment{Re-test failed: discard parent estimate}
        \EndIf
        \State \Return $\mathcal{L} \cup \{\hat{b}\} \cup \mathcal{R}$ \Comment{Both children empty: keep parent estimate}
    \Else
        \State \Return $\emptyset$
    \EndIf
\EndFunction
\end{algorithmic}
}
\end{algorithm}

There are two phases at each recursive level. The proposal phase is that of classical binary segmentation: a candidate $\hat{b}$ is identified by maximizing the contrast over the current interval, and if the contrast exceeds the threshold, the algorithm recurses into the two subintervals. In the look-ahead refinement phase, the change-point sets returned by the child recursions define a narrower interval $(s', e']$ on which the contrast is re-evaluated, and the parent change-point is discarded if the re-test does not exceed the threshold. The difference from classical binary segmentation is that information flows not only from parent to children, through the interval boundaries, but also from children back to parent, through the detected change-points.

To see how this corrects the trap problem, consider the trap signal~\eqref{eq:trap} with change-points at $\tau_1 = 200$ and $\tau_2 = 400$. The initial contrast on $(0, 600]$ is maximized near the midpoint $\hat{b} \approx 300$. Recursing, the left child $(0, 300]$ contains only $\tau_1$ and detects it; the right child $(300, 600]$ contains only $\tau_2$ and detects it. The look-ahead refinement then re-evaluates the contrast on the narrower interval $(\hat{\tau}_1, \hat{\tau}_2] \approx (200, 400]$, which corresponds to the flat middle section and contains no new change-points. The re-test does not exceed the threshold, and the spurious midpoint change-point is discarded. We revisit this example with noisy data in Section~\ref{sec:simulations}.


\subsection{Grid-based extension}
\label{sec:grid}

The basic LABS algorithm evaluates the contrast on the single interval $(s, e]$ at each recursive call. A natural extension evaluates the contrast on multiple sub-intervals of $(s, e]$, providing additional localization within each call. We form a grid of $M \geq 2$ equally spaced points within the interval $[s, e]$:
\[
g_1 = s, \quad g_2 = \left\lfloor s + \tfrac{e - s}{M - 1} \right\rceil, \quad \ldots, \quad g_M = e,
\]
where $\lfloor \cdot \rceil$ denotes rounding to the nearest integer. For every pair $(g_i, g_j)$ with $i < j$ and $g_j - g_i \geq 3$, we compute $\mathcal{C}_{g_i, g_j}(b)$ over each candidate $b \in \{g_i+2, \ldots, g_j-1\}$. The candidate change-point is the location achieving the global maximum across all such sub-intervals, and a change-point is declared if this maximum exceeds $\lambda$. When $M = 2$, only the full interval $(s, e]$ is evaluated, recovering Algorithm~\ref{alg:labs}.

The grid-based search complements the look-ahead mechanism. The look-ahead narrows the re-test interval using information from child recursions, across levels of the recursion tree. By contrast, the grid search identifies favorable sub-intervals within one recursive call. In our implementation it is applied in both the proposal and re-test phases. We use $M$ points for the proposal grid and $R$ points for the re-test grid. A single-interval evaluation, $R=2$, should usually suffice for the re-test because the look-ahead has already narrowed its interval, so the paper considers the case $R > 2$ for completeness rather than out of computational necessity.

Write $\textsc{GridSearch}(X,s,e,K)$ for the routine returning the pair $(\hat b,W)$, where $\hat b$ gives the maximum contrast across all sub-intervals formed by a $K$-point grid on $[s,e]$, and $W$ is that maximum. LABS-Grid replaces the proposal maximization in Algorithm~\ref{alg:labs} by $\textsc{GridSearch}(X,s,e,M)$ and the re-test maximization by $\textsc{GridSearch}(X,s',e',R)$. In each case it compares $W$ with $\lambda$. It also replaces the short-interval guard by the condition that the relevant grid search can be evaluated. Algorithm~S1 of the supplement gives the full procedure. A proposal search costs $O\{M^2(e-s)\}$ and a re-test search costs $O\{R^2(e'-s')\}$. The typical overall cost under balanced splits is therefore $O\{(M^2+R^2)n\log n\}$. For fixed small grids, such as those with $M\leq 5$ used below, these are fixed multiples of the corresponding LABS costs.

The same threshold $\lambda$ is used at the proposal and re-test stages. Because $\phi^b_{(s,e]}$ has unit norm, contrasts on sub-intervals of different lengths are directly comparable. A larger grid searches more sub-intervals, however, so its finite-sample null distribution changes. In Section~\ref{sec:simulations} an information criterion selects from a threshold path and accounts for this change. For a fixed-threshold implementation, the constant in $\lambda$ may need to depend on $M$ and $R$.

A distinction between the grid-based extension and NOT concerns the number of sub-intervals required for consistency. NOT draws $\mathcal{M}$ random intervals from the entire sample and requires $\mathcal{M} \to \infty$ as $n \to \infty$ to guarantee that at least one interval isolates each change-point. Under the assumptions of Theorem~\ref{thm:labs_consistency}, LABS is consistent even without grid refinement, the case $M=R=2$, because localization is provided by the look-ahead rather than by interval sampling. Increasing $M$ can improve finite-sample performance but is not needed for consistency. Section~S4 of the supplement extends the consistency result for LABS (Section \ref{sec:consistency}) to every fixed $M,R\geq2$. Its population recursion and its assumptions concern proposal windows, so they depend on $M$. A re-test window contains either no undetected kink or one undetected kink. With one kink, the full re-test interval has the unique largest population contrast for every fixed $R$, and with no kink the population contrast is zero. Thus no extra signal assumption is needed for $R$, and $R=2$ suffices. The proof uses the same threshold order for every fixed $M$ and $R$.

\section{Consistency of LABS}
\label{sec:consistency}

We now establish the consistency of LABS in the continuous piecewise-linear model of Section~\ref{sec:motivation}, with the contrast \eqref{eq:lin_contrast}. The algorithm analyzed is Algorithm~\ref{alg:labs} as stated, without grid refinement, so the result applies to the case $M = R = 2$. For an extension to $M, R \ge 2$, see Section S4 of the supplement.

We work in the fixed change-point regime, formulated as a triangular array. Fix $N$, fix $q_1 < \cdots < q_N$ in $(0,1)$ and non-zero reals $d_1, \ldots, d_N$, and put $q_0=0$, $q_{N+1}=1$ and $\delta_{\min}=\min_{0\leq i\leq N}(q_{i+1}-q_i)$. When $N\geq1$, also put $\underline d=\min_j|d_j|$ and $\bar d=\max_j|d_j|$. When $N=0$, the signal is affine, meaning that it is a straight line, with a constant signal included as a special case. For each $n$, let $f$ be as in \eqref{eq:model}, with change-points $\tau_j=\lfloor q_jn\rfloor$, slope jumps $\Delta_j:=\theta_{j+1,2}-\theta_{j,2}=d_j/n$, and independent, mean-zero, $\sigma$-sub-Gaussian errors. Let $g$ be the continuous piecewise-linear function on $[0,1]$ with kinks $q_j$ and slope jumps $d_j$. The overall level and initial slope of $f$ and $g$ are arbitrary. Adding a straight line to the signal does not change the contrast, which is orthogonal to constant and linear functions. Apart from this arbitrary component, $f_t$ differs uniformly from $g(t/n)$ by $O(n^{-1})$.

We prescribe $\Delta_j=d_j/n$ exactly. If instead one sets $f_t=g(t/n)$, the second difference at $\tau_j$ is $d_j(1-\{nq_j\})/n$. It then depends on the fractional part of $nq_j$ and can be arbitrarily small. The exact prescription avoids this accidental weakening of a slope change. All rates below use $|\Delta_j|\asymp n^{-1}$. The threshold is
\begin{equation}
\label{eq:threshold}
\lambda = \lambda_n = \Theta \, \sigma \, (8 \log n)^{1/2}, \qquad \Theta > 1 \text{ fixed}.
\end{equation}

The factor 8 is a convenient, non-sharp proof constant, not a finite-sample calibration. The proof applies a union bound, equivalent here to a Bonferroni bound, to at most $n^3$ triples $(s,e,b)$. Any fixed constant greater than 6 can replace 8 inside the square root and still make that probability tend to one. Model selection via the strengthened Schwarz Information Criterion (sSIC), introduced in Section \ref{sec:mms} and analyzed theoretically in S5 of the supplement, removes the need to select a threshold.

\subsection{Signal strength on a window}

On a window $(s,e]$, let $B_{s,e}$ be the valid candidate set defined in Section~\ref{sec:labs_algorithm}. The contrast vector \eqref{eq:contrast_vector} is the unit vector $\hat\psi_b=\psi_b/\|\psi_b\|$, where $\psi_b$ is the projection of the hinge $(t-b)_+$ onto the orthogonal complement of the constant and linear functions on the window. For a kink $\tau_j\in B_{s,e}$, set
\[
x_j=\min(\tau_j-s,e+1-\tau_j),
\]
its distance from the nearer edge, and set $x_j=0$ otherwise. Define
\[
\ell_j(s,e) \; = \; |\Delta_j| \, \|\psi_{\tau_j}\|
\]
as the \emph{signal strength} of the kink on that window, and set $\ell_j(s,e)=0$ when $\tau_j\notin B_{s,e}$. Then
$\max_{b\in B_{s,e}}|\langle f,\hat\psi_b\rangle|\leq\sum_j\ell_j(s,e)$, with equality when exactly one kink is present. Thus detectability is governed by the signal strength. Section~S2 of the supplement shows that there are absolute constants $0<c_\ell\leq C_\ell<\infty$ such that
\begin{equation}
\label{eq:levlaw}
c_\ell \, \frac{|d_j| \, x_j^{3/2}}{n} \;\leq\; \ell_j(s,e) \;\leq\; C_\ell \, \frac{|d_j| \, x_j^{3/2}}{n} .
\end{equation}
Up to universal constant factors, signal strength is therefore determined by the distance from the kink to the nearer edge, and grows as $x^{3/2}$; the exact norm depends on both distances, but only its order is controlled by $x$.

Write $\bar r_n = (n \log n)^{1/2}$. Two regimes follow from \eqref{eq:levlaw}: a kink within $O(\bar r_n)$ of a window edge has signal strength $O\{n^{-1/4}(\log n)^{3/4}\}=o(1)$, below the noise level $\sigma(8\log n)^{1/2}$, and is not detected. However, a kink at distance at least $\kappa n$ from both edges, for a constant $\kappa>0$, has signal strength at least $c_\ell\underline d\,\kappa^{3/2}n^{1/2}$, above that level, and is detected. The two regimes are separated by a factor of order $n^{3/4}(\log n)^{-3/4}$.

What has to be shown is that LABS produces no window in which a kink falls between the two regimes. The only window edge a recursive call creates is $\hat b$, and $\hat b$ is either within $O(\bar r_n)$ of a kink or, under Assumption~\ref{ass:nondeg} below, at distance of order $n$ from every kink.

\subsection{The population recursion}

The assumptions below are imposed on the finitely many windows the recursion visits, rather than on all intervals. These windows are described by a deterministic object depending on the signal alone. For $0\leq\alpha<\beta\leq1$, let $\Psi^{\alpha\beta}_v$ denote the continuum analogue of $\psi_b$ on $(\alpha,\beta)$. Let $V(\alpha,\beta)=\{j:\alpha<q_j<\beta\}$ be the set of \emph{live} kinks, meaning the true kinks still lying strictly inside the current population window. Define the noiseless profile and its leftmost maximizer
\begin{equation}
\label{eq:popprofile}
\bar D_{\alpha\beta}(v) = \Big\langle \sum_{j \in V(\alpha,\beta)} d_j \Psi^{\alpha\beta}_{q_j}, \ \Psi^{\alpha\beta}_v \big/ \|\Psi^{\alpha\beta}_v\| \Big\rangle, \qquad c(\alpha,\beta) = \min \arg\max_{v} \big| \bar D_{\alpha\beta}(v) \big| .
\end{equation}

\begin{defin}[Population LABS, POP]
\label{def:pop}
$\mathrm{POP}(\alpha,\beta)$ returns $\emptyset$ if $V(\alpha,\beta) = \emptyset$. Otherwise it sets $c = c(\alpha,\beta)$ and computes $\mathcal{L} = \mathrm{POP}(\alpha,c)$ and $\mathcal{R} = \mathrm{POP}(c,\beta)$; if both are empty it returns $\{c\}$, and otherwise, with $\alpha' = \max \mathcal{L}$ if $\mathcal{L} \neq \emptyset$ and $\alpha' = \alpha$ else, and $\beta'$ defined symmetrically, it returns $\mathcal{L} \cup \{ c(\alpha',\beta') \} \cup \mathcal{R}$ when $V(\alpha',\beta') \neq \emptyset$ and $\mathcal{L} \cup \mathcal{R}$ otherwise. Write $\mathcal{T}(g)$ for the set of windows at which $\mathrm{POP}$ is invoked, starting from $\mathrm{POP}(0,1)$.
\end{defin}

This is Algorithm~\ref{alg:labs} with ``the contrast exceeds $\lambda$'' criterion replaced by ``the window contains a kink in its interior''.

\begin{prop}
\label{prop:popexact}
If $\mathcal{T}(g)$ is finite and $c(\alpha,\beta)$ is well defined at each of its nodes, then $\mathrm{POP}(\alpha,\beta) = \{ q_j : j \in V(\alpha,\beta) \}$ at every node. In particular $\mathrm{POP}(0,1) = \{q_1, \ldots, q_N\}$.
\end{prop}

Proposition~\ref{prop:popexact}, proved in Section~S3 of the supplement, makes precise the argument of Section~\ref{sec:labs_algorithm}. If the split lands away from every kink, each child returns its own kinks, the re-test interval falls between two consecutive detected kinks and so is empty, and the spurious parent estimate is discarded. If instead the split lands on a kink, that kink is detected by neither child and is recovered by the re-test, on an interval where it is the only live kink.

\subsection{Assumptions and the consistency theorem}

\begin{assumption}[Termination]
\label{ass:term}
The population recursion started at $(0,1)$ terminates, that is, $\mathcal{T}(g)$ is finite.
\end{assumption}

\begin{assumption}[Non-degeneracy]
\label{ass:nondeg}
For every $(\alpha,\beta) \in \mathcal{T}(g)$ with $V(\alpha,\beta) \neq \emptyset$, the function $|\bar D_{\alpha\beta}|$ has a unique maximizer $c(\alpha,\beta)$ on $(\alpha,\beta)$, and $\partial_v^2 |\bar D_{\alpha\beta}| (c(\alpha,\beta)) < 0$.
\end{assumption}

Assumption~\ref{ass:nondeg} is the standard requirement that a population criterion have a well-separated optimum; without it the sample maximizer need not converge at any rate.
Assumption~\ref{ass:term} holds in particular whenever $c(\alpha,\beta)$ lies strictly between the extreme live kinks at every node. In that case the population recursion tree $\mathcal{T}(g)$ of Definition~\ref{def:pop}, whose nodes are the windows visited by $\mathrm{POP}$, has depth at most $N+1$.
Both assumptions concern $g$ alone, and neither refers to $n$, to the noise, or to any class of intervals. 

No condition separating $c(\alpha,\beta)$ from the change-points is assumed. Because $\mathcal{T}(g)$ is finite, the quantity $\min_l |c(\alpha,\beta) - q_l|$, which at each node is either zero or strictly positive, has a strictly positive minimum over the nodes at which it does not vanish.

\begin{theorem}
\label{thm:labs_consistency}
Let $X_1, \ldots, X_n$ follow \eqref{eq:model} in the regime described above, and suppose Assumptions~\ref{ass:term} and~\ref{ass:nondeg} hold for $g$. Let $\hat{\mathcal{T}}_n$ be the set returned by Algorithm~\ref{alg:labs} with the contrast \eqref{eq:lin_contrast} and threshold \eqref{eq:threshold}, written in increasing order as $\hat\tau_{(1)} < \cdots < \hat\tau_{(\hat N)}$, where $\hat N = |\hat{\mathcal{T}}_n|$, and let $\tau_1 < \cdots < \tau_N$ be the true change-points. Then there is a constant $C < \infty$, depending only on $g$ and $\sigma$, such that as $n \to \infty$:
\begin{itemize}
\item if $N \geq 1$:
\[
\mathbb{P}\!\left( \hat{N} = N \;\;\text{and}\;\; \max_{j=1,\ldots,N} |\hat{\tau}_{(j)} - \tau_j| \leq C \, (n \log n)^{1/2} \right) \to 1;
\]
\item if $N = 0$: $\mathbb{P}(\hat N = 0) \to 1$.
\end{itemize}
\end{theorem}

The proof, given in Section~S3 of the supplement, shows that the sample recursion follows the population recursion: on an event of probability tending to one, every window LABS visits has both endpoints within $O(\bar r_n)$ of $n$ times the endpoints of a node of $\mathcal{T}(g)$, and the call on it returns exactly the kinks live at that node, each localized to within $O(\bar r_n)$.

\begin{remark}[Rate]
\label{rem:rate}
{\em
The rate $(n\log n)^{1/2}$ agrees with the refined single-change-point rate of \cite{bcf16}, obtained in their Section~3.4 under additional regularity and a more restrictive threshold choice, rather than the $O_p\{n^{2/3}(\log n)^{1/3}\}$ of their general theorem. Thus LABS attains for $N$ change-points the rate that the contrast attains for one. A coarse comparison of the deterministic gap, meaning the loss in the noiseless contrast away from its population maximum, with a global noise bound gives only $n^{3/4}(\log n)^{1/4}$. At that distance a kink can remain strong enough for a child recursion to detect it again. The proof instead controls noise increments between nearby contrast vectors and obtains the sharper $(n\log n)^{1/2}$ rate. At this distance the kink's signal strength on the child window is $o(1)$, so the duplicate detection does not occur.
}
\end{remark}

\begin{remark}[Scope]
\label{rem:scope}
{\em
Theorem~\ref{thm:labs_consistency} covers Algorithm~\ref{alg:labs}, equivalently LABS-Grid with $M=R=2$. Section~S4 of the supplement proves the same conclusion for arbitrary fixed $M,R\geq2$, under conditions of the same form imposed on the $M$-grid proposal recursion. As explained there, the re-test contains at most one live kink, so no additional node-level condition is required for $R$. Section~S5 proves consistency when the model is chosen from a solution path by the strengthened Schwarz criterion, as in Section~\ref{sec:simulations}.
}
\end{remark}

\section{Numerical illustrations}
\label{sec:simulations}

\subsection{Simulation design}

We compare sSIC-selected LABS with five existing procedures on twelve
piecewise-linear scenarios.  The scenarios comprise two null signals and five
non-null signal families, each of the latter considered at two noise levels.
Table~\ref{tab:simulation_design} summarizes the design and
Figure~\ref{fig:signals} shows one realization of each.  The trapezoid is the
noisy analogue of the trap signal in Section~\ref{sec:motivation}.  The bump
contains three closely spaced changes, the teeth signal contains nineteen
regularly spaced changes, and the irregular signal combines unequal segment
lengths with heterogeneous slopes.  Each noise-free signal is constructed by
cumulatively summing its segment-wise slopes, so its change-points occur
exactly at the cumulative segment lengths in Table~\ref{tab:simulation_design}.

The two null scenarios differ in the presence of a global linear trend and in
the noise scale, which is $1$ in 1a and $200$ in 1b.  Every method considered
here uses a statistic that is invariant to the addition of a linear function of
the index and equivariant under rescaling, and the threshold is scaled by
$\hat\sigma$, so scenario 1b serves as a check on those two properties rather
than as an independent scenario; the differences between the two null columns
below are Monte Carlo variation.

For all methods except CPOP we use 500 Monte Carlo replications per scenario.
CPOP is substantially slower, so it is run for 50 replications per scenario.
It is applied to the first 50 realizations used by the other methods, making
those comparisons paired, but its estimates have appreciably larger Monte
Carlo uncertainty.

\begin{table}[h]
\centering
\small
\resizebox{\textwidth}{!}{%
\begin{tabular}{llrrl@{\hspace{3em}}l@{\hspace{3em}}l}
\hline
Scenario & Labels & $n$ & $N$ & $\sigma$ & Segment slopes & Segment lengths \\
\hline
Null, flat       & 1a    & 500  & 0  & 1       & $(0)$                         & $(500)$ \\
Null, linear     & 1b    & 500  & 0  & 200     & $(1)$                         & $(500)$ \\
Single kink      & 2a/2b & 400  & 1  & 50/100  & $(0,1)$                       & $(200,200)$ \\
Trapezoid        & 3a/3b & 600  & 2  & 50/100  & $(1,0,-1)$                    & $(200,200,200)$ \\
Short bump       & 4a/4b & 420  & 3  & 2/3     & $(0,1,-1,0)$                  & $(200,10,10,200)$ \\
Teeth            & 5a/5b & 2000 & 19 & 50/70   & $\mathrm{rep}\{(1,-1),10\}$   & $\mathrm{rep}\{100,20\}$ \\
Irregular slopes & 6a/6b & 1000 & 5  & 200/250 & $(2,-1,3,-2.5,0.5,-1.5)$     & $(100,300,100,200,150,150)$ \\
\hline
\end{tabular}%
}
\caption{Simulation scenarios.  The noise-free signal is
$f_t=\sum_{s=1}^t b_s$, where $b_s$ takes the listed segment slopes for the
listed numbers of observations.  The change-points are therefore the
cumulative segment lengths, excluding $n$.  The notation
$\mathrm{rep}\{x,k\}$ means that $x$ is repeated $k$ times.  The noise is
independent Gaussian with standard deviation $\sigma$.  Each non-null row
represents two scenarios, with the lower noise level labeled ``a'' and the
higher level labeled ``b''.}
\label{tab:simulation_design}
\end{table}

\begin{figure}[htbp]
\centering
\includegraphics[width=\textwidth,height=0.7\textheight,keepaspectratio]{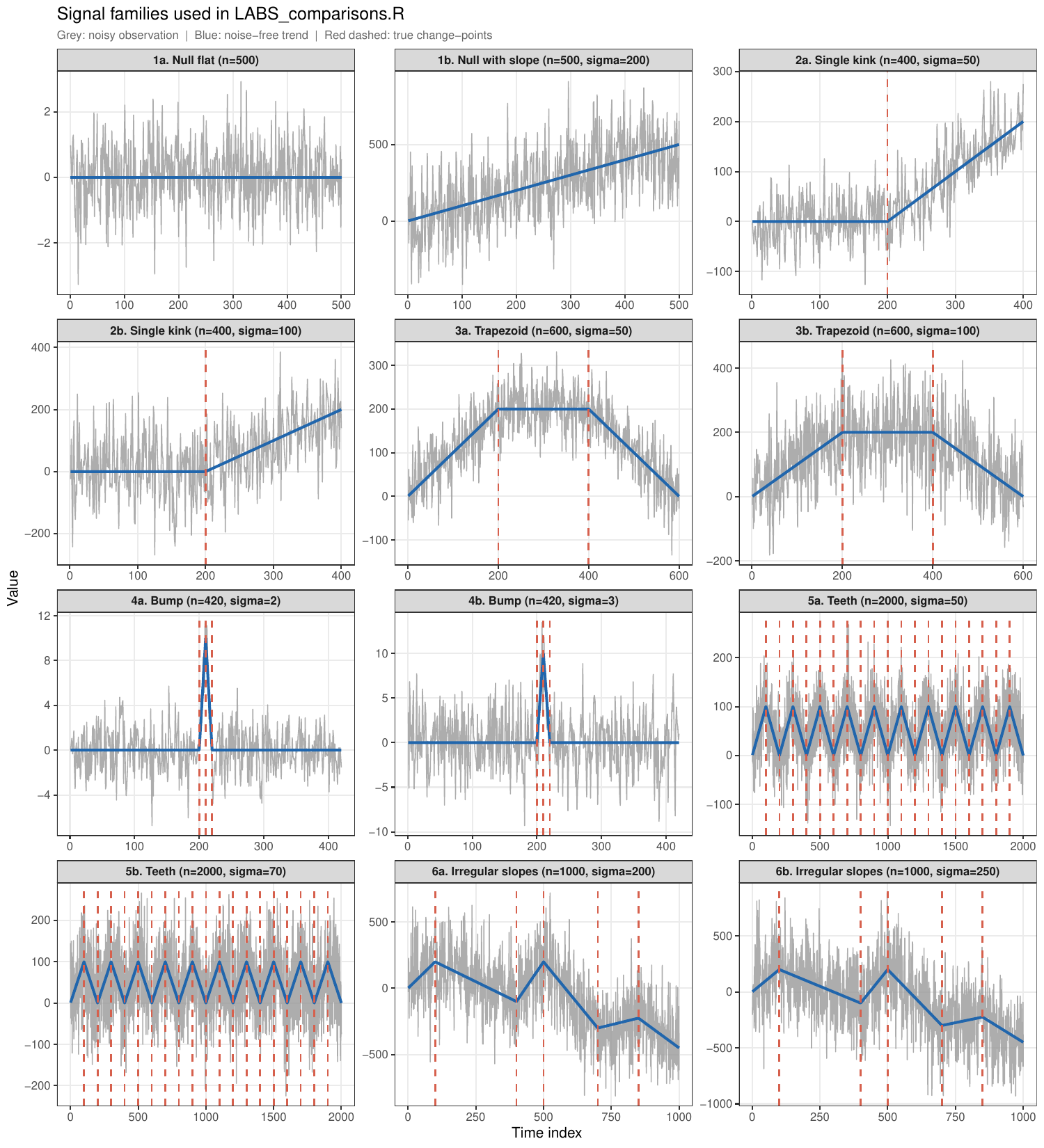}
\caption{One realization of each of the twelve scenarios of
Table~\ref{tab:simulation_design}.  Grey: observed data.  Blue: noise-free
signal.  Red dashed: true change-points.  Vertical scales differ between
panels.}
\label{fig:signals}
\end{figure}

\subsection{Methods and model selection}
\label{sec:mms}

For the piecewise-linear contrast \eqref{eq:lin_contrast}, define the robust
noise estimator
\[
    \hat{\sigma} =
    \mathrm{MAD}\left(\frac{\Delta^2 X}{\sqrt{6}}\right),
\]
where $\Delta^2 X_t=X_{t+1}-2X_t+X_{t-1}$ and the MAD includes its usual
Gaussian consistency factor.  Each SIC-LABS solution path is computed at
\[
    \lambda_a = a\hat{\sigma}\sqrt{2\log n},
    \qquad a=0.50,0.55,\ldots,1.50.
\]
For every distinct change-point set $\mathcal{T}$ on this path, we fit a
continuous piecewise-linear signal and minimize
\[
    \operatorname{sSIC}(\mathcal{T})
    =n\log\{\operatorname{RSS}(\mathcal{T})/n\}
     +(2|\mathcal{T}|+3)(\log n)^{1.01},
\]
where $\operatorname{RSS}(\mathcal{T})$ is the residual sum of squares from the
least-squares continuous piecewise-linear fit with change-point set
$\mathcal{T}$.  This is the strengthened Schwarz information criterion with
exponent $1.01$.  We use ``SIC'', ``sSIC'' and ``strengthened SIC'' for this
criterion below; the shorter SIC form is retained in method labels.  The same
continuous piecewise-linear convention is used for NOT and IDetect in the
\texttt{breakfast} R package \citep{breakfast}.  Section~S5 of the supplement shows that the selected model is
consistent, at the rate of Theorem~\ref{thm:labs_consistency}, provided the
grid of multipliers reaches into the range of thresholds admitted by that
theorem; the exponent being strictly greater than one is what separates the
correct model from models with superfluous change-points.  
Section~S5 also treats a variant in which the path is generated over the whole
range of thresholds, indexed by the number of change-points and truncated at a
cap.

We use sSIC selection because preliminary experiments, not reported here,
showed that fixed thresholds either lost power or did not control false
positives uniformly across the signal scales.  This also agrees with the usual
role of information criteria in change-point analysis: the threshold produces a
solution path and the criterion selects its model size.

We consider five LABS configurations:
\[
\begin{gathered}
\mathrm{LABS(SIC,}M=2),\quad
\mathrm{LABS(SIC,}M=3),\quad
\mathrm{LABS(SIC,}M=3,R=2),\\
\mathrm{LABS(SIC,}M=5),\quad
\mathrm{LABS(SIC,}M=5,R=2).
\end{gathered}
\]
Here $M$ is the proposal grid size and $R$ is the grid size used at the
look-ahead re-test; when $R$ is omitted, $R=M$.  Every configuration applies
the threshold in the proposal phase and again at the re-test.

The external competitors are IDetect \citep{af22} with SIC model selection,
NOT \citep{bcf16} with SIC model selection, TrendSegment \citep{mf23} in its
continuous and independent-noise setting, multiscale MOSUM \citep{koc24} with
BIC bandwidth merging, and CPOP \citep{mfl19} with penalty
$2\log n$.  For CPOP the noise standard deviation is estimated by the same
second-difference MAD used above.  MOSUM is run using the implementation of
\cite{koc24} with the tuning those authors recommend: $\eta=0.3$,
$\theta=0.8$, significance level $0.05$, and BIC bandwidth merging.  Its
bandwidths follow a Fibonacci sequence
$G_1,2G_1,3G_1,5G_1,\ldots$, truncated before $n/\log_{10}n$, where
$G_1=\max\{10,10\lceil n/1000\rceil\}$. The numerical work used R 4.5.2 \citep{Rsoftware}. The competitor package versions were 1.0.8 for \texttt{cpop} \citep{cpopSoftware}, 2.6 for \texttt{breakfast} \citep{breakfast} (IDetect and NOT), and 1.3.2 for \texttt{trendsegmentR} \citep{trendsegmentSoftware}. The MOSUM implementation was obtained from the authors' \texttt{MovingSumLin} repository \citep{movingSumLinSoftware}.

Two features of the MOSUM comparison should be borne in mind.  First, MOSUM
tests jointly for a jump and a slope change at each candidate location, and has
no variant that imposes continuity, so unlike the other methods considered here
it does not use the knowledge that the signal is continuous.  Second, its
theory requires the change-points to be separated by more than twice the
bandwidth.  For the bump signal, whose changes are ten observations apart, no
bandwidth in the recommended set satisfies this, so that scenario lies outside
the conditions under which the method is designed to operate; for the teeth and
irregular signals only the smallest bandwidths satisfy it.  We retain the full
recommended bandwidth set in all cases, as the authors do in their own study,
since the BIC merging step demotes poorly fitting bandwidths.

For each replication we record the absolute error in the estimated number of
change-points, the directed location errors in both directions, their maximum
(the Hausdorff distance), and computation time.  The principal
measure below is the probability of estimating the exact number of
change-points.  On a null signal its complement is the false-positive rate
(FPR).

\subsection{Results}

Table~\ref{tab:simulation_summary} gives null control and performance averaged
over the ten non-null scenarios.  All methods control the FPR below 0.05.  Each
SIC-LABS variant has FPR 0.008 on both null signals, and SIC-NOT has maximum
null FPR 0.008.  Under the specified bandwidth sequence, MOSUM also controls
the null, with FPR 0.002 on the flat signal and 0.006 on the linear signal.
CPOP has no false positive in its 50 flat-null replications and one in its 50
linear-null replications, giving FPR 0.020 on the latter.

\begin{table}[h]
\centering
\small
\resizebox{\textwidth}{!}{%
\begin{tabular}{lrrrrrrrr}
\hline
Method & Reps & Flat FPR & Linear FPR & Exact count & Count error & Hausdorff & Time (ms) & CPOP/time \\
\hline
CPOP                   &  50 & 0.000 & 0.020 & 0.844 & 0.202 &  46.10 & 1142.25 &   1.0 \\
LABS(SIC,$M=5,R=2$)   & 500 & 0.008 & 0.008 & 0.827 & 0.351 &  46.15 &   28.50 &  40.1 \\
LABS(SIC,$M=5$)       & 500 & 0.008 & 0.008 & 0.827 & 0.350 &  45.79 &   33.00 &  34.6 \\
LABS(SIC,$M=3$)       & 500 & 0.008 & 0.008 & 0.792 & 0.519 &  58.24 &   11.00 & 103.8 \\
LABS(SIC,$M=3,R=2$)   & 500 & 0.008 & 0.008 & 0.791 & 0.518 &  59.37 &   10.50 & 108.8 \\
NOT(SIC)                 & 500 & 0.008 & 0.002 & 0.726 & 0.921 &  57.19 &  108.50 &  10.5 \\
LABS(SIC,$M=2$)         & 500 & 0.008 & 0.008 & 0.704 & 1.386 & 104.35 &    4.50 & 253.8 \\
IDetect                  & 500 & 0.002 & 0.000 & 0.662 & 1.068 &  93.78 &   23.00 &  49.7 \\
TrendSegment             & 500 & 0.000 & 0.000 & 0.355 & 2.951 & 211.85 &  214.00 &   5.3 \\
MOSUM                    & 500 & 0.002 & 0.006 & 0.236 & 1.953 & 117.33 &    3.00 & 380.8 \\
\hline
\end{tabular}%
}
\caption{Null false-positive rates, averages over the ten non-null scenarios,
and computation times, ranked by mean exact-count probability. ``Reps'' is the
number of replications per scenario, ``Exact count'' is the mean probability of
estimating the correct number of change-points, ``Count error'' is the mean
absolute count error, and ``Hausdorff'' is the mean finite Hausdorff distance.
``Time'' is the median of the ten scenario-specific median wall-clock times, so
that each signal receives equal weight, and ``CPOP/time'' is the CPOP median
divided by the method median.  The common LABS solution-path preparation time
is apportioned equally among the five LABS variants.  Timings compare the
implementations used in the simulation and should not be read as
hardware-independent complexity measurements.  The CPOP estimates have greater
Monte Carlo uncertainty because CPOP uses only 50 replications per scenario,
compared with 500 for every other method.}
\label{tab:simulation_summary}
\end{table}

CPOP has the highest mean exact-count probability, 0.844, but this estimate is
based on only 50 replications per scenario.  Among the methods run for 500
replications, the $M=5,R=2$ SIC-LABS version is highest, while being around 40 times faster than CPOP (more on execution times below). Its mean exact-count
probability 0.827, followed almost exactly by $M=5$, also 0.827 after rounding.
The $M=3$ versions attain 0.791 to 0.792, followed by SIC-NOT (0.726), $M=2$
LABS (0.704), and IDetect (0.662).  The $R=2$ re-test therefore retains the
accuracy of the full $M=5$ re-test while reducing computation.  The gap between
$M=2$ and $M=5$ is substantial, and is largest on the teeth and bump signals,
where the noise is frequently large enough for both child recursions to return
empty sets; as noted in Section~\ref{sec:labs_algorithm}, the look-ahead does
not correct the proposal in that configuration, and the grid search is what
provides the localization instead.

CPOP requires 1.142 seconds for a typical replication, compared with 28.5
milliseconds for the $M=5,R=2$ LABS variant, a factor of about 40; the
unrestricted $M=5$ re-test is about 35 times faster than CPOP.  The difference
persists on the largest and most structurally complex signals: the maximum
scenario median is 7.016 seconds for CPOP against 0.134 seconds for the
$M=5,R=2$ LABS variant.  NOT and TrendSegment are also slower than the
recommended LABS version.  MOSUM is the fastest method and controls the null
false-positive rate, but its mean exact-count probability is substantially
lower.

Table~\ref{tab:simulation_exact_count} reports the exact-count probability
scenario by scenario.  The $M=5$ variants are particularly effective on the
dense teeth signals: their exact-count probabilities are 0.994 at
$\sigma=50$ and approximately 0.75 at $\sigma=70$.  The $M=5,R=2$ variant also
leads the methods run for 500 replications on both irregular-slope signals.
The unrestricted $M=5$ version is highest on the bump at $\sigma=3$, while
the smaller $M=2$ grid is preferable on the high-noise trapezoid, so no single
grid size dominates in every geometry.

\begin{table}[h]
\centering
\scriptsize
\resizebox{\textwidth}{!}{%
\begin{tabular}{lrrrrrrrrrr}
\hline
Method & 2a & 2b & 3a & 3b & 4a & 4b & 5a & 5b & 6a & 6b \\
\hline
LABS(SIC,$M=2$)       & .996 & .982 & .984 & .744 & .952 & .640 & .762 & .102 & .576 & .302 \\
LABS(SIC,$M=3$)       & .992 & .982 & .984 & .734 & .978 & .888 & .972 & .532 & .574 & .282 \\
LABS(SIC,$M=3,R=2$)   & .992 & .982 & .984 & .732 & .978 & .884 & .972 & .534 & .574 & .282 \\
LABS(SIC,$M=5$)       & .994 & .982 & .982 & .714 & .984 & .894 & .994 & .746 & .646 & .330 \\
LABS(SIC,$M=5,R=2$)   & .994 & .982 & .982 & .714 & .982 & .888 & .994 & .748 & .652 & .332 \\
IDetect                  & .996 & .956 & .984 & .696 & .964 & .514 & .916 & .212 & .284 & .096 \\
NOT(SIC)                 & .992 & .976 & .994 & .728 & .984 & .732 & .960 & .354 & .368 & .174 \\
TrendSegment             & .998 & .654 & .994 & .302 & .518 & .076 & .004 & .000 & .008 & .000 \\
MOSUM                    & .916 & .042 & .854 & .022 & .120 & .014 & .392 & .002 & .000 & .000 \\
CPOP                     & .980 & .940 & .960 & .780 & .980 & .980 & .980 & .900 & .640 & .300 \\
\hline
\end{tabular}%
}
\caption{Probability of estimating the exact number of change-points in each
non-null scenario.  Scenario labels are defined in
Table~\ref{tab:simulation_design}.  CPOP uses 50 replications per scenario; all
other methods use 500.}
\label{tab:simulation_exact_count}
\end{table}
\FloatBarrier

Overall, SIC model selection gives LABS good power without loss of null
control.  Among the LABS variants, the $M=5,R=2$ version has the highest
aggregate exact-count performance of the methods run for 500 replications, is
faster than $M=5,R=5$, and improves on SIC-NOT and IDetect in the dense- and
irregular-change settings.  CPOP performs well but is much slower, and its
50-replication estimates remain less precise than the 500-replication
comparisons for the other methods.

\section{Data application -- global temperature anomalies}
\label{sec:data}

We analyze the HadCRUT5 annual land--sea surface-temperature anomalies made
available by Our World in Data \citep{owidtemp}.  The series measure differences in degrees
Celsius from the 1861--1890 mean.  We consider the World, Northern Hemisphere
and Southern Hemisphere series separately.  Although the downloaded data run
from 1850 to 2026, the 2026 observation represents an incomplete year, so the
primary analysis uses the 176 annual observations from 1850 to 2025.  The grey
bands in the top row of Figure~\ref{fig:temperature_comparison} are the
pointwise 95\% uncertainty intervals supplied with the data; they are displayed
for context but are not used as observation-specific weights.

We use the best-performing LABS version among the methods run for 500
replications in Section~\ref{sec:simulations},
$\mathrm{LABS(SIC,}M=5,R=2)$, and all five external competitors from that
study: IDetect, NOT(SIC), TrendSegment, MOSUM and CPOP.  The implementations
and tuning are unchanged from the simulation study.  A reported change year
is the first year of the segment following the fitted slope change.  

\begin{table}[h]
\centering
\small
\setlength{\tabcolsep}{4pt}
\begin{tabular}{p{1.20in}p{1.20in}p{1.66in}p{1.66in}}
\hline
Method & World & Northern Hemisphere & Southern Hemisphere \\
\hline
LABS(SIC,\newline $M=5,R=2$) & 1913, 1942, 1972 & 1863, 1877, 1884, 1918, 1941, 1975 & 1881, 1910, 1942, 1965 \\
IDetect & 1914, 1945, 1977 & 1919, 1945, 1977 & 1912, 1967 \\
NOT(SIC) & 1910, 1943, 1972 & 1914, 1941, 1975 & 1880, 1912, 1945, 1965 \\
TrendSegment & 1974 & 1937, 1977 & 1912 \\
MOSUM & 1861, 1911, 1969 & 1909, 1946, 1974 & 1912, 1940, 1969 \\
CPOP & 1876, 1877, 1884, 1911, 1942, 1968, 2012 & 1862, 1876, 1878, 1879, 1912, 1941, 1974 & 1876, 1877, 1885, 1911, 1914, 1933, 1945, 1946, 1964, 2022 \\
\hline
\end{tabular}
\caption{First years following the slope changes detected in the three
temperature-anomaly series.  All methods use the configurations from the
simulation study.}
\label{tab:temperature_changes}
\end{table}

For the World series, LABS detects changes in 1913, 1942 and 1972
(right-hand column of Figure~\ref{fig:temperature_comparison}).  Its fitted slopes over 1850--1912,
1913--1941, 1942--1971 and 1972--2025 are, respectively, $-0.020$, $0.142$,
$-0.039$ and $0.205$ degrees Celsius per decade.  IDetect and NOT place each
of these three principal transitions within five years of LABS, while MOSUM
places the latter two in 1911 and 1969 and additionally reports an early
change in 1861.  TrendSegment returns only the onset of the most recent
warming trend, in 1974.  CPOP gives a much finer segmentation, including
adjacent changes in the late nineteenth century and an additional recent
change in 2012.

For the Northern Hemisphere, LABS detects six changes: 1863, 1877, 1884,
1918, 1941 and 1975 (left-hand column of
Figure~\ref{fig:temperature_comparison}).  The three short,
early segments have fitted slopes of $-0.247$, $0.206$ and $-0.258$ degrees
Celsius per decade; from 1884 to 1917 the fitted trend is nearly flat.  The
more persistent twentieth-century segments have slopes $0.220$ over
1918--1940, $-0.081$ over 1941--1974, and $0.287$ over 1975--2025.  The other
methods concentrate on the same broad transitions in the 1910s, 1940s and
1970s: IDetect, NOT and MOSUM each return one change in each period, and
TrendSegment returns changes in 1937 and 1977.  The additional early LABS
changes are supported in broad location by CPOP, although CPOP again resolves
them into several adjacent changes.

For the Southern Hemisphere, LABS detects changes in 1881, 1910, 1942 and
1965 (middle column of Figure~\ref{fig:temperature_comparison}).  Its corresponding segment slopes
are $0.027$, $-0.076$, $0.114$, $-0.010$ and $0.133$ degrees Celsius per
decade.  NOT gives a particularly close result, with changes in 1880, 1912,
1945 and 1965.  IDetect retains the 1912 and 1967 changes, MOSUM returns 1912,
1940 and 1969, and TrendSegment retains only 1912.  CPOP places changes near
these common locations but also returns several adjacent changes and one near
the end of the series, for a total of ten.

\begin{figure}[p]
\centering
\includegraphics[width=\textwidth,height=0.78\textheight,keepaspectratio]{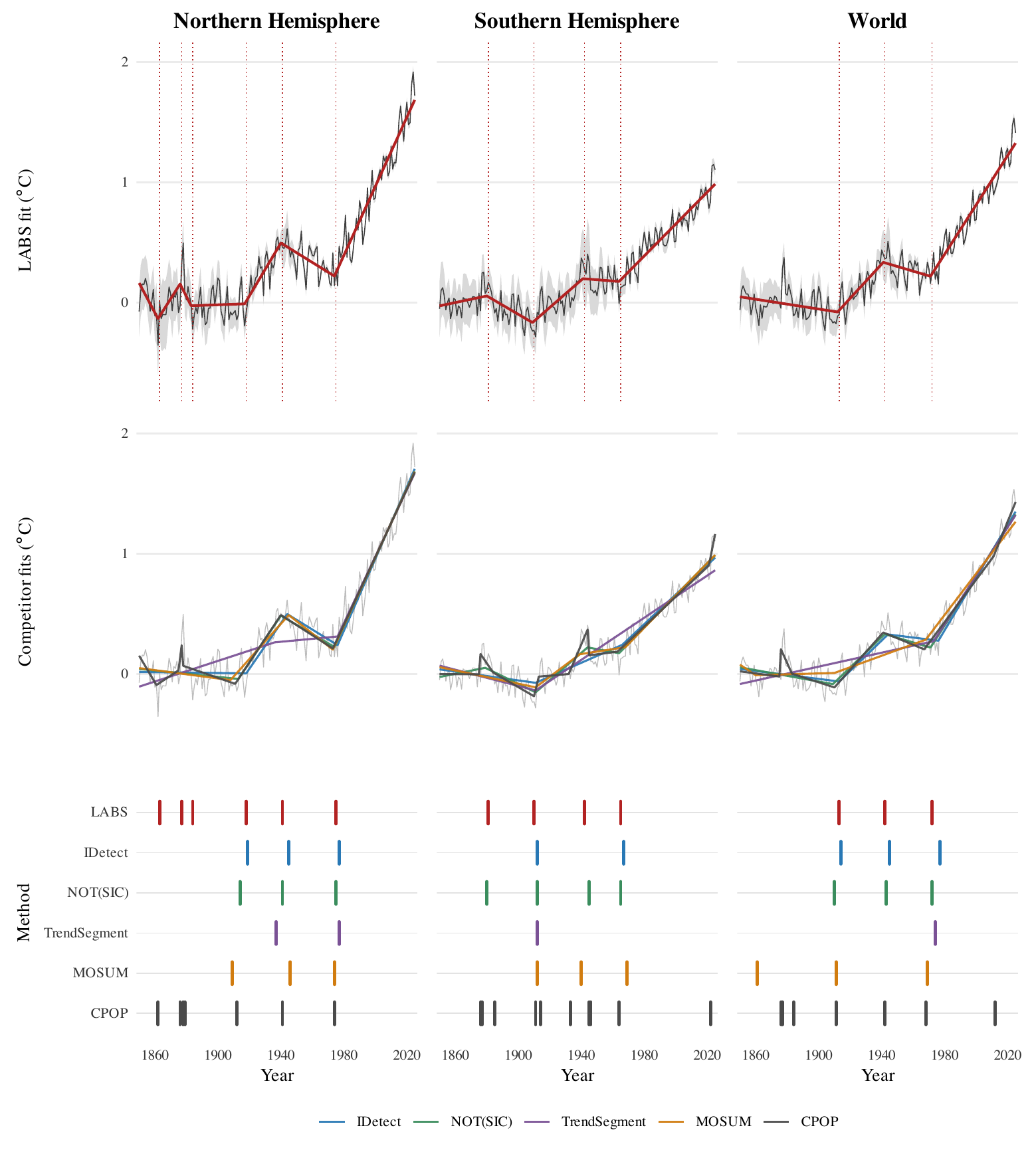}
\caption{Temperature anomalies and detected slope changes, 1850--2025.  The
columns show, from left to right, the Northern Hemisphere, Southern Hemisphere
and World series, all on common vertical and horizontal scales.  The top row
shows the annual anomaly, its supplied 95\% uncertainty interval, the LABS fit,
and dotted lines at the LABS changes.  The middle row superimposes the five
competitor fits; for comparability, each is the continuous piecewise-linear
least-squares refit at that method's estimated change-points.  The bottom row
shows the change-point locations for all six methods.  CPOP's automatic $0$
and $n-1$ boundaries are omitted.}
\label{fig:temperature_comparison}
\end{figure}

Taken together, the three analyses give a common chronology: early
twentieth-century warming, a mid-century flattening or reversal, and renewed
warming in the later twentieth century.  The principal differences are in
timing and magnitude.  LABS places the onset of the latest sustained warming
in 1965 in the Southern Hemisphere, seven years before the World change and
ten years before the Northern Hemisphere change.  The subsequent Northern
Hemisphere slope, $0.287$ degrees Celsius per decade, is more than twice the
Southern Hemisphere slope of $0.133$, with the World slope intermediate at
$0.205$.  The Northern series also has more short-lived nineteenth-century
features, on which the methods agree less strongly.  Across regions, LABS,
IDetect, NOT and MOSUM broadly agree on the major twentieth-century changes;
TrendSegment favors fewer changes, whereas CPOP systematically favors a
finer segmentation.  Including the incomplete 2026 value in a sensitivity
analysis leaves this principal twentieth-century comparison unchanged,
although it affects some fine early Northern Hemisphere changes and CPOP's
finer segmentation.

\section*{Declarations}
\textbf{Funding.} No funding was received. \textbf{Competing interests.} None. \textbf{Generative AI use.} OpenAI Codex (gpt-5.6-sol) and Anthropic's Opus 5 assisted with drafting and checking prose and mathematical proofs, R implementation, and numerical studies, for clarity and verification. The author reviewed, revised, and independently checked all AI output and takes full responsibility. \textbf{Data and code.} HadCRUT5 data are from Our World in Data \citep{owidtemp}; code, frozen data files, and simulation output are available at \url{https://github.com/pfryz/LABS}.

\clearpage
\setcounter{section}{0}
\setcounter{equation}{0}
\setcounter{algorithm}{0}
\renewcommand{\thesection}{S\arabic{section}}
\renewcommand{\theequation}{S\arabic{equation}}
\renewcommand{\thealgorithm}{S\arabic{algorithm}}
\renewcommand{\theHsection}{S.\arabic{section}}
\renewcommand{\theHequation}{S.\arabic{equation}}
\renewcommand{\theHalgorithm}{S.\arabic{algorithm}}

\title{Supplement to ``LABS: Extending the scope of binary segmentation via a look-ahead device''} 
\author{Piotr Fryzlewicz\thanks{Department of Statistics, London School of Economics and Political Science, London WC2A 2AE, UK. Correspondence: \texttt{p.fryzlewicz@lse.ac.uk}. ORCID: 0000-0002-9676-902X.}}
\date{}
\maketitle

\begin{abstract}
This document supplements the main paper. Section~\ref{sup:grid} states the
LABS-Grid algorithm in full. Section~\ref{sup:toolkit} establishes the
properties of the contrast used in the consistency proof.
Section~\ref{sup:proof} proves \proppop{} and \mainthm{}.
Section~\ref{sup:gridtheory} extends the consistency result to LABS-Grid.
Section~\ref{sup:ssic} treats selection from a solution path by a
strengthened Schwarz criterion. Section~\ref{sup:lookahead} describes related
uses of look-ahead in other algorithmic settings.
\vspace{5pt}

\noindent {\bf Keywords:} Binary segmentation, change-point detection, piecewise-linear signal, recursive algorithm, trend estimation.
\end{abstract}

Numbered items of the form ``Assumption 4.1'', ``Theorem 4.1'' and ``Algorithm 2'' refer to the main paper; items numbered with an S prefix refer to this document. Displayed equations (4) and (5) of the main paper are the contrast vector and the contrast function.
\section{The LABS-Grid algorithm}
\label{sup:grid}

Section 3.2 of the main paper describes LABS-Grid. It evaluates the proposal contrast on all sub-intervals formed by an $M$-point equally spaced grid, and the re-test contrast on those formed by an $R$-point grid. For completeness, the algorithm is stated in full below. Here $\textsc{GridSearch}(X,s,e,K)$ returns the pair $(\hat b,W)$, where $\hat b$ is the location giving the maximum contrast across all sub-intervals formed by the $K$-point grid and $W$ is that maximum. Setting $M=R=2$ recovers \algmain{} of the main paper, which is the algorithm analyzed in \mainthm{}.

\begin{algorithm}[H]
{\footnotesize
\caption{LABS with grid-based contrast evaluation (LABS-Grid)}
\label{alg:labs_grid}
\begin{algorithmic}[1]
\Function{LABS-Grid}{$X, s, e, \lambda, M, R$}
    \If{the $M$-grid search cannot be evaluated on $(s,e]$ (interval too short)}
        \State \Return $\emptyset$
    \EndIf
    \State $(\hat{b}, W) \gets \textsc{GridSearch}(X, s, e, M)$ \Comment{Proposal search on $(s, e]$}
    \If{$W > \lambda$} \Comment{Proposal phase}
        \State $\mathcal{L} \gets$ \Call{LABS-Grid}{$X, s, \hat{b}, \lambda, M, R$}
        \State $\mathcal{R} \gets$ \Call{LABS-Grid}{$X, \hat{b}, e, \lambda, M, R$}
        \State \Comment{Look-ahead refinement phase}
        \If{$\mathcal{L} \neq \emptyset$ \textbf{or} $\mathcal{R} \neq \emptyset$}
            \State $s' \gets \begin{cases} \max(\mathcal{L}) & \text{if } \mathcal{L} \neq \emptyset \\ s & \text{otherwise} \end{cases}$
            \State $e' \gets \begin{cases} \min(\mathcal{R}) & \text{if } \mathcal{R} \neq \emptyset \\ e & \text{otherwise} \end{cases}$
            \If{the $R$-grid search can be evaluated on $(s',e']$}
                \State $(\hat{b}', W') \gets \textsc{GridSearch}(X, s', e', R)$ \Comment{Grid re-test on $(s', e']$}
                \If{$W' > \lambda$}
                    \State \Return $\mathcal{L} \cup \{\hat{b}'\} \cup \mathcal{R}$ \Comment{Accept refined estimate}
                \EndIf
            \EndIf
            \State \Return $\mathcal{L} \cup \mathcal{R}$ \Comment{Re-test failed: discard parent estimate}
        \EndIf
        \State \Return $\mathcal{L} \cup \{\hat{b}\} \cup \mathcal{R}$ \Comment{Both children empty: keep parent estimate}
    \Else
        \State \Return $\emptyset$
    \EndIf
\EndFunction
\end{algorithmic}
}
\end{algorithm}

The proposal and re-test searches cost $O\{M^2(e-s)\}$ and $O\{R^2(e'-s')\}$, respectively. Under balanced recursive splits the typical total cost is $O\{(M^2+R^2)n\log n\}$, and the worst-case cost is $O\{(M^2+R^2)n^2\}$. Ordinary binary segmentation also has worst-case cost $O(n^2)$ when successive splits are maximally unbalanced. In the notation of the main paper, an omitted $R$ means $R=M$.

\section{Auxiliary properties of the contrast}
\label{sup:toolkit}

This section collects the four properties of the piecewise-linear GLR contrast (5) of the main paper that are used in the proof. Throughout, $(s,e]$ is a window, $W = \{s+1,\ldots,e\}$,
$m = e-s$, and $B_{s,e} = \{s+2,\ldots,e-1\}$; $\Pi^\perp_W$ is orthogonal
projection in $\mathbb{R}^W$ onto the orthogonal complement of
$\mathrm{span}\{\mathbf 1, t\}$; and
\[
\psi_\theta := \Pi^\perp_W\big[(t-\theta)_+\big]_{t \in W},
\qquad
\hat\psi_\theta := \psi_\theta / \|\psi_\theta\| .
\]
The contrast vector $\phi^b_{(s,e]}$ of the main paper equals $-\hat\psi_b$, extended by zero off $W$; with the opposite hinge convention it would equal $+\hat\psi_b$. Either way it is the unit vector in $\mathbb{R}^W$ orthogonal to $\mathbf 1$ and $t$ that is continuous and piecewise linear with its only knot at $b$, which is the defining property of the closed forms (4) and (5) there, and the sign is immaterial because the contrast takes an absolute value: $\mathcal{C}_{s,e}(b) = |\langle X, \hat\psi_b\rangle|$.

For a kink $\tau_j$, define its distance from the nearer end of $W$ by
\[
x_j(s,e):=\min(\tau_j-s,e+1-\tau_j)
\quad\text{if }\tau_j\in B_{s,e},
\]
and set $x_j(s,e):=0$ otherwise. Its signal strength on $(s,e]$ is
$\ell_j(s,e):=|\Delta_j|\,\|\psi_{\tau_j}\|$ when
$\tau_j\in B_{s,e}$, and zero otherwise. As in Section~4.1 of the main paper,
$\underline d:=\min_j|d_j|$ and $\bar d:=\max_j|d_j|$ when $N\ge1$.

\subsection{Which kinks the window sees}

\begin{lemma}[Representation]
\label{lem:rep}
Let $V(s,e) := \{ j : s+2 \le \tau_j \le e-1 \}$. Then:
\begin{enumerate}
\item[\emph{(a)}] there are $a,c \in \mathbb{R}$ with
$f_t = a + ct + \sum_{j \in V(s,e)} \Delta_j (t-\tau_j)_+$ for $t \in W$, and
hence
$\langle f, \hat\psi_b\rangle = \sum_{j \in V(s,e)} \Delta_j \langle
\psi_{\tau_j}, \hat\psi_b\rangle$ for every $b \in B_{s,e}$;
\item[\emph{(b)}] $\max_{b} |\langle f, \hat\psi_b\rangle| \le \sum_j
\ell_j(s,e)$, and for any $U$ and any $b, b' \in B_{s,e}$,
\[
\Big| \sum_{j \notin U} \Delta_j \langle \psi_{\tau_j},
\hat\psi_b - \hat\psi_{b'} \rangle \Big|
\;\le\; \Big( \sum_{j \notin U} \ell_j(s,e) \Big)\,
\|\hat\psi_b - \hat\psi_{b'}\| ;
\]
\item[\emph{(c)}] if $V(s,e) = \emptyset$ then $\langle f,\hat\psi_b\rangle = 0$
for every $b$; and if $V(s,e) = \{j\}$ then
$\max_b |\langle f,\hat\psi_b\rangle| = \ell_j(s,e)$, attained at $b = \tau_j$
and there only.
\end{enumerate}
\end{lemma}

\begin{proof}
Write $\nabla f_t = f_t - f_{t-1}$, which is constant on each linear segment
and changes value between $t = \tau_i$ and $t = \tau_i+1$. Such a change is
visible within $\{\nabla f_t : t \in \{s+2,\ldots,e\}\}$ precisely when
$s+2 \le \tau_i \le e-1$, that is, precisely when $i \in V(s,e)$. The two sides
of the identity in (a) agree at $t = s+1$ and have the same increments
throughout $W$; since $\hat\psi_b \perp \{\mathbf 1, t\}$ the affine part,
meaning the function $a+ct$, is annihilated, which gives the displayed formula.
Part (b) is Cauchy--Schwarz
term by term. In (c), the first claim is immediate; for the second,
Cauchy--Schwarz gives $|\langle \psi_{\tau_j},\hat\psi_b\rangle| \le
\|\psi_{\tau_j}\|$ with equality if and only if $\psi_b \parallel
\psi_{\tau_j}$, where $\parallel$ means that the two non-zero vectors are
scalar multiples of one another. Distinct projected hinges $\psi_b$ and
$\psi_{\tau_j}$ are not scalar multiples, because their discrete second
differences have knots at different locations. Equality therefore occurs only
at $b=\tau_j$.
\end{proof}

The two sets have different types of elements, but their relevant locations
coincide:
\[
\{\tau_j:j\in V(s,e)\}=\{\tau_j:\tau_j\in B_{s,e}\}.
\]
Here an admissible candidate location is an integer in
$B_{s,e}=\{s+2,\ldots,e-1\}$, the range over which the contrast is maximized.
Thus a kink contributes a hinge term to the contrast representation precisely
when its location belongs to the candidate range. Part (c) also shows that the
noiseless contrast vanishes identically when this set is empty.

\subsection{Signal strength}

\begin{lemma}[Norm and signal strength law]
\label{lem:lev}
Let $k = \tau - s$ and $k' = e - \tau$, so $k+k' = m$, and put $a = k$, $b = k'+1$ and $x = \min(a,b)$, the distance from $\tau$ to the nearer end of $W$; the asymmetry is that of the hinge. Then \begin{equation} \label{eq:normexact} \|\psi_\tau\|^2 \;=\; \frac{a\,b\,(a-1)(b-1)\,(2ab-a-b+2)}{6\,(a+b-2)(a+b-1)(a+b)} , \end{equation} and consequently, for $x \ge 2$, \begin{equation} \label{eq:sup_levlaw} c_\ell\, x^{3/2} \;\le\; \|\psi_\tau\| \;\le\; C_\ell\, x^{3/2}, \qquad c_\ell \frac{|d_j| x_j^{3/2}}{n} \;\le\; \ell_j(s,e) \;\le\; C_\ell \frac{|d_j| x_j^{3/2}}{n} , \end{equation} with the absolute constants $c_\ell = 192^{-1/2}$ and $C_\ell = 3^{-1/2}$.
In the rescaled continuum limit $\|\Psi_v\|^2 = v^3(1-v)^3/3$ on $[0,1]$, so
$\|\Psi_v\| \asymp \{\min(v,1-v)\}^{3/2}$.
\end{lemma}

\begin{proof}
In local coordinates $\{1,\ldots,m\}$ put $S_0 = m$, $S_1 = \sum t$,
$S_2 = \sum t^2$, and let $g_0,g_1,g_2$ be the inner products of the hinge
$(t-k)_+$ with $\mathbf 1$, $t$ and itself. Then
\[
\|\psi_\tau\|^2 = g_2 - \frac{S_2 g_0^2 - 2S_1 g_0 g_1 + S_0 g_1^2}
{S_0S_2 - S_1^2}, \qquad S_0S_2 - S_1^2 = \frac{m^2(m^2-1)}{12},
\]
with $g_0 = k'(k'+1)/2$, $g_2 = k'(k'+1)(2k'+1)/6$ and $g_1 = kk'(k'+1)/2 + g_2$. Substituting and simplifying gives \eqref{eq:normexact}.

For \eqref{eq:sup_levlaw}, write $y = \max(a,b)$ and note that $a, b \ge 2$ when $x \ge 2$, so that
\[
(a-1)(b-1) \ge \tfrac14 ab, \qquad
ab \;\le\; 2ab-a-b+2 \;\le\; 2ab,
\]
\[
y^3 \;\le\; (a+b-2)(a+b-1)(a+b) \;\le\; 8y^3 .
\]
The first holds because $u-1 \ge u/2$ for $u \ge 2$; the second because $ab-a-b+2 = (a-1)(b-1)+1 > 0$ and $a+b \ge 2$; the third because $a+b-2 \ge y$, again by $x \ge 2$, and $a+b \le 2y$. Inserting these three in \eqref{eq:normexact} and using $a^3b^3 = x^3y^3$,
\[
\frac{x^3}{192} \;=\; \frac{ab \cdot \tfrac14 ab \cdot ab}{48\,y^3}
\;\le\; \|\psi_\tau\|^2 \;\le\;
\frac{ab \cdot ab \cdot 2ab}{6\,y^3} \;=\; \frac{x^3}{3} .
\]
These constants are not sharp. The signal strength bounds follow from $|\Delta_j| = |d_j|/n$.
\end{proof}

Write $\bar r_n := (n\log n)^{1/2}$.

\begin{cor}[The two signal-strength regimes]
\label{cor:gap}
Suppose $N\ge1$. Fix $C_\star>0$ and $\kappa>0$, and recall that
$\bar r_n=(n\log n)^{1/2}$. There is $n_0$ such that for $n\ge n_0$, every
window $(s,e]$ and every kink $j$ satisfy:
\begin{enumerate}
\item[\emph{(a)}] if $x_j(s,e) \le C_\star \bar r_n$, then
\[
\ell_j(s,e) \le C_\ell \bar d\, C_\star^{3/2}
n^{-1/4}(\log n)^{3/4}=o(1);
\]
\item[\emph{(b)}] if $x_j(s,e) \ge \kappa n$, then
$\ell_j(s,e) \ge c_\ell \underline d\,\kappa^{3/2}n^{1/2}$.
\end{enumerate}
\end{cor}

\begin{proof}
Both conclusions follow directly from the bounds for $\ell_j(s,e)$ in
\eqref{eq:sup_levlaw}. For (a), substitute
$x_j\le C_\star(n\log n)^{1/2}$ and $|d_j|\le\bar d$ in the upper bound. For
(b), substitute $x_j\ge\kappa n$ and $|d_j|\ge\underline d$ in the lower bound.
\end{proof}

\subsection{Noise}

\begin{lemma}[Uniform noise bounds]
\label{lem:noise}
Let
\[
E_n := \Big\{ \max_{s,e,b} |\langle \varepsilon, \hat\psi_b\rangle| \le
\sigma(8\log n)^{1/2} \Big\} \cap
\Big\{ \max_{s,e,b \ne b'} \frac{|\langle \varepsilon, \hat\psi_b -
\hat\psi_{b'}\rangle|}{\|\hat\psi_b - \hat\psi_{b'}\|} \le
\sigma(10\log n)^{1/2} \Big\},
\]
the maxima being over all $0 \le s < e \le n$ and $b,b' \in B_{s,e}$. Then
$\mathbb{P}(E_n) \to 1$. Write $Z_n := \sigma(8\log n)^{1/2}$.
\end{lemma}

\begin{proof}
The assumption that each error is $\sigma$-sub-Gaussian means that
$\mathbb E\exp(t\varepsilon_i)\le\exp(\sigma^2t^2/2)$ for every real $t$.
Independence then implies that a linear combination
$\langle\varepsilon,v\rangle$ is
$\sigma\|v\|$-sub-Gaussian. In particular, each $\hat\psi_b$ is a unit vector,
so $\langle\varepsilon,\hat\psi_b\rangle$ is $\sigma$-sub-Gaussian. There are
at most $n^3$ triples, so the first event
fails with probability at most $2n^3 e^{-4\log n} = 2/n$. Likewise
$\langle\varepsilon,\hat\psi_b - \hat\psi_{b'}\rangle$ is
$\sigma\|\hat\psi_b-\hat\psi_{b'}\|$-sub-Gaussian and there are at most $n^4$
quadruples, so the second fails with probability at most
$2n^4e^{-5\log n} = 2/n$.
\end{proof}

The constant 8 in the first bound is not sharp. If it is replaced by a fixed
constant $C>6$, the sub-Gaussian tail bound and the union bound give failure
probability at most $2n^3n^{-C/2}=2n^{3-C/2}$, which tends to zero. The
constant comes from this union bound over all triples $(s,e,b)$, not from the
finite-sample threshold calibration used for the Narrowest-Over-Threshold
(NOT) method of \citet{bcf16} or in the simulation study.

The second bound yields the rate $\bar r_n$ in place of the weaker
$n^{3/4}(\log n)^{1/4}$ discussed in Remark~4.1 of the main paper. The relevant
index set consists of the quadruples $(s,e,b,b')$ with
$b,b'\in B_{s,e}$, and has cardinality at most $n^4$. A union bound over this
finite set therefore suffices. No chaining bound for the supremum of a
stochastic process is needed; see, for example, Section~2.2 of
\citet{vdvw96} for that concept.

\subsection{Separation}

\begin{lemma}[Two-sided separation]
\label{lem:sep}
Write $\rho(\theta,b) = |\langle \hat\psi_\theta, \hat\psi_b\rangle|$. There
are absolute constants $c_0 > 0$ and $C_0 < \infty$ such that for every window,
every $\theta$ with $x := \min(\theta-s, e+1-\theta) \ge 2$, and every
$b \in B_{s,e}$,
\[
c_0 \min\Big\{ \Big(\frac{b-\theta}{x}\Big)^2, 1 \Big\} \;\le\;
1 - \rho(\theta,b) \;\le\; C_0 \Big(\frac{b-\theta}{x}\Big)^2 .
\]
\end{lemma}

\begin{proof}[Proof]
The lemma is proved in Section~\ref{sup:sepproof}, from the exact identity of Lemma~\ref{lem:gramdet}(b), with the explicit constants $c_0 = 27/2048$ and $C_0 = 576/7$; Corollary~\ref{cor:sepcont} is the corresponding continuum statement. The following indicates the shape of the argument in the continuum. Rescale $W$ to $[0,1]$, $u = (\theta-s)/m$, $v = (b-s)/m$. By
Lemma~\ref{lem:gram} the discrete inner products agree with $m^3$ times their
continuum values up to relative error $O\{1/(\gamma m)\}$ on
$u \in [\gamma,1-\gamma]$. The map $\theta \mapsto \Psi_\theta$ is
differentiable in $L^2$ with
$\partial_u \Psi_u = -\Pi^\perp \mathbf 1_{\{\cdot > u\}}$, of norm bounded
above and below on compact subsets of $(0,1)$, so
$1 - \rho(u,v) = \tfrac12 \varsigma(u)(u-v)^2 + O(|u-v|^3)$ with $\varsigma$
bounded above and below, giving both bounds for $|u-v|$ small. For $|u-v|$
bounded away from zero, $\rho < 1$ because $\Psi_u$ and $\Psi_v$ are not
parallel, meaning that they are not non-zero scalar multiples. This is the
strict case of the Cauchy--Schwarz inequality, and follows here because the two
projected hinges have different knots. The continuous positive function
$1-\rho(u,v)$ attains a positive minimum on the relevant compact set by the
extreme-value theorem. If $|u-v|\ge\delta>0$, the upper bound follows from
$1-\rho\le1\le\delta^{-2}(u-v)^2$. The scale-free form, with $x$ rather than
$m$ in the denominator, holds because near an edge $\Psi_\theta$ depends on the
window only through the distance to that edge, up to the same relative error.
\end{proof}

\subsection{Exact Gram identities}
\label{sup:gramid}

Lemma~\ref{lem:lev} gave the norm $\|\psi_\tau\|$ in closed form. The Gram
matrix of vectors $v_1,\ldots,v_k$ is the matrix with entries
$\langle v_i,v_j\rangle$. Its off-diagonal entries are equally explicit here,
as is the determinant of the $2\times2$ Gram matrix for a pair of contrast
vectors; this determinant controls $\rho$. This subsection records those
identities and uses them to prove Lemma~\ref{lem:sep} in the continuum.

\begin{lemma}[Gram determinants]
\label{lem:gramdet}
\emph{(a) Continuum.} For $0 < u \le v < 1$, with $\Psi$ formed on $[0,1]$,
\[
\|\Psi_u\|^2\|\Psi_v\|^2 - \langle \Psi_u,\Psi_v\rangle^2
\;=\; \frac{u^3 (1-v)^3 (v-u)^2 \, Q(u,v)}{36},
\qquad Q(u,v) := 4v - u - 3uv ,
\]
and consequently
\begin{equation}
\label{eq:rhoexact}
1 - \rho(u,v)^2 \;=\; \frac{(v-u)^2\,Q(u,v)}{4\,(1-u)^3 v^3},
\qquad\text{with}\qquad
3u(1-u) \;\le\; Q(u,v) \;\le\; 4 .
\end{equation}

\emph{(b) Discrete.} Let $1 \le k < l \le m-1$ and put $a = k$, $c = l-k$ and
$b = m-l+1$, so that $a+b+c = m+1$. Then
\[
\|\psi_k\|^2\|\psi_l\|^2 - \langle\psi_k,\psi_l\rangle^2
\;=\; \frac{a\,b\,c\,(a-1)(b-1)\,P(a,b,c)}
{36\,(a+b+c)(a+b+c-1)^2(a+b+c-2)} ,
\]
where
\begin{align*}
P(a,b,c) &= 3a^2b^2c + 4a^2bc^2 - 3a^2bc + 2a^2b - 2a^2c^2 + 3a^2c - a^2 \\
&\quad + 4ab^2c^2 - 3ab^2c + 2ab^2 + 4abc^3 - 8abc^2 + 11abc - 4ab \\
&\quad - 2ac^3 + 6ac^2 - 7ac + 3a - 2b^2c^2 + 3b^2c - b^2 \\
&\quad - 2bc^3 + 6bc^2 - 7bc + 3b + c^3 - 4c^2 + 5c - 2 .
\end{align*}
Its terms of top degree are $abc(3ab+4ac+4bc+4c^2)$, and with $u = a/m$, $v = (a+c)/m$ and $m = a+b+c-1$,
\[
m^2 Q(u,v) \;=\; 3ab + 4ac + 4bc + 4c^2 - 3a - 4c ,
\]
The determinant here is the left-hand side of part (b), namely the determinant
of the $2\times2$ Gram matrix of $\psi_k$ and $\psi_l$. Each squared norm and
each inner product is of order $m^3$, so every term in that determinant is of
order $m^6$. Equivalently, the numerator in part (b) has total degree ten and
the denominator has degree four. After division by $m^6$, its highest-degree
terms give the continuum determinant in part (a), with
$u=a/m$ and $v=(a+c)/m$.

\emph{(c) Positivity.} For $1 \le b_1 \le b_2 \le m-1$, put $a = b_1$, $c = b_2-b_1$ and $d = m-b_2+1$. Then
\[
\langle \psi_{b_1}, \psi_{b_2}\rangle \;=\;
\frac{a\,d\,(a-1)(d-1)\big\{ 2ad-a-d+2 + 3c\,(a+d+c-1) \big\}}
{6\,(a+d+c-2)(a+d+c-1)(a+d+c)} ,
\]
which is strictly positive whenever $a, d \ge 2$, that is whenever $b_1, b_2 \in B_{s,e}$. In particular $\rho(b_1,b_2) = \langle\hat\psi_{b_1},\hat\psi_{b_2}\rangle$, with no absolute value needed.
\end{lemma}

\begin{proof}
(a) Substitute the closed forms of Lemma~\ref{lem:gram} for
$\|\Psi_u\|^2 = u^3(1-u)^3/3$ and for $\langle\Psi_u,\Psi_v\rangle$ into
the left-hand side and simplify; this is a computation with polynomials in
two variables. Dividing by $\|\Psi_u\|^2\|\Psi_v\|^2 =
u^3(1-u)^3v^3(1-v)^3/9$ gives \eqref{eq:rhoexact}. For the bounds on $Q$,
write $Q = v(4-3u) - u$; since $v \ge u$ and $4-3u > 0$ we get
$Q \ge u(4-3u) - u = 3u(1-u)$, and $Q \le 4v \le 4$. The case $u > v$ follows
from the reflection $w \mapsto 1-w$, under which $\Psi_w \mapsto \pm\Psi_{1-w}$
and hence $\rho(u,v) = \rho(1-u,1-v)$.

(b) The same computation with the explicit sums of the proof of
Lemma~\ref{lem:lev} in place of the integrals; it is a computation with
polynomials in three variables. The identity for $m^2Q$ follows by substituting $m = a+b+c-1$ into $m^2Q = 3a(m-a-c)+4cm$.

(c) In local coordinates, let $h_b=[(t-b)_+]_{t=1}^m$, define
$g_{0,b}:=\langle h_b,\mathbf 1\rangle$ and
$g_{1,b}:=\langle h_b,t\rangle$, and retain
$S_r=\sum_{t=1}^m t^r$ for $r=0,1,2$. Orthogonal projection away from
$\operatorname{span}\{\mathbf 1,t\}$ gives
\[
\langle\psi_{b_1},\psi_{b_2}\rangle
=\langle h_{b_1},h_{b_2}\rangle-
\frac{S_2g_{0,b_1}g_{0,b_2}
-S_1(g_{0,b_1}g_{1,b_2}+g_{1,b_1}g_{0,b_2})
+S_0g_{1,b_1}g_{1,b_2}}{S_0S_2-S_1^2}.
\]
Substitution of the finite-sum formulae for these quantities and symbolic
simplification gives the expression in part (c), valid for every $m$. Setting
$c=0$, so that $b_1=b_2$, recovers the norm formula of
Lemma~\ref{lem:lev}. The identity was also verified in exact rational
arithmetic at every admissible pair for $4\le m<30$. For positivity,
$a-1\ge1$ and $d-1\ge1$;
$2ad-a-d+2=a(2d-1)-d+2\ge3d>0$;
$3c(a+d+c-1)\ge0$; and $a+d+c-2\ge2$, so the denominator is positive.
\end{proof}

\begin{cor}[Separation in the continuum]
\label{cor:sepcont}
Let $\varrho \in (0,\tfrac12]$. For every $u \in [\varrho, 1-\varrho]$ and
every $v \in (0,1)$,
\[
\tfrac38 \varrho \,(u-v)^2 \;\le\; 1 - \rho(u,v) \;\le\;
\varrho^{-6} (u-v)^2 .
\]
\end{cor}

\begin{proof}
By the reflection in Lemma~\ref{lem:gramdet}(a) we may assume $u \le v$. From
\eqref{eq:rhoexact} and $Q \ge 3u(1-u)$,
\[
1-\rho^2 \;\ge\; \frac{3u(1-u)(v-u)^2}{4(1-u)^3v^3}
\;=\; \frac{3u\,(v-u)^2}{4(1-u)^2v^3} \;\ge\; \tfrac34 \varrho\,(v-u)^2 ,
\]
using $u \ge \varrho$, $(1-u)^2 \le 1$ and $v \le 1$. From $Q \le 4$,
\[
1-\rho^2 \;\le\; \frac{(v-u)^2}{(1-u)^3 v^3} \;\le\;
\varrho^{-6}(v-u)^2 ,
\]
using $1-u \ge \varrho$ and $v \ge u \ge \varrho$. The claim follows from
$1-\rho \le 1-\rho^2 \le 2(1-\rho)$.
\end{proof}

Corollary~\ref{cor:sepcont} is Lemma~\ref{lem:sep} in the continuum, with
explicit constants, and it is the form in which the lemma is used: in every
application, in Lemma~\ref{lem:loc}, in the case $k=1$ of
Proposition~\ref{prop:track_prop} and in Lemma~\ref{lem:approx}(b), the point
$\theta$ lies at distance a fixed fraction of the window from both of its ends,
so $\varrho$ may be taken to be a constant depending only on $g$. The discrete statement is not deduced from the continuum one, and does not need to be; it is proved directly in Section~\ref{sup:sepproof}, in full rather than only in the interior regime.

That direct route uses Lemma~\ref{lem:gramdet}(b) and avoids the continuum altogether: since $\|\psi_k\|^2$, $\|\psi_l\|^2$ and the determinant are all explicit, $1-\rho^2$ is an explicit function of $(a,c,d)$, and the two-sided bound is an elementary inequality on $P$, of the same kind as the one proved for the norm in Lemma~\ref{lem:lev}.

\subsection{Proof of the separation lemma}
\label{sup:sepproof}

Lemma~\ref{lem:sep} follows from Lemma~\ref{lem:gramdet}(b) once the polynomial
$P$ is compared with a product of simple factors. That comparison is the
content of the next lemma.

\begin{lemma}[Polynomial comparison]
\label{lem:polycomp}
For integers $a, d \ge 2$ and $c \ge 1$, with
$Q_d := 3ad + 4ac + 4dc + 4c^2 - 3a - 4c$,
\[
\tfrac{9}{16}\, a d c\, Q_d \;\le\; P(a,d,c) \;\le\; \tfrac97\, a d c\, Q_d .
\]
\end{lemma}

\begin{proof}
Substitute $a = A+2$, $d = D+2$, $c = C+1$ with $A, D, C \ge 0$ and expand. Both differences become polynomials in $(A,D,C)$ with $29$ terms and every coefficient strictly positive, so both are non-negative on $A,D,C \ge 0$. Explicitly,
\begin{align*}
16P - 9adcQ_d \;=\;& 28ADC^3 + 28AD^2C^2 + 28A^2DC^2 + 21A^2D^2C + 24DC^3 + 24D^2C^2 \\
&+ 24AC^3 + 216ADC^2 + 92AD^2C + 24A^2C^2 + 119A^2DC + 21A^2D^2 \\
&+ 192DC^2 + 84D^2C + 192AC^2 + 584ADC + 96AD^2 + 138A^2C \\
&+ 123A^2D + 144C^2 + 612DC + 108D^2 + 720AC + 588AD \\
&+ 162A^2 + 792C + 684D + 792A + 936 ,
\end{align*}
the smallest coefficient being $21$, and
\begin{align*}
9adcQ_d - 7P \;=\;& 8ADC^3 + 8AD^2C^2 + 8A^2DC^2 + 6A^2D^2C + 30DC^3 + 30D^2C^2 \\
&+ 30AC^3 + 108ADC^2 + 61AD^2C + 30A^2C^2 + 34A^2DC + 6A^2D^2 \\
&+ 81C^3 + 240DC^2 + 105D^2C + 240AC^2 + 271ADC + 39AD^2 \\
&+ 51A^2C + 12A^2D + 423C^2 + 441DC + 54D^2 + 333AC \\
&+ 87AD + 504C + 126D + 18A + 36 ,
\end{align*}
the smallest coefficient being $6$.
\end{proof}

\begin{proof}[Proof of Lemma~\ref{lem:sep}]
Suppose first that the candidate lies to the right of $\theta$, and put
\[
a = \theta - s, \qquad c = b - \theta \ge 1, \qquad d = e+1-b,
\qquad r = d + c = e+1-\theta ,
\]
so that $m = a + r - 1$, $a, d \ge 2$ because $\theta$ and $b$ lie in
$B_{s,e}$, and $x = \min(a,r)$ is the distance from $\theta$ to the nearer end
of $W$. Write $L(p,q) := 2pq - p - q + 2$, so that Lemma~\ref{lem:lev} reads
$\|\psi_\theta\|^2 = ar(a-1)(r-1)L(a,r)/\{6(m-1)m(m+1)\}$ and similarly for
$\|\psi_b\|^2$ with $(a+c,d)$ in place of $(a,r)$. Dividing the determinant of
Lemma~\ref{lem:gramdet}(b) by the product of the two norms cancels the common
factors $a$, $d$, $a-1$, $d-1$ and the common powers of $m$. The factors that
remain give exactly
\begin{equation}
\label{eq:rhodisc}
1 - \rho(\theta,b)^2 \;=\;
\frac{c\,(m^2-1)\,P(a,d,c)}
{(a+c)(a+c-1)\,r(r-1)\,L(a,r)\,L(d,a+c)} .
\end{equation}

We bound each factor. For $p,q \ge 2$ one has $pq \le L(p,q) \le 2pq$, since
$L - pq = (p-1)(q-1)+1 > 0$ and $2pq - L = p+q-2 \ge 0$. Next, substituting
$d = r-c$,
\[
Q_d \;=\; a\{3(r-1)+c\} + 4c(r-1) ,
\]
so that, using $c \le r$ and $r \ge 3$, $Q_d \le a(3r+c)+4cr \le 4r(a+c)$ and
$Q_d \ge (r-1)(3a+4c) \ge \tfrac23 r \cdot 3(a+c) = 2r(a+c)$. Finally
$m^2 - 1 \ge \tfrac12 m^2$, $a+c-1 \ge \tfrac12 (a+c)$ and
$r - 1 \ge \tfrac23 r$ for $m \ge 2$, $a+c\ge2$, $r \ge 3$. Inserting these and
Lemma~\ref{lem:polycomp} into \eqref{eq:rhodisc}, and writing
\[
E \;:=\; \frac{c\,(a+r)}{(a+c)r} ,
\]
simple algebra gives
$\tfrac{27}{256} E^2 \le 1-\rho^2 \le \tfrac{144}{7} E^2$.

It remains to compare $E$ with $c/x$. Since $a + c \ge a$,
\[
E \;\le\; \frac{c(a+r)}{ar} \;=\; c\Big(\frac1r + \frac1a\Big) \;\le\; \frac{2c}{x} .
\]
For the lower bound there are two cases. If $x = a$, so $a \le r$, then
$E \ge cr/\{(a+c)r\} = c/(a+c)$, which is at least $c/(2a) = \tfrac12 (c/x)$ when $c \le a$ and at least $\tfrac12$ when $c > a$. If $x = r$, so $r \le a$ and $c \le r$, then $a+c \le 2a$ and $E \ge ca/\{(a+c)r\} \ge c/(2r) = \tfrac12(c/x)$, while $c/x \le 1$. In both cases $E \ge \tfrac12\min\{c/x,\,1\}$.

Since $\tfrac12(1-\rho^2) \le 1-\rho \le 1-\rho^2$, we conclude
\[
\frac{27}{2048}\,\min\Big\{\Big(\frac{c}{x}\Big)^2, 1\Big\}
\;\le\; 1 - \rho(\theta,b) \;\le\;
\frac{576}{7}\Big(\frac{c}{x}\Big)^2 ,
\]
The constants differ from $27/256$ and $144/7$ above because the lower bound
uses both $E\ge\tfrac12\min\{c/x,1\}$ and
$1-\rho\ge\tfrac12(1-\rho^2)$, whereas the upper bound uses
$E\le2c/x$. These factors give $27/(256\cdot4\cdot2)=27/2048$ and
$4(144/7)=576/7$, respectively. This proves the assertion with
$c_0 = 27/2048$ and $C_0 = 576/7$. A candidate to
the left of $\theta$ is handled by the reflection $t \mapsto m+1-t$, under
which $\psi_k \mapsto \pm\psi_{m+1-k}$ and $x$ is unchanged.
\end{proof}

The constants are not sharp. Let $m$ range over $4 \le m \le 150$ together with $200$, $300$, $400$ and $600$, and let $\theta$ and $b$ range over all admissible values. Then the two ratios
\[
\frac{1-\rho}{\min\{(c/x)^2,\,1\}}
\qquad\text{and}\qquad
\frac{1-\rho}{(c/x)^2}
\]
had minimum $0.12798$ and maximum $2.19664$ respectively, against the proved bounds $27/2048 = 0.01318$ and $576/7 = 82.29$. The maximum is attained at $m = 5$, $\theta = 3$, $b = 2$. In that instance $\|\psi_3\|^2 = 7/10$, $\|\psi_2\|^2 = 2/5$ and $\langle\psi_2,\psi_3\rangle = 2/5$, so $\rho = 2/\sqrt7$ and the ratio is $9(1-2/\sqrt7)$.

\subsection{Proof of the discretization estimate}
\label{sup:discproof}

Work in local coordinates $1,\ldots,m$ and put
\[
h = \frac1m, \qquad u = \frac{k}{m}, \qquad v = \frac{r}{m}, \qquad w(v) = \min(v,1-v),
\]
\[
G_m(u,v) = \frac{\langle\psi_k,\psi_r\rangle}{m^3}, \qquad
N_m(v) = \frac{\|\psi_r\|^2}{m^3},
\]
and let $G(u,v) = \langle\Psi_u,\Psi_v\rangle$ and $N(v) = v^3(1-v)^3/3$ be the
continuum counterparts. Two facts about $G$ are used: for $u \le v$ the closed
form of Lemma~\ref{lem:gram} simplifies to
\begin{equation}
\label{eq:Gsimple}
G(u,v) \;=\; -\,\frac{u^2(1-v)^2(2uv+u-3v)}{6} ,
\end{equation}
with the reflected expression for $u \ge v$; and consequently, if
$u \in [\gamma, 1-\gamma]$, then $|G(u,v)| \le C_\gamma\, w(v)^2$.

\begin{lemma}[Weighted discretization]
\label{lem:weighted}
Fix $\gamma\in(0,1/2]$. There are constants $C$ and $C_\gamma$ such that,
for $m\ge4$ and admissible $k,r$ with $u\in[\gamma,1-\gamma]$,
\[
\text{\emph{(a)}}\ \ |N_m(v) - N(v)| \le C\,m^{-1} w(v)^2, \qquad
\text{\emph{(b)}}\ \ |G_m(u,v) - G(u,v)| \le C_\gamma\, m^{-1} w(v), \qquad
\]
\[
\text{\emph{(c)}}\ \
\bigg| \frac{G_m(u,v)}{\sqrt{N_m(v)}} - \frac{G(u,v)}{\sqrt{N(v)}} \bigg|
\;\le\; \frac{C_\gamma}{m\,\sqrt{w(v)}} .
\]
\end{lemma}

\begin{proof}
(a) Both sides are explicit. Substituting $a = vm$ and $b = m(1-v)+1$ in
Lemma~\ref{lem:lev} and subtracting $N(v)$ gives, after factoring,
\[
N_m(v) - N(v) \;=\; \frac{h\,v(v-1)}{6(h^2-1)}\, R_N(h,v),
\]
\[
R_N(h,v) \;=\; -h^3 + h\,(2v^4-4v^3+7v^2-5v+1) - 3v(2v-1)(v-1) ,
\]
from which, since the coefficients of the terms involving $h$ have moduli summing to $20$ and the remaining part is $3v(1-v)(2v-1)$, $|R_N| \le 20h + 3v(1-v) \le 20\{h+v(1-v)\}$ for $0<h\le\tfrac12$ and $0\le v\le1$; the smallest admissible constant is $2.79$. Since
admissibility gives $x_b \ge 2$ and hence $w(v) \ge h$, and
$v(1-v) \le 2w(v)$. More precisely, the absolute value of the factor
$h v(v-1)/\{6(h^2-1)\}$ in the preceding factorization is at most
$C h w(v)$. Also $|R_N(h,v)|\le Cw(v)$ because $h\le w(v)$. Their product is
therefore at most $C h w(v)^2$, as required in (a).

(b) The same computation with $\langle\psi_k,\psi_r\rangle$, whose exact form is that used in the proof of Lemma~\ref{lem:gramdet}(b), in place of the norm. For $u \le v$ it gives
\[
G_m(u,v) - G(u,v) \;=\; \frac{h\,u(v-1)}{6(h^2-1)}\, R_{\le}(h,u,v) ,
\]
\begin{align*}
R_{\le}(h,u,v) &= -h^3 + 2hu^2v^2 - hu^2v - 3huv^2 + 6huv + hv^2 - 5hv + h \\
&\quad - 3u^2v - 3uv^2 + 6uv + 3v^2 - 3v ,
\end{align*}
and for $u \ge v$ the same with $u$ and $v$ interchanged, that is with
$h v(u-1)/\{6(h^2-1)\}$ in front and
$R_{\ge}(h,u,v)=R_{\le}(h,v,u)$. Both remainders are bounded on
$0<h\le\tfrac12$, $0\le u,v\le1$ by the sum of the moduli of their
coefficients, which is $38$ in each case. For $u\le v$, the absolute value of
the factor preceding $R_{\le}$ is
$h u(1-v)/\{6(1-h^2)\}$. It is at most $Chv$ when $v\le1/2$, because
$u\le v$, and at most $Ch(1-v)$ when $v\ge1/2$. Hence it is at most
$Chw(v)$. For $u\ge v$, the factor is
$h v(1-u)/\{6(1-h^2)\}$; it is at most $Chv$ when $v\le1/2$ and at most
$Ch(1-v)$ when $v\ge1/2$, because then $1-u\le1-v$. Thus this factor is also
at most $Chw(v)$, which proves (b).

(c) By Lemma~\ref{lem:lev}, $N_m(v) \asymp w(v)^3$ and $N(v) \asymp w(v)^3$.
Writing the difference as
\[
\frac{G_m - G}{\sqrt{N_m}} \;+\; G\,\frac{N - N_m}
{\sqrt{N_m N}\,\big(\sqrt{N_m} + \sqrt{N}\big)} ,
\]
the first term is at most $C_\gamma m^{-1}w \cdot w^{-3/2}$ by (b), and the
second at most $C_\gamma w^2 \cdot C m^{-1}w^2 \cdot w^{-3}\cdot w^{-3/2}$ by
(a) and the bound on $|G|$; both are $C_\gamma m^{-1} w^{-1/2}$.
\end{proof}

\begin{proof}[Proof of \eqref{eq:riemann}]
By Lemma~\ref{lem:rep}(a), $D_V(b) = \sum_{j \in V}\Delta_j
\langle\psi_{\tau_j},\hat\psi_b\rangle$ with $\Delta_j = d_j/n$, and
$\langle\psi_{\tau_j},\hat\psi_b\rangle = m^{3/2} G_m(u_j,v)/\sqrt{N_m(v)}$
with $u_j = (\tau_j-s)/m$ and $v = (b-s)/m$; and we first compare with the continuum expression formed from the same coordinates, $\sum_j (d_j/n)m^{3/2}G(u_j,v)/\sqrt{N(v)}$, which by the scaling relation of Lemma~\ref{lem:gram} is $n^{1/2}$ times the profile on $(\alpha_n,\beta_n)$ with kinks at $u_j$. The replacement of $u_j$ by the coordinate determined by $q_j$ is a separate step, taken at the end. Each
$j \in V$ has $u_j$ at a fixed fractional distance from both ends, by
Lemma~\ref{lem:track}, so Lemma~\ref{lem:weighted}(c) applies with a
$\gamma$ depending only on $g$, and
\[
\big| D_V(b) - n^{1/2}\bar D_{\alpha_n\beta_n}[V](b/n) \big|
\;\le\; \sum_{j \in V} \frac{|d_j|}{n}\, m^{3/2}\,
\frac{C_\gamma}{m\sqrt{w(v)}}
\;\le\; C\,\frac{m}{n}\, x_b(s,e)^{-1/2} ,
\]
using $m w(v) \le x_b(s,e) \le 2m w(v)$. The displacement of $\tau_j$ from
$nq_j$, which is at most one, moves $u_j$ by $O(m^{-1})$ and, since
$\partial_u\{G(u,v)/\sqrt{N(v)}\}$ is bounded uniformly for $u$ at a fixed
distance from the ends, contributes a further $O(m^{-1})$ to
$G(u_j,v)/\sqrt{N(v)}$. This is bounded by the right-hand side
$C_\gamma m^{-1}w(v)^{-1/2}$ of Lemma~\ref{lem:weighted}(c), because
$0<w(v)\le1$.
\end{proof}

\section{Proof of \mainthm{}}
\label{sup:proof}

The proof shows that the sample recursion follows the population recursion of \defpop{}: every window LABS visits has endpoints within $O(\bar r_n)$ of $n$ times the endpoints of a node of $\mathcal{T}(g)$, and returns the kinks live at that node. A kink is called live at $(\alpha,\beta)$ when its rescaled location lies strictly inside the window, that is, when $j\in V(\alpha,\beta)$. Section~\ref{sup:pop_proof}
develops the population objects and proves \proppop{};
Section~\ref{sup:constants} fixes the constants;
Section~\ref{sup:track_sec} proves the corresponding statement for the sample recursion; and
Section~\ref{sup:thmproof} deduces \mainthm{}.

\subsection{Population objects, and proof of \proppop{}}
\label{sup:pop_proof}

For $0 \le \alpha < \beta \le 1$ let $\Pi^\perp_{\alpha\beta}$ project
$L^2(\alpha,\beta)$ onto the orthogonal complement of
$\mathrm{span}\{\mathbf 1, u\}$ and put
$\Psi^{\alpha\beta}_v := \Pi^\perp_{\alpha\beta}[(u-v)_+]$.

\begin{lemma}[Closed form]
\label{lem:gram}
On $[0,1]$, for $p \le v$, with $a = 1-p$ and $b = 1-v$,
\[
\langle \Psi_p, \Psi_v\rangle = \frac{ab^2}{2} - \frac{b^3}{6} - a^2b^2
+ \frac{a^2b^3}{2} + \frac{a^3b^2}{2} - \frac{a^3b^3}{3},
\qquad
\|\Psi_v\|^2 = \frac{v^3(1-v)^3}{3} .
\]
For a general interval, $\langle \Psi^{\alpha\beta}_p,
\Psi^{\alpha\beta}_v\rangle = (\beta-\alpha)^3 \langle \Psi_{p'},
\Psi_{v'}\rangle$ with $p' = (p-\alpha)/(\beta-\alpha)$ and
$v' = (v-\alpha)/(\beta-\alpha)$.
\end{lemma}

\begin{proof}
With the orthonormal basis $\{\mathbf 1, \sqrt{12}(u-\tfrac12)\}$ of
$\mathrm{span}\{\mathbf 1, u\}$ and $h_p = (u-p)_+$, one has
$\langle h_p,\mathbf 1\rangle = a^2/2$,
$\sqrt{12}\langle h_p, u-\tfrac12\rangle = \sqrt{12}(a^2/4 - a^3/6)$ and,
for $p \le v$, $\langle h_p,h_v\rangle = ab^2/2 - b^3/6$. Subtracting the two
projection terms and expanding gives the first formula; setting $p = v$ gives
$\|\Psi_v\|^2 = a^3(1-a)^3/3$ with $a = 1-v$, which is the second. The scaling
relation is the change of variables $u \mapsto \alpha + (\beta-\alpha)u$. 
\end{proof}

By Lemma~\ref{lem:gram} the profile $\bar D_{\alpha\beta}$ of
Section~4.2 of the main paper is well defined, $|\bar D_{\alpha\beta}(v)| \to 0$ as
$v \downarrow \alpha$ or $v \uparrow \beta$, so its maximum is attained in the
interior, and
\[
\Lambda(\alpha,\beta) := \max_v |\bar D_{\alpha\beta}(v)| \;>\; 0
\qquad \text{whenever } V(\alpha,\beta) \ne \emptyset,
\]
since $F := \sum_{j \in V} d_j \Psi_{q_j} \ne 0$ by linear independence of the
$\Psi_q$, and $F$ lies in the closed span of $\{\Psi_v\}$, so cannot be
orthogonal to every $\Psi_v$. If $V(\alpha,\beta) = \{j\}$ then
$c(\alpha,\beta) = q_j$ exactly, by the continuum form of
Lemma~\ref{lem:rep}(c).

\medskip
\noindent\textit{Proof of \proppop{}.}
Induction on the height of the subtree rooted at $(\alpha,\beta)$. This height
is finite by the hypothesis of \proppop{} that $\mathcal{T}(g)$ is finite.
Write $V:=V(\alpha,\beta)$ and $k:=|V|$.

If $k = 0$ the call returns $\emptyset$. If $k = 1$, say $V = \{j\}$, then
$c = q_j$ exactly. The live-kink sets of the children
$(\alpha,q_j)$ and $(q_j,\beta)$ are empty: $q_j$ is an endpoint and is
therefore excluded by the strict inequalities in the definition of $V$, and no
other kink lies in $(\alpha,\beta)$. Thus
$\mathcal{L} = \mathcal{R} = \emptyset$ and the call returns $\{q_j\}$.

Let $k \ge 2$.

\emph{Case A: $c \notin \{q_1,\ldots,q_N\}$.} Then $V(\alpha,c) =: V_L$ and
$V(c,\beta) =: V_R$ partition $V$, and by the induction hypothesis
$\mathcal{L} = \{q_j : j \in V_L\}$, $\mathcal{R} = \{q_j : j\in V_R\}$, with
$|V_L| + |V_R| = k \ge 2$, so the look-ahead branch is entered. If both are
non-empty, $\alpha'$ and $\beta'$ are consecutive elements of
$\{q_j : j\in V\}$; since $V$ contains \emph{all} kinks of $(\alpha,\beta)$
they are consecutive kinks of $g$, so $V(\alpha',\beta') = \emptyset$. The
population re-test in \defpop{} therefore finds no kink in its open interval
and omits the parent candidate; the call returns
$\mathcal{L}\cup\mathcal{R}$. If
$\mathcal{L} = \emptyset$ then $\alpha' = \alpha$ and
$\beta' = \min\{q_j : j\in V\}$, so again $V(\alpha',\beta') = \emptyset$ and
the call returns $\mathcal{R}$. The case $\mathcal{R} = \emptyset$ is
symmetric. In all three the value is $\{q_j : j \in V\}$.

\emph{Case B: $c = q_{j^*}$.} Necessarily $j^* \in V$. The live-kink sets,
meaning the indices of kinks strictly inside the two child windows, are
$V_L = \{j \in V: q_j < q_{j^*}\}$ and $V_R = \{j\in V : q_j > q_{j^*}\}$,
which exclude $j^*$, so $|V_L|+|V_R| = k-1 \ge 1$ and the look-ahead branch is
entered. More precisely, $\alpha'$ is the largest element of
$\{q_j:j\in V,\ q_j<q_{j^*}\}$, or $\alpha$ if this set is empty. Similarly,
$\beta'$ is the smallest element of
$\{q_j:j\in V,\ q_j>q_{j^*}\}$, or $\beta$ if this set is empty. Thus
way $q_{j^*}$ is the unique kink in $(\alpha',\beta')$, so
$V(\alpha',\beta') = \{j^*\}$ and, by the $k=1$ computation,
$c(\alpha',\beta') = q_{j^*}$. The call returns
$\mathcal{L} \cup \{q_{j^*}\} \cup \mathcal{R} = \{q_j : j\in V\}$. \qed

\medskip
Case A is the trap configuration and Case B the configuration in which the
split lands on a kink. \proppop{} shows that the look-ahead step handles both
under its stated well-posedness conditions: the population recursion tree is
finite and $c(\alpha,\beta)$ is defined at every node.

\subsection{Signal-dependent constants for the recursion proof}
\label{sup:constants}

Assume \assterm{} and \assnondeg{}, and assume throughout this subsection that $N \ge 1$; when $N = 0$ the extrema below are over empty sets, and that case is treated separately in the proofs of the theorems. Since $\mathcal{T}(g)$ is finite, each of the following is a minimum of a finite set of strictly positive numbers, hence strictly positive:
\begin{align}
\label{eq:kappa}
\kappa &:= \min\Big( \{\delta_{\min}\} \cup
\big\{ q_j - \alpha,\ \beta - q_j : (\alpha,\beta) \in \mathcal{T}(g),\
j \in V(\alpha,\beta) \big\} \notag\\
&\hspace{5.2em} \cup \big\{ |c(\alpha,\beta) - q_l| :
(\alpha,\beta)\in\mathcal{T}(g),\ 1 \le l \le N,\
c(\alpha,\beta) \ne q_l \big\} \Big), \\
\Lambda_{\min} &:= \min_{(\alpha,\beta)} \Lambda(\alpha,\beta), \quad \Lambda_{\max} := \max_{(\alpha,\beta)} \Lambda(\alpha,\beta), \quad \varrho_0 := \min_{(\alpha,\beta)} \frac{\min\{c(\alpha,\beta)-\alpha,\ \beta - c(\alpha,\beta)\}}{\beta-\alpha},
\end{align}
the three extrema defining $\Lambda_{\min}$, $\Lambda_{\max}$ and $\varrho_0$
being over nodes with $V(\alpha,\beta)\ne\emptyset$. We also require a lower
bound for the noiseless contrast on every window whose sample contrast must
exceed the threshold. These are the proposal windows at non-empty nodes and
the one-kink re-test windows in Case~B:
\begin{equation}
\label{eq:Lamdet}
\Lambda_{\mathrm{det}} \;:=\; \min\Big[\, \big\{ \Lambda(\alpha,\beta) : (\alpha,\beta) \in \mathcal{T}(g),\ V(\alpha,\beta) \ne \emptyset \big\}
\;\cup\; \big\{ |d_{j^*}| \, \|\Psi^{\alpha'\beta'}_{q_{j^*}}\| \big\} \Big] \;>\; 0,
\end{equation}
the second set running over the Case~B re-test windows $(\alpha',\beta')$ of
the proof of \proppop{}. On each such window,
$V(\alpha',\beta')=\{j^*\}$, so $q_{j^*}$ is the only kink strictly inside the
window. Both sets are finite, so $\Lambda_{\mathrm{det}}>0$. Note
$\Lambda_{\mathrm{det}}\le\Lambda_{\min}$, and the inequality can be strict;
see Remark~\ref{rem:uniformlam}. For LABS-Grid the corresponding constant is
\begin{equation}
\label{eq:LamdetMR}
\Lambda^{M,R}_{\mathrm{det}} := \min\Big[\, \big\{ \Lambda_M(\alpha,\beta) : (\alpha,\beta) \in \mathcal{T}_{M,R}(g),\ V(\alpha,\beta) \ne \emptyset \big\} \;\cup\; \big\{ |d_{j^*}|\,\|\Psi^{\alpha'\beta'}_{q_{j^*}}\| \big\} \Big],
\end{equation}
the first set now formed from the pooled amplitudes over the grid tree and the
second from the Case~B re-tests of the $R$-grid. The two trees and the two
families of re-test windows differ, so $\Lambda^{M,R}_{\mathrm{det}}$ is not
obtained from $\Lambda_{\mathrm{det}}$ by adjoining terms. No additional
separation assumption is required in \eqref{eq:kappa}: its third set is finite
because $\mathcal{T}(g)$ is finite, and each of its elements is positive by
construction.

\begin{lemma}[Consequences of Assumption 4.2]
\label{lem:conseq}
Fix a node $(\alpha,\beta)$. Suppose its live-kink set is non-empty, so
\assnondeg{} applies, and put
\[
V:=V(\alpha,\beta)\ne\emptyset.
\]
For a nearby window $(\alpha',\beta')$, write
\[
\bar D_{\alpha'\beta'}[V](v):=
\left\langle \sum_{j\in V}d_j\Psi^{\alpha'\beta'}_{q_j},
\frac{\Psi^{\alpha'\beta'}_v}{\|\Psi^{\alpha'\beta'}_v\|}\right\rangle .
\]
Thus this profile is formed from the fixed kinks $q_j$, $j\in V$, with their
slope jumps $d_j$; no kink is added or removed when the endpoints are
perturbed. There are constants
$\varkappa>0$, $\varsigma_0>0$ and $C_L<\infty$, depending only on $g$, such
that:
\begin{enumerate}
\item[\emph{(a)}] $\Lambda(\alpha,\beta) - |\bar D_{\alpha\beta}(v)| \ge
\varkappa (v - c(\alpha,\beta))^2$ for all $v \in (\alpha,\beta)$;
\item[\emph{(b)}] for $|\alpha'-\alpha| + |\beta'-\beta| < \varsigma_0$, the
profile $\bar D_{\alpha'\beta'}[V]$ has a unique maximizer $c'$ of its modulus,
with $|c' - c(\alpha,\beta)| \le C_L(|\alpha'-\alpha| + |\beta'-\beta|)$,
$\max_v |\bar D_{\alpha'\beta'}[V]| \ge \tfrac12\Lambda(\alpha,\beta)$,
\[
\max_v \big|\bar D_{\alpha'\beta'}[V]\big| -
\big|\bar D_{\alpha'\beta'}[V](v)\big| \ge \tfrac12 \varkappa (v-c')^2
\quad\text{on } (\alpha',\beta'),
\]
and $\min\{c'-\alpha', \beta'-c'\} \ge \tfrac12 \varrho_0(\beta'-\alpha')$.
\end{enumerate}
\end{lemma}

\begin{proof}
(a) Write $c=c(\alpha,\beta)$ and
$\Lambda=\Lambda(\alpha,\beta)$. Near $c$, Taylor expansion and
\assnondeg{} give
\[
\Lambda-|\bar D_{\alpha\beta}(v)|
\ge \tfrac13|\partial_v^2|\bar D_{\alpha\beta}|(c)|(v-c)^2
\]
on some $(c-\epsilon,c+\epsilon)$. Moreover,
$|\bar D_{\alpha\beta}(v)|\to0$ at both endpoints. Consequently the function
$\Lambda-|\bar D_{\alpha\beta}(v)|$ has a continuous extension to
$[\alpha,\beta]$, obtained by assigning it the value $\Lambda$ at each
endpoint. By uniqueness of $c$, this extension is strictly positive on
\[
K_{\epsilon}:=\{v\in[\alpha,\beta]:|v-c|\ge\epsilon\}.
\]
It therefore has a positive minimum on $K_\epsilon$. Combining this minimum
with the local Taylor bound proves (a) for the fixed node. Finally, take the
smallest resulting positive constant over the finitely many nodes
$(\alpha,\beta)\in\mathcal T(g)$ with $V(\alpha,\beta)\ne\emptyset$.

(b) Here $C^p$ means that derivatives through order $p$ exist and are
continuous. With $V$ fixed, the map
$(\alpha',\beta',v)\mapsto\bar D_{\alpha'\beta'}[V](v)$ is continuous, is
$C^2$ in $v$, and has a jointly continuous second derivative in $v$, by
Lemma~\ref{lem:gram} and continuity of
$\Pi^\perp_{\alpha'\beta'}g$ in the endpoints. Since
$|\bar D_{\alpha\beta}(c)|=\Lambda>0$, its sign is constant in a neighbourhood
of $(\alpha,\beta,c)$, so the same regularity holds for the absolute profile.

Apply the implicit function theorem to
\[
H(\alpha',\beta',v):=\partial_v
|\bar D_{\alpha'\beta'}[V](v)|.
\]
At $(\alpha,\beta,c)$, $H=0$ and
$\partial_vH=\partial_v^2|\bar D_{\alpha\beta}|(c)<0$. The theorem therefore
provides neighbourhoods of $(\alpha,\beta)$ and $c$, and a unique $C^1$
function $(\alpha',\beta')\mapsto c'(\alpha',\beta')$ in those neighbourhoods
such that $H(\alpha',\beta',c')=0$ and $c'(\alpha,\beta)=c$. A $C^1$ function
on a sufficiently small compact neighbourhood is Lipschitz, which gives the
stated bound with a finite constant $C_L$.

We now choose $\varsigma_0>0$ small enough that every endpoint pair satisfying
$|\alpha'-\alpha|+|\beta'-\beta|<\varsigma_0$ lies in this neighbourhood. If
necessary, choosing a smaller positive value of $\varsigma_0$ ensures, by the
joint continuity just established, that the second derivative at $c'$ remains
negative, the maximum value is at least $\Lambda/2$, and $c'$ remains at least
half the relative endpoint distance attained by $c$. Part (a) gives
\[
\Lambda-|\bar D_{\alpha\beta}(v)|\ge\varkappa\epsilon^2
\quad\text{when }|v-c|\ge\epsilon.
\]
Uniform continuity of the absolute profile on the corresponding compact set
of endpoint and candidate triples preserves at least half this positive
difference for all such $(\alpha',\beta')$. Thus no point outside the local
neighbourhood of $c'$ can maximize the perturbed profile. Inside that
neighbourhood, the uniform negative-curvature bound and Taylor's theorem give
the quadratic inequality in (b), after reducing the common constant
$\varkappa$ if required. This also proves that $c'$ is the unique global
maximizer. Taking the smallest admissible $\varsigma_0$ and the largest
$C_L$ over the finitely many non-empty nodes makes the constants depend only
on $g$.
\end{proof}

\subsection{The sample recursion on tracking windows}
\label{sup:track_sec}

We work on the event $E_n$ of Lemma~\ref{lem:noise}.

\begin{defin}[Tracking]
\label{def:track}
A window $(s,e]$ \emph{tracks} a node $(\alpha,\beta) \in \mathcal{T}(g)$ at
scale $C_\star$ if $|s - n\alpha| \le C_\star \bar r_n$,
$|e - n\beta| \le C_\star\bar r_n$, and moreover $s = 0$ when $\alpha = 0$ and
$e = n$ when $\beta = 1$.
Thus the scale is the fixed constant multiplying the endpoint error
$\bar r_n$. For example, tracking at scale $K$ means that both endpoint
errors are at most $K\bar r_n$, subject to the exact boundary conditions just
stated.
\end{defin}

\begin{lemma}[Tracking separates interior and exterior kinks]
\label{lem:track}
Fix $C_\star$. There is $n_0$ such that for $n \ge n_0$, if $(s,e]$ tracks
$(\alpha,\beta)$ at scale $C_\star$ then
\[
x_j(s,e) \ge \tfrac12\kappa n \ \text{ for } j \in V(\alpha,\beta),
\qquad
x_j(s,e) \le C_\star\bar r_n + 1 \ \text{ for } j \notin V(\alpha,\beta).
\]
Hence the kinks in $V(\alpha,\beta)$ have signal strength, as defined in
Section~\ref{sup:toolkit}, at least
\[
c_\ell\underline d(\kappa/2)^{3/2}n^{1/2}.
\]
Moreover,
\[
\varrho:=\sum_{j\notin V(\alpha,\beta)}\ell_j(s,e)
=O\{n^{-1/4}(\log n)^{3/4}\}=o(1).
\]
\end{lemma}

\begin{proof}
If $j \in V(\alpha,\beta)$ then $q_j - \alpha \ge \kappa$ and
$\beta - q_j \ge \kappa$ by \eqref{eq:kappa}, so
$\tau_j - s \ge \kappa n - C_\star\bar r_n - 1 \ge \tfrac12\kappa n$ for $n$
large, and symmetrically on the right. If $j \notin V(\alpha,\beta)$ and
$q_j \in \{\alpha,\beta\}$ then $x_j \le C_\star\bar r_n + 1$. Otherwise
$q_j < \alpha$ or $q_j > \beta$ strictly, and in the first case
$\alpha - q_j \ge \kappa$, because $\alpha$ is $0$, or one of the $q_l$, or a
value $c(\cdot,\cdot)$ at an ancestor node. The definition \eqref{eq:kappa}
or the minimum kink spacing $\delta_{\min}$ then gives
$\alpha-q_j\ge\kappa$; hence $\tau_j<s$ and
$x_j = 0$. The case $q_j > \beta$ is symmetric. The signal strength statements are
Corollary~\ref{cor:gap}.
\end{proof}

\begin{lemma}[No detection]
\label{lem:nodet}
Fix $C_\star$. There is $n_0$ such that for $n \ge n_0$, on $E_n$: if $(s,e]$
is any window with $x_j(s,e) \le C_\star\bar r_n + 1$ for every $j$, then
$\max_b \mathcal{C}_{s,e}(b) \le o(1) + \sigma(8\log n)^{1/2} < \lambda$, and
the call on $(s,e]$ returns $\emptyset$.
\end{lemma}

\begin{proof}
By Lemma~\ref{lem:rep}(b) and Corollary~\ref{cor:gap}(a),
\[
\max_b \big|\langle f,\hat\psi_b\rangle\big| \;\le\; \sum_j \ell_j(s,e) \;=\; O\big\{n^{-1/4}(\log n)^{3/4}\big\} \;\to\; 0 .
\]
Adding $Z_n$ and recalling that $\lambda = \Theta\sigma(8\log n)^{1/2}$ with $\Theta > 1$ gives the strict inequality for $n$ large, since $(\Theta-1)\sigma(8\log n)^{1/2} \to \infty$.
If $B_{s,e} = \emptyset$ the call returns $\emptyset$ immediately.
\end{proof}

\begin{lemma}[Detection and localization]
\label{lem:loc}
Let $(s,e]$ be a window and put
\[
D_U(w):=\sum_{j\in U}\Delta_j\langle\psi_{\tau_j},w\rangle,
\qquad D_U(b):=D_U(\hat\psi_b),
\qquad
\varrho := \sum_{j\notin U}\ell_j(s,e) ,
\]
for vectors $w\in\mathbb R^W$ and some
$U\subseteq\{1,\ldots,N\}$. Suppose there are $b^* \in B_{s,e}$ and constants $a_0, \varkappa', \varrho_1 > 0$ such that
\begin{enumerate}
\item[(i)] $\min(b^*-s,\, e-b^*) \ge \varrho_1 n$;
\item[(ii)] $A := |D_U(b^*)| \ge a_0 n^{1/2}$, and
$|D_U(b^*)| - |D_U(b)| \ge \varkappa' A \{(b-b^*)/n\}^2$
for every $b \in B_{s,e}$ with $|b-b^*| \ge r_0$, where $r_0 \le C_r n^{1/2}$;
\item[(iii)] $\varrho \le \sigma(\log n)^{1/2}$.
\end{enumerate}
Then there are $n_0$ and $C_1$, depending only on $a_0,\varkappa',\varrho_1,C_r,\sigma,c_0,C_0$, such that for $n \ge n_0$ on
$E_n$,
\[
\max_b \mathcal{C}_{s,e}(b) \ge \tfrac12 A > \lambda,
\qquad
|\hat b - b^*| \le C_1 \bar r_n .
\]
\end{lemma}

\begin{proof}
Define the linear functional
\[
R(w):=\sum_{j\notin U}\Delta_j\langle\psi_{\tau_j},w\rangle.
\]
Then $\langle f,w\rangle=D_U(w)+R(w)$,
$|R(w)|\le\varrho\|w\|$, and
$|R(\hat\psi_b)-R(\hat\psi_{b'})|
\le\varrho\|\hat\psi_b-\hat\psi_{b'}\|$ by Lemma~\ref{lem:rep}(b). Since $A \ge a_0n^{1/2}$ and
$\varrho + Z_n = O\{(\log n)^{1/2}\}$ we get
$\mathcal{C}(b^*) \ge A - \varrho - Z_n \ge \tfrac12 A > \lambda$, the first
claim.

If $|\hat b - b^*| < r_0$ the second claim holds, since $r_0 \le C_r n^{1/2} \le C_r \bar r_n$; assume therefore $|\hat b - b^*| \ge r_0$, so that hypothesis (ii) may be used at $b = \hat b$. Put $\varsigma^* = \mathrm{sign}\, D_U(b^*)$, $u^* = \varsigma^*\hat\psi_{b^*}$,
and $u = \varsigma \hat\psi_{\hat b}$ with $\varsigma$ chosen so that
$\langle X,u\rangle = \mathcal{C}(\hat b) \ge 0$. From
$\langle X,u\rangle \ge \langle X,u^*\rangle$ and Lemma~\ref{lem:noise},
\begin{equation}
\label{eq:basic}
\langle f,u^*\rangle - \langle f,u\rangle \;\le\;
\langle \varepsilon, u - u^*\rangle .
\end{equation}
Since $D_U(u^*)=|D_U(b^*)|$ by the choice of $\varsigma^*$, and
$D_U(u)\le|D_U(\hat b)|$,
\[
\langle f,u^*\rangle-\langle f,u\rangle
\ge |D_U(b^*)|-|D_U(\hat b)|-|R(u^*-u)|
\ge |D_U(b^*)|-|D_U(\hat b)|-\varrho\|u-u^*\|.
\]
Combining this inequality with \eqref{eq:basic} and hypothesis (ii), and
writing $\Delta:=|\hat b-b^*|$, gives
\begin{equation}
\label{eq:gapineq}
\varkappa' A(\Delta/n)^2
\le |D_U(b^*)|-|D_U(\hat b)|
\le \varrho\|u-u^*\|+\langle\varepsilon,u-u^*\rangle.
\end{equation}

The noise term is treated differently before and after Step~2 proves that
$\varsigma=\varsigma^*$. The increment bound is the second event in the
definition of $E_n$; it controls
$\langle\varepsilon,\hat\psi_b-\hat\psi_{b'}\rangle$, a difference of
unsigned contrast vectors. If $\varsigma=-\varsigma^*$, then
$u-u^*=\pm(\hat\psi_{\hat b}+\hat\psi_{b^*})$ is instead a sum, so the second
event does not apply. Before Step~2 excludes this case, we use the pointwise
part of $E_n$, namely its first event, separately at $\hat b$ and $b^*$:
\[
\big| \langle\varepsilon, u-u^*\rangle \big| \;\le\; \big|\langle\varepsilon,u\rangle\big| + \big|\langle\varepsilon,u^*\rangle\big| \;\le\; 2Z_n ,
\]
This inequality holds for either relation between the signs. After Step~2 has
proved $\varsigma=\varsigma^*$,
$u-u^*=\pm(\hat\psi_{\hat b}-\hat\psi_{b^*})$, and the second event in $E_n$
can be used. Lemma~\ref{lem:sep} then makes both terms on the right of
\eqref{eq:gapineq} proportional to $\Delta/n$, whereas its left-hand side is
proportional to $(\Delta/n)^2$. For $\Delta>0$, one factor $\Delta/n$ can then
be canceled. The bound $|R(w)|\le\varrho\|w\|$ holds for every $w$ and does
not depend on the relation between the signs.

\emph{Step 1 (preliminary bound).} Using $\|u-u^*\| \le 2$, the pointwise bound just described, and (iii), $\varkappa' A(\Delta/n)^2 \le 2\varrho + 2Z_n \le 2\sigma(\log n)^{1/2} + 2\sigma(8\log n)^{1/2} \le 8\sigma(\log n)^{1/2}$, so $\Delta/n = O\{n^{-1/4}(\log n)^{1/4}\} \to 0$. The second, increment event in $E_n$ is not used at this stage.

\emph{Step 2 (sign of the inner product).} By Step~1 and Lemma~\ref{lem:sep} with $x = \min(b^*-s,e-b^*) \ge \varrho_1 n$ we get $\rho := \rho(b^*,\hat b) \to 1$. Suppose $\langle u,u^*\rangle = -\rho$. Then $\|u+u^*\|^2 = 2-2\rho \to 0$, while $\|\Pi^\perp_W f\| \le \sum_j \ell_j \le N\bar d\, n^{1/2}/\sqrt2$, so
$|\langle f,u\rangle + \langle f,u^*\rangle| \le Cn^{1/2}\sqrt{2-2\rho} =
o(A)$. Since $\langle f,u^*\rangle \ge A - \varrho \ge \tfrac34 A$, this forces
$\langle f,u\rangle \le -\tfrac12 A < 0$; but
$\langle f,u\rangle \ge \mathcal{C}(\hat b) - Z_n \ge \tfrac12 A - Z_n > 0$, a
contradiction. Hence $\langle u,u^*\rangle = \rho > 0$. This does not yet give $\varsigma = \varsigma^*$: a positive signed correlation could equally arise from opposite multipliers applied to negatively correlated vectors. What excludes that is Lemma~\ref{lem:gramdet}(c), by which $\langle\hat\psi_{\hat b},\hat\psi_{b^*}\rangle > 0$ for any two admissible candidates; since $\langle u,u^*\rangle = \varsigma\varsigma^*\langle\hat\psi_{\hat b},\hat\psi_{b^*}\rangle$, it follows that $\varsigma = \varsigma^*$. Consequently $\|u-u^*\| = \{2(1-\rho)\}^{1/2} \le (2C_0)^{1/2}\varrho_1^{-1}\Delta/n$.

\emph{Step 3 (conclusion).} The signs now agree, so
$u-u^*=\pm(\hat\psi_{\hat b}-\hat\psi_{b^*})$ and the increment event in
$E_n$ gives
$\langle\varepsilon,u-u^*\rangle
\le\sigma(10\log n)^{1/2}\|u-u^*\|$. Substituting this and the bound of
Step~2 into \eqref{eq:gapineq} gives
\[
\varkappa' A\Big(\frac{\Delta}{n}\Big)^2
\le \big\{\varrho+\sigma(10\log n)^{1/2}\big\}
(2C_0)^{1/2}\varrho_1^{-1}\frac{\Delta}{n}.
\]
For $\Delta>0$, canceling one factor of $\Delta/n$ gives
\[
\frac{\Delta}{n}
\le\frac{(2C_0)^{1/2}}{\varrho_1\varkappa' A}
\big\{\varrho+\sigma(10\log n)^{1/2}\big\};
\]
the conclusion is immediate when $\Delta=0$. By (iii), the quantity in braces
is at most $5\sigma(\log n)^{1/2}$, and $A\ge a_0n^{1/2}$. Therefore
$\Delta\le C_1\bar r_n$ with
\[
C_1=\max\{C_r,\ 5(2C_0)^{1/2}\sigma/(\varrho_1\varkappa'a_0)\}.
\]
The term $C_r$ covers the case $\Delta<r_0$ considered before
\eqref{eq:basic}; the second term covers $\Delta\ge r_0$.
\end{proof}

\begin{remark}
\label{rem:fixedpoint}
Within Lemma~\ref{lem:loc}, $\varrho$ is fixed once the window $(s,e]$ and the
set $U$ are fixed; it does not depend on the localization error
$\Delta=|\hat b-b^*|$. Consequently, the final bound on $\Delta$ cannot be
substituted into the definition
$\varrho=\sum_{j\notin U}\ell_j(s,e)$ to obtain a smaller value of $\varrho$.
The rate $\bar r_n$ follows directly because both
$R(u^*-u)$ and the noise increment
$\langle\varepsilon,u-u^*\rangle$ are bounded by a constant times
$\|u-u^*\|$, and hence by a constant times $\Delta/n$. By contrast, the
noiseless contrast loss on the left of \eqref{eq:gapineq} is bounded below by
a constant times $(\Delta/n)^2$.

If $|R(u^*)-R(u)|$ were instead bounded by $2\varrho$, without its dependence
on $u^*-u$, the resulting calculation would give only
$\Delta=O\{n^{5/8}(\log n)^{3/8}\}$. Corollary~\ref{cor:gap} then shows that a
kink at this distance from the child-window edge created by its own detection
has signal strength of order $n^{-1/16}$, apart from logarithmic factors. This
does not tend to zero quickly enough for the induction in
Proposition~\ref{prop:track_prop}.
\end{remark}

\begin{lemma}[Transfer to the sample window]
\label{lem:transfer}
Let $(\alpha,\beta)$ be a fixed rescaled window with
$\beta-\alpha\ge\kappa$ and $V:=V(\alpha,\beta)\ne\emptyset$, at which the
conclusions of Lemma~\ref{lem:conseq} hold with constants
$\varkappa,\varrho_0$ and maximizer $c$. Let $(s,e]$ track this window at scale
$C_\star$ in the sense of Definition~\ref{def:track}. There is $n_0$ such that
for $n\ge n_0$, hypotheses (i)--(iii) of Lemma~\ref{lem:loc} hold with
\[
U=V,\quad \varrho_1=\varrho_0\kappa/4,\quad
\varkappa'=\varkappa/(32\Lambda_{\max}),\quad
a_0=\Lambda_{\min}/4,\quad r_0=C_rn^{1/2}
\]
for a constant $C_r$. Moreover, $b^*$ may be chosen so that
\[
|b^*-nc|\le(C'_0+C'_LC_\star)\bar r_n.
\]
Here $C'_0$ does not depend on $n$ or $C_\star$, although it may depend on the
signal and on the fixed window through the constants of
Lemma~\ref{lem:conseq}. The additive term $C'_0\bar r_n$ cannot in general be
omitted. At the root, $C_\star=0$ because $(0,n]$ tracks $(0,1)$ exactly, but
$b^*$ is an integer whereas $nc$ need not be an integer. This conclusion
applies to every node of $\mathcal T(g)$ under \assnondeg{}. It also applies,
in Section~\ref{sup:gridtheory}, to every grid sub-window for which
Assumption~\ref{ass:gnondeg} states that the absolute population profile has a
unique maximizer with a strictly negative second derivative there.
\end{lemma}

\begin{proof}
Write $\alpha_n = s/n$, $\beta_n = e/n$, so
$|\alpha_n-\alpha| + |\beta_n-\beta| \le 2C_\star\bar r_n/n \to 0$. This
quantity is less than $\varsigma_0$ for all sufficiently large $n$, so
Lemma~\ref{lem:conseq}(b) applies. By
Lemma~\ref{lem:rep}(a) restricted to $V$, together with Lemma~\ref{lem:gram}
and Lemma~\ref{lem:weighted}, proved in Section~\ref{sup:discproof},
\begin{equation}
\label{eq:riemann}
\Big|D_V(b)-n^{1/2}\bar D_{\alpha_n\beta_n}[V](b/n)\Big|
\le C\frac{m}{n}x_b(s,e)^{-1/2}
\qquad\text{for every }b\in B_{s,e},
\end{equation}
where $x_b(s,e)=\min(b-s,e+1-b)$. Two consequences will be used. The error
on the left of \eqref{eq:riemann} is $O(1)$ uniformly in $b$. For each fixed
$\eta>0$, it also satisfies the sharper interior bound
\begin{equation}
\label{eq:riemann2}
\sup_{b:\,x_b(s,e)\ge\eta m}
\Big|D_V(b)-n^{1/2}\bar D_{\alpha_n\beta_n}[V](b/n)\Big|
=O(n^{-1/2})
\end{equation}
when $m\asymp n$.

The quantities in \eqref{eq:riemann} are algebraic functions of $m$ and the
integer locations: they are formed from these arguments by finitely many
arithmetic operations and square roots. The square roots arise from the
normalization by $\|\psi_b\|$. Section~\ref{sup:discproof} proves the estimate
from Lemma~\ref{lem:weighted}. The error on the left of \eqref{eq:riemann}
accounts for two differences: the discrete Gram entries differ from $m^3$
times the continuum entries $G(u,v)$ and $N(v)$ defined in
Section~\ref{sup:discproof}, and the integer location
$\tau_j=\lfloor nq_j\rfloor$ differs from $nq_j$ by at most one. There is no
additional error from the slope jumps, because the triangular-array regime of
Section~4.1 of the main paper, meaning the sequence of signal models indexed by
$n$, prescribes $\Delta_j=d_j/n$ exactly.

For a more general sequence of signal models, the limits
$\tau_j/n\to q_j$ and $n\Delta_j\to d_j$ alone do not imply an $o(1)$ error in
\eqref{eq:riemann}. Moving a kink by $h$ changes $D_V$ by order
$hm^{1/2}/n$, and changing $n\Delta_j$ by $\epsilon$ changes it by order
$\epsilon n^{1/2}$. Thus an $o(1)$ error requires
$|\tau_j-nq_j|=o(n^{1/2})$ and
$|n\Delta_j-d_j|=o(n^{-1/2})$. A location error of order $\bar r_n$ does not
satisfy the first requirement and therefore does not ensure an $o(1)$
discretization error.

The dependence on $x_b$ in \eqref{eq:riemann} is necessary: the rate $O(n^{-1/2})$ does not hold uniformly, and fails at candidates within $O(1)$ of an edge, where the discrete and continuum norms differ by a constant factor rather than by a relative $O(1/m)$. At $b = s+2$, \eqref{eq:normexact} gives $\|\psi_b\|^2 = (m-1)(m-2)/\{m(m+1)\} \to 1$, whereas $m^{3/2}\|\Psi_{2/m}\| \to 2^{3/2}3^{-1/2}$. The form \eqref{eq:riemann} is sharp.
Let $c_n$ maximize $|\bar D_{\alpha_n\beta_n}[V]|$ and $b^*$ maximize
$|D_V|$ over $B_{s,e}$. Note that $\bar D_{\alpha\beta}$ is defined through the unnormalized inner product of $L^2(\alpha,\beta)$, so it already includes a factor $(\beta-\alpha)^{3/2}$ and the normalization in \eqref{eq:riemann} is the same for every window; this is what makes contrasts computed on windows of different lengths directly comparable. Then \eqref{eq:riemann} and Lemma~\ref{lem:conseq}(b) give $|D_V(b^*)| \ge \tfrac14\Lambda_{\min}n^{1/2}$. Fix $\eta := \varrho_0/4$ and call $b$ \emph{interior} if it is at distance at least $\eta m$ from both ends of the window; on interior $b$ the error is $O(n^{1/2}m^{-1}) = O(n^{-1/2})$ by \eqref{eq:riemann2}, and on all $b$ it is $O(n^{1/2}m^{-1/2}) = O(1)$ by \eqref{eq:riemann}.

\emph{The discrete maximizer is close to $nc_n$.} Let $b_n$ be the integer
nearest $nc_n$, which is interior. By \eqref{eq:riemann2} and the quadratic
bound of Lemma~\ref{lem:conseq}(b) for the profile
$|\bar D_{\alpha_n\beta_n}[V]|$ at its maximizer $c_n$,
\[
|D_V(b_n)|=n^{1/2}\max_w|\bar D_{\alpha_n\beta_n}[V](w)|+o(1).
\]
First, $b^*$ is interior. A non-interior $b$ lies within
$\eta(\beta-\alpha)$ of an endpoint, whereas $c_n$ is at least
$\tfrac12\varrho_0(\beta-\alpha)$ from each endpoint by the final inequality
in Lemma~\ref{lem:conseq}(b). Hence
$|b/n-c_n|\ge\tfrac14\varrho_0(\beta-\alpha)$. Comparison must be with the
perturbed maximum $\max_w|\bar D_{\alpha_n\beta_n}[V](w)|$, not with
$\Lambda(\alpha,\beta)$: the rescaled endpoints move by $O(\bar r_n/n)$, so
the two maxima differ by $O(\bar r_n/n)$. After multiplication by $n^{1/2}$,
this difference is $O\{(\log n)^{1/2}\}$ rather than $O(1)$.

For a non-interior $b$, the deficit
\[
\max_w|\bar D_{\alpha_n\beta_n}[V](w)|
-|\bar D_{\alpha_n\beta_n}[V](b/n)|
\]
is bounded below by a fixed constant $\delta_0>0$, uniformly for large $n$.
This follows from Lemma~\ref{lem:conseq}(b), the fixed lower bound on
$|b/n-c_n|$, and continuity in the endpoints. Equation~\eqref{eq:riemann}
therefore gives
\[
|D_V(b)|\le n^{1/2}
\{\max_w|\bar D_{\alpha_n\beta_n}[V](w)|-\delta_0\}+O(1),
\]
which is smaller than $|D_V(b^*)|$ for large $n$. Second, compare $b^*$ with
$b_n$. Both are interior, so the $O(n^{-1/2})$ approximation error in
\eqref{eq:riemann2} applies at both points. Hence
\[
\tfrac12\varkappa\, n^{1/2}\big( b^*/n - c_n \big)^2 \;\le\; |D_V(b_n)| - |D_V(b^*)| + O(n^{-1/2}) \;\le\; O(n^{-1/2}),
\]
since $b^*$ maximizes $|D_V|$. It follows that
$|b^*-nc_n|\le C_b n^{1/2}$ for a constant $C_b<\infty$.

\emph{Transfer of the profile deficit to $b^*$.} Let
$r_0:=2C_bn^{1/2}=:C_rn^{1/2}$ and suppose $|b-b^*|\ge r_0$. Then
$|b-nc_n|\ge|b-b^*|-C_bn^{1/2}\ge\tfrac12|b-b^*|$, so
$(b/n-c_n)^2\ge\tfrac14\{(b-b^*)/n\}^2$. For interior $b$, the continuum
profile deficit provided by Lemma~\ref{lem:conseq}(b), after multiplication by
$n^{1/2}$, is at least
$\tfrac12\varkappa n^{1/2}(b/n-c_n)^2$. By choosing $C_b$ large enough, this
lower bound exceeds twice the two $O(n^{-1/2})$ discretization errors from
\eqref{eq:riemann2}. For non-interior $b$, the continuum profile deficit is of
order $n^{1/2}$ and exceeds the $O(1)$ error in \eqref{eq:riemann}. Thus, in
both ranges,
\[
|D_V(b^*)| - |D_V(b)| \;\ge\; \tfrac{\varkappa}{16}\, n^{1/2} \big\{(b-b^*)/n\big\}^2
\qquad\text{whenever } |b-b^*| \ge r_0 ,
\]
which is hypothesis (ii) in the form required by Lemma~\ref{lem:loc} with
$\varkappa'=\varkappa/(32\Lambda_{\max})$. The additional factor two absorbs
the $\{1+o(1)\}$ term in
$|D_V(b^*)|\le\Lambda_{\max}n^{1/2}\{1+o(1)\}$. Finally,
Lemma~\ref{lem:conseq}(b) gives
$|c_n-c|\le2C_LC_\star\bar r_n/n$, and hence
\[
|b^*-nc|\le C_bn^{1/2}+2C_LC_\star\bar r_n
\le(C'_0+C'_LC_\star)\bar r_n,
\]
with $C'_0:=C_b$ and $C'_L:=2C_L$. Here we used
$n^{1/2}\le\bar r_n$. Also,
$|b^*-nc_n|=O(n^{1/2})=o(m)$, while Lemma~\ref{lem:conseq}(b) places $c_n$
at least $\tfrac12\varrho_0(\beta_n-\alpha_n)$ from each endpoint. Therefore
$\min(b^*-s,e-b^*)\ge\tfrac14\varrho_0m
\ge\tfrac14\varrho_0\kappa n$ for large $n$, which verifies hypothesis (i).
Hypothesis (iii) follows from Lemma~\ref{lem:track}.
\end{proof}

\begin{defin}[Node tracking constants]
\label{def:radii}
For $v\in\mathcal T(g)$, let $C_{1,v}$ be the localization constant provided
by Lemma~\ref{lem:loc} at $v$. Let $C'_0$ and $C'_L$ be the largest of the
corresponding constants in Lemma~\ref{lem:transfer} over the finitely many
non-empty nodes. Set $K_v:=0$ at the root and, for each child $v'$ of $v$, set
\[
K_{v'} \;:=\; \max\big\{ K_v, \ C_{1,v} + C'_0 + C'_L K_v \big\} ,
\qquad\text{and}\qquad
C_\star \;:=\; \max\Big\{ \max_{v \in \mathcal{T}(g)} K_v, \ C_{\mathrm{rt}} \Big\} ,
\]
where $C_{\mathrm{rt}}$ is the largest localization constant provided by
Lemma~\ref{lem:loc} at the finitely many Case~B re-test windows in the proof of
\proppop{}. On each such window, $U$ is a singleton, meaning that
$U=\{j^*\}$ contains only the unique kink strictly inside the re-test window.
The quantities $K_v\bar r_n$ are called the node tracking radii because they
bound the two endpoint errors of a sample window associated with node $v$.
Thus $K_v$ is the dimensionless tracking constant and $K_v\bar r_n$ is the
corresponding radius. Both maxima are over finite sets, so
$C_\star<\infty$ by \assterm{}.

The re-test constant is included separately because a Case~B re-test window is
generally not a node of $\mathcal T(g)$, so its localization constant is not
among the $C_{1,v}$. Also $K_v\le C_\star$ for every $v$, and $C_{1,v}$ does
not depend on a tracking constant $K$. Indeed, Lemma~\ref{lem:transfer} shows
that the parameters $a_0,\varkappa',\varrho_1,r_0$ used in
Lemma~\ref{lem:loc} are determined by
$\kappa,\varkappa,\Lambda_{\min},\Lambda_{\max},\varrho_0$. The tracking
constant enters only through the requirement
$2K\bar r_n/n<\varsigma_0$, which holds for all large $n$ for each of the
finitely many constants $K_v$.
\end{defin}

The constants are allowed to increase with depth. Suppose a sample window has
endpoint errors at most $K_v\bar r_n$ relative to node $v$. By
Lemma~\ref{lem:transfer}, the maximizer of its noiseless discrete profile is at
most $(C'_0+C'_LK_v)\bar r_n$ from the population split
$n c(\alpha,\beta)$. Lemma~\ref{lem:loc} places the maximizer of the observed
contrast at most a further $C_{1,v}\bar r_n$ away. Each child window of the
sample recursion retains one endpoint of its parent and uses this observed
maximizer as its other endpoint. Its two endpoint errors are therefore bounded
by $K_{v'}\bar r_n$, with $K_{v'}$ as defined above. The constants $K_v$ need
not decrease from parent to child. Finiteness of $\mathcal T(g)$ is enough:
the recursively defined constants remain finite and their maximum is
$C_\star$.

\begin{prop}[Recursion on tracking windows]
\label{prop:track_prop}
Assume \assterm{} and \assnondeg{}. There is $n_0$ such
that for $n \ge n_0$, on $E_n$, the following holds for every
$v = (\alpha,\beta) \in \mathcal{T}(g)$ and every window $(s,e]$ tracking it at
scale $K_v$: the call $\textnormal{\textsc{LABS}}(X,s,e,\lambda)$ returns a set
$\hat{\mathcal{T}}$ with $|\hat{\mathcal{T}}| = |V(\alpha,\beta)|$ and, writing
$V(\alpha,\beta) = \{j_1 < \cdots < j_k\}$ and the elements of
$\hat{\mathcal{T}}$ in increasing order as
$\hat\tau^{(1)} < \cdots < \hat\tau^{(k)}$,
\[
|\hat\tau^{(i)} - \tau_{j_i}| \;\le\; C_\star\bar r_n, \qquad i = 1,\ldots,k .
\]
\end{prop}

\begin{proof}
The radii of Definition~\ref{def:radii} are fixed first and $n_0$ afterwards. We induct on the height of the subtree rooted at $(\alpha,\beta)$, finite by
\assterm{}. Write $V := V(\alpha,\beta)$ and $k := |V|$.

\emph{Case $k=0$.} By Lemma~\ref{lem:track} every kink has
$x_j(s,e) \le C_\star\bar r_n + 1$, so Lemma~\ref{lem:nodet} returns
$\emptyset$.

\emph{Case $k=1$, $V = \{j\}$.} Here $c(\alpha,\beta) = q_j$ exactly, and it is
simplest to apply Lemma~\ref{lem:loc} directly with $U = \{j\}$ and
$b^* = \tau_j$: (i) holds with $\varrho_1 = \kappa/2$ by
Lemma~\ref{lem:track}; (ii) holds with $A = \ell_j \ge c_\ell\underline d(\kappa/2)^{3/2}n^{1/2}$ and $\varkappa' = c_0$, since by
Lemma~\ref{lem:rep}(c) and Lemma~\ref{lem:sep},
\[
|D_U(\tau_j)| - |D_U(b)| = \ell_j\{1-\rho(\tau_j,b)\}
\ \ge\ \ell_j c_0 \min\{ ((b-\tau_j)/x_j)^2, 1\} \ \ge\ \ell_j c_0 \big((b-\tau_j)/n\big)^2, \] using $x_j \le n$ and $|b-\tau_j| \le n$; and (iii) is Lemma~\ref{lem:track}.
Lemma~\ref{lem:loc} now shows that the proposal statistic
$\max_{b\in B_{s,e}}\mathcal C_{s,e}(b)$ exceeds $\lambda$, so the proposal
test on $(s,e]$ is accepted, and
$|\hat b-\tau_j|\le C_{1,v}\bar r_n$. For the left child, the endpoint $s$ is
unchanged from the parent and still satisfies
$|s-n\alpha|\le K_v\bar r_n$, while its new right endpoint satisfies
$|\hat b-nq_j|\le C_{1,v}\bar r_n$. Similarly, the right child retains $e$
and has new left endpoint $\hat b$. Hence $(s,\hat b]$ and $(\hat b,e]$ track
$(\alpha,q_j)$ and $(q_j,\beta)$ at scale
$\max\{K_v,C_{1,v}\}\le K_{v'}$. Neither population child contains a kink
strictly inside it, so both live-kink sets are empty and both recursive calls
return $\emptyset$. The look-ahead branch is not entered, and the parent call
returns $\{\hat b\}$, with $|\hat b-\tau_j|\le C_\star\bar r_n$.

\emph{Case $k \ge 2$.} Lemma~\ref{lem:transfer}, applied at scale $K_v$,
provides a $b^*$ with
$|b^*-n\,c(\alpha,\beta)|\le(C'_0+C'_LK_v)\bar r_n$ and verifies hypotheses
(i)--(iii) of Lemma~\ref{lem:loc}, which then gives
$|\hat b-b^*|\le C_{1,v}\bar r_n$. Hence
\[
\big| \hat b - n\,c(\alpha,\beta) \big| \;\le\; \big( C_{1,v} + C'_0 + C'_L K_v \big)\bar r_n \;\le\; K_{v'}\bar r_n ,
\]
and since the other endpoint of each child is inherited unchanged, at scale $K_v \le K_{v'}$, the children $(s,\hat b]$ and $(\hat b,e]$ track $(\alpha,c)$ and $(c,\beta)$ at scale $K_{v'}$. The induction hypothesis applies to both, giving $|\mathcal{L}| = |V_L|$ and $|\mathcal{R}| = |V_R|$ with each returned point within $C_\star\bar r_n$ of its kink. We follow the two cases of \proppop{}.

\emph{Case A: $c \notin \{q_l\}$.} Then $|V_L|+|V_R| = k \ge 2$ and the
look-ahead branch is entered. Every kink has $x_j(s',e') \le C_\star\bar r_n+1$:
if both children are non-empty, $s'$ and $e'$ lie within $C_\star\bar r_n$ of
$\tau_{j_a}$ and $\tau_{j_b}$ for consecutive kinks $q_{j_a} < q_{j_b}$ of $g$,
so no kink lies strictly between and all others are outside $(s',e']$ by the
$\kappa$-separation of \eqref{eq:kappa}; if $\mathcal{L} = \emptyset$ then
$s' = s$ and $e'$ is within $C_\star\bar r_n$ of $\tau_{j_1}$, and no kink lies
in $(\alpha, q_{j_1})$; the case $\mathcal{R} = \emptyset$ is symmetric. By
Lemma~\ref{lem:nodet}, the maximum contrast in the re-test window is below
$\lambda$. The re-test therefore rejects the parent candidate and the call returns
$\mathcal{L}\cup\mathcal{R}$, of size $k$. Order preservation holds because
$\tau_{j_b} - \tau_{j_a} \ge \delta_{\min}n - 1 > 2C_\star\bar r_n$ for $n$
large.

\emph{Case B: $c = q_{j^*}$, $j^* \in V$.} Then $|V_L|+|V_R| = k-1 \ge 1$ and
the look-ahead branch is entered. Exactly as in \proppop{}, $\tau_{j^*}$ is the
only kink at distance more than $C_\star\bar r_n + 1$ from both edges of
$(s',e']$, its neighbors in $\{q_j : j\in V\}\cup\{\alpha,\beta\}$ sitting at
the edges, and
\[
\min(\tau_{j^*} - s',\ e' - \tau_{j^*}) \ \ge\
n\min(\kappa,\delta_{\min}) - C_\star\bar r_n - 1 \ \ge\ \tfrac12\kappa n .
\]
In particular $e'-s' \ge \kappa n$, so $B_{s',e'} \ne \emptyset$ and the
re-test is performed; applying Lemma~\ref{lem:loc} to $(s',e']$ as in the case
$k=1$, with $U = \{j^*\}$ and $b^* = \tau_{j^*}$, the re-test statistic exceeds
$\lambda$ and $|\hat b' - \tau_{j^*}| \le C_\star\bar r_n$. The call returns
$\mathcal{L}\cup\{\hat b'\}\cup\mathcal{R}$, of size $k$, and order
preservation follows as in Case~A.
\end{proof}

\subsection{Proof of \mainthm{}}
\label{sup:thmproof}

Work on $E_n$, which has probability tending to one by Lemma~\ref{lem:noise},
and let $n \ge n_0$ with $n_0$ as in Proposition~\ref{prop:track_prop} and $C_\star$ as in Definition~\ref{def:radii}.

If $N \ge 1$, the initial call is on $(0,n]$, which tracks the root node $(0,1) \in \mathcal{T}(g)$ at scale $K_{(0,1)} = 0$, both endpoints being exact, and
$V(0,1) = \{1,\ldots,N\}$ since every kink is interior.
Proposition~\ref{prop:track_prop} applied at the root gives $\hat N = N$ and
$|\hat\tau_{(j)} - \tau_j| \le C_\star\bar r_n$ for every $j$, which is the
assertion with $C = C_\star$.

If $N = 0$ then $g$ is affine, meaning that it is a straight line, and no
window contains a kink. Lemma~\ref{lem:nodet} applies to $(0,n]$: the maximum
contrast is at most $Z_n<\lambda$, the root proposal test is rejected, and
$\hat{\mathcal{T}}_n=\emptyset$. \qed

\begin{remark}[Uniformity in the threshold]
\label{rem:uniformlam}
Fix $\eta>0$ and $\varsigma>0$ with
$\varsigma<\tfrac12\Lambda_{\mathrm{det}}$, or with
$\varsigma<\tfrac12\Lambda^{M,R}_{\mathrm{det}}$ for LABS-Grid. Conditional
on $E_n$, the argument above is deterministic once $\lambda$ is specified.
The threshold enters only through two comparisons. On a window with no visible
kink, the maximum contrast is at most $o(1)+Z_n$ and must be below $\lambda$.
On every non-empty population node used for a proposal, and on every Case~B
one-kink re-test window, the corresponding maximum population contrast is at
least $\tfrac12\Lambda_{\mathrm{det}}n^{1/2}$ and the sample statistic must
exceed $\lambda$. Both comparisons hold simultaneously for every
$\lambda\in[(1+\eta)Z_n,\varsigma n^{1/2}]$ once $n$ is large. Hence, on
$E_n$ and for $n\ge n_0(\eta,\varsigma)$, the conclusion of \mainthm{} holds
for all thresholds in this band simultaneously, including a data-dependent
threshold whose value lies in the band. The same statement holds for
Theorem~\ref{thm:grid}, with $\Lambda^{M,R}_{\mathrm{det}}$ from
\eqref{eq:LamdetMR} in place of $\Lambda_{\mathrm{det}}$.

It is $\Lambda_{\mathrm{det}}$ and not $\Lambda_{\min}$ that must appear here, because the Case~B re-test window is in general not a node of $\mathcal{T}(g)$ and its amplitude can be the smaller of the two. Take $q_1 = \tfrac12-\delta$, $q_2 = \tfrac12$, $q_3 = \tfrac12+\delta$ with $d_1=d_2=d_3=1$ and $\delta$ small. By symmetry the root maximizer is $q_2$, so the root split is Case~B; each child carries one live kink, of amplitude $\delta^{3/2}(1-2\delta)^{3/2}/\sqrt3$, and this is $\Lambda_{\min}$, the root amplitude being of order one. The re-test acts on $(q_1,q_3)$, where the amplitude is $\delta^{3/2}/(2^{3/2}\sqrt3)$, a fraction $\{2^{3/2}(1-2\delta)^{3/2}\}^{-1} \to 2^{-3/2} < \tfrac12$ of $\Lambda_{\min}$.  Any $\lambda = c\,n^{1/2}$ with the re-test amplitude below $c$ and $c < \tfrac12\Lambda_{\min}$ therefore lies in the band that $\Lambda_{\min}$ would permit, yet LABS recovers $q_1$ and $q_3$ and discards $q_2$.
\end{remark}

\section{Consistency of LABS-Grid}
\label{sup:gridtheory}

\mainthm{} concerns \algmain{}, which is LABS-Grid with $M = R = 2$. This section
extends it to arbitrary fixed $M \ge 2$ and $R \ge 2$. The conclusion and the
rate are unchanged. The conditions required are the analogues of \assterm{}
and \assnondeg{} for the grid, and they constrain $M$ only. The proposal grid
size $M$ determines the proposed split and hence the population recursion
tree. By Proposition~\ref{prop:gpopexact}, a re-test window contains either no
live kink or exactly one live kink; in the latter case its pooled profile has
a unique, non-degenerate maximizer for every fixed $R$. Therefore no analogous
condition is needed for the re-test grid size $R$.

Two points make the extension possible. The first is that the noise bound of
Lemma~\ref{lem:noise} is already a maximum over \emph{all} sub-intervals of
$\{1,\ldots,n\}$, so it covers every sub-interval a grid of any size can
produce, and the same threshold works for every $M$. The second is that a grid
sub-interval of the current window is contained in that window, so a kink close
to an edge of the window is close to an edge of, or outside, every sub-interval;
the signal strength dichotomy of Corollary~\ref{cor:gap} therefore transfers to the
sub-intervals on which the grid search evaluates the contrast. What the grid
does change is the location of the split, and hence the recursion tree, so the
population objects have to be redefined for each $M$.

\subsection{The grid statistic and the population grid recursion}
\label{sup:gridpop}

For a window $(s,e]$ and $M \ge 2$ let $g_1 = s < g_2 < \cdots < g_M = e$ be
the $M$-point grid of Section~3.2 of the main paper, and write
\[
\mathcal{P}_M(s,e) := \big\{ (i,j) : 1 \le i < j \le M, \ g_j - g_i \ge 3 \big\}
\]
for the admissible pairs. The grid statistic and its maximizer are
\begin{equation}
\label{eq:gridstat}
\begin{aligned}
\mathcal{C}^{(M)}_{s,e}(b)
&:= \max\big\{ \big| \langle X, \hat\psi^{g_i,g_j}_b \rangle \big| :
 (i,j) \in \mathcal{P}_M(s,e), \ b \in B_{g_i,g_j} \big\}, \\
\hat b &:= \min\arg\max_b \mathcal{C}^{(M)}_{s,e}(b).
\end{aligned}
\end{equation}
where $\hat\psi^{g_i,g_j}_b$ denotes the unit contrast vector of
Section~\ref{sup:toolkit} formed on the sub-window $(g_i,g_j]$. This is the
routine $\textsc{GridSearch}(X,s,e,M)$, and \eqref{eq:gridstat} returns the
pair $(\hat b, W)$ with $W = \mathcal{C}^{(M)}_{s,e}(\hat b)$.

At the population level, for a rescaled window $(\alpha,\beta)$ put
$\gamma_i = \alpha + (i-1)(\beta-\alpha)/(M-1)$, $i = 1,\ldots,M$, and define
the \emph{pooled profile} and its leftmost maximizer
\begin{equation}
\label{eq:pooled}
\bar D^{(M)}_{\alpha\beta}(v) := \max_{i<j \, : \, \gamma_i < v < \gamma_j}
\big| \bar D_{\gamma_i\gamma_j}(v) \big| ,
\qquad
c_M(\alpha,\beta) := \min \arg\max_v \bar D^{(M)}_{\alpha\beta}(v),
\end{equation}
with $\Lambda_M(\alpha,\beta) := \max_v \bar D^{(M)}_{\alpha\beta}(v)$. Here
$\bar D_{\gamma\gamma'}$ is the profile of Section~4.2 of the main paper formed
on $(\gamma,\gamma')$, built from the kinks in $V(\gamma,\gamma') =
\{l : \gamma < q_l < \gamma'\}$. Because that profile is defined through the
unnormalized inner product of $L^2(\gamma,\gamma')$, values on sub-windows of
different lengths are on a common scale, as required in \eqref{eq:pooled}; this
is the population counterpart of the observation in Section~3.2 of the main
paper that unit-norm contrast vectors make sub-interval contrasts directly
comparable. Since $(\gamma_1,\gamma_M) = (\alpha,\beta)$ is itself an
admissible pair, $\Lambda_M(\alpha,\beta) \ge \Lambda(\alpha,\beta) > 0$
whenever $V(\alpha,\beta) \ne \emptyset$, and
$\Lambda_M(\alpha,\beta) = 0$ otherwise.

At a node with a non-empty live set, call a pair $(i,j)$ \emph{attaining} if
the largest absolute profile value on its sub-window equals the pooled
maximum. Write
\[
\mathcal{P}^*(\alpha,\beta)
:= \left\{(i,j):1\le i<j\le M,\ 
\max_v\big|\bar D_{\gamma_i\gamma_j}(v)\big|
=\Lambda_M(\alpha,\beta)\right\}
\]
for the set of all such pairs.

\begin{defin}[Population LABS-Grid]
\label{def:gpop}
Fix $M, R \ge 2$. $\mathrm{POP}_{M,R}(\alpha,\beta)$ is
\defpop{} with $c(\alpha,\beta)$ replaced by $c_M(\alpha,\beta)$ in the
proposal step and by $c_R(\alpha',\beta')$ in the re-test step. Write
$\mathcal{T}_{M,R}(g)$ for the set of windows at which $\mathrm{POP}_{M,R}$ is
invoked, starting from $(0,1)$, and $\mathcal{G}_{M,R}(g)$ for the set of grid
points of all those windows together with the grid points of the associated
re-test windows.
\end{defin}

\begin{prop}[Population exactness]
\label{prop:gpopexact}
If $\mathcal{T}_{M,R}(g)$ is finite and $c_M, c_R$ are well defined at the
relevant windows, then $\mathrm{POP}_{M,R}(\alpha,\beta) = \{q_j : j \in
V(\alpha,\beta)\}$ at every node; in particular
$\mathrm{POP}_{M,R}(0,1) = \{q_1,\ldots,q_N\}$.
\end{prop}

\begin{proof}
The proof of \proppop{} uses the maximizer only through its position relative
to the kinks, and applies verbatim once two facts are checked.

First, detection is unchanged: $\Lambda_M(\alpha,\beta) > 0$ if and only if
$V(\alpha,\beta) \ne \emptyset$, as noted above.

Second, if $V(\alpha,\beta) = \{j\}$ then $c_M(\alpha,\beta) = q_j$.
Indeed, a sub-window with $q_j$ in its interior has $q_j$ as its only live
kink, so by Lemma~\ref{lem:rep}(c) its profile attains its maximum at $q_j$ and
nowhere else; a sub-window without $q_j$ in its interior has an identically
zero profile. The pooled profile is therefore maximized at $q_j$ and nowhere
else, whatever the relative sizes of the individual maxima. The same argument
applies to $c_R$ on the re-test window, which in Case~B of the proof of
\proppop{} includes exactly one live kink.
\end{proof}

\subsection{Assumptions}
\label{sup:gridass}

\begin{assumption}[Termination, grid]
\label{ass:gterm}
$\mathrm{POP}_{M,R}$ started at $(0,1)$ terminates, that is,
$\mathcal{T}_{M,R}(g)$ is finite.
\end{assumption}

\begin{assumption}[Non-degeneracy, grid]
\label{ass:gnondeg}
\mbox{}\par
At every node of $\mathcal{T}_{M,R}(g)$ with a non-empty live set, the pooled profile
$\bar D^{(M)}_{\alpha\beta}$ has a unique
maximizer $c_M(\alpha,\beta)$ on $(\alpha,\beta)$, and for every
$(i,j)\in\mathcal{P}^*(\alpha,\beta)$,
\[
\partial_v^2 \big| \bar D_{\gamma_i\gamma_j} \big| \big( c_M(\alpha,\beta)
\big) \;<\; 0 .
\]
\end{assumption}

Two remarks on the form of these conditions. First, no condition is imposed on
$R$; the reason is given in Remark~\ref{rem:noR} below. Second, when
$V(\alpha,\beta)$ is a singleton, Assumption~\ref{ass:gnondeg} is automatic:
the second step of the proof of Proposition~\ref{prop:gpopexact} shows that the
maximizer is then unique and equal to the kink, and non-degeneracy at that
point follows from Lemma~\ref{lem:sep}, exactly as in the case $k=1$ of
Proposition~\ref{prop:track_prop}. The condition therefore has content only at
nodes including two or more live kinks.

Under Assumptions~\ref{ass:gterm} and~\ref{ass:gnondeg}, and again for
$N \ge 1$, the sets of nodes, grid points and attaining pairs are finite. We
now define the finite-set extrema used in the grid proof. Let $\mathcal Q$ be
the set of triples $((\alpha,\beta),i,j)$ for which
$(\alpha,\beta)\in\mathcal{T}_{M,R}(g)$ has a non-empty live set and
$(i,j)\notin\mathcal{P}^*(\alpha,\beta)$:
\begin{align}
\label{eq:kappaM}
\kappa_M &:= \min\Big( \{\delta_{\min}\} \cup \big\{ |q_l - \gamma| \, : \,
\gamma \in \mathcal{G}_{M,R}(g), \ 1 \le l \le N, \ q_l \ne \gamma \big\}
\notag \\
&\hspace{5.2em} \cup \big\{ |c_M(\alpha,\beta) - q_l| \, : \,
(\alpha,\beta) \in \mathcal{T}_{M,R}(g), \ c_M(\alpha,\beta) \ne q_l \big\}
\Big), \\
\iota &:= \min \Big\{ \Lambda_M(\alpha,\beta) -
\textstyle\max_v |\bar D_{\gamma_i\gamma_j}(v)| \, : \,
((\alpha,\beta),i,j)\in\mathcal Q \Big\} .
\end{align}
The minimum defining $\iota$ is taken only over nodes with a non-empty live
set and non-attaining pairs. If there is no such pair, set $\iota=1$. Also set
\[
\Lambda_{\min}:=
\min_{(\alpha,\beta)\in\mathcal{T}_{M,R}(g):\,V(\alpha,\beta)\ne\emptyset}
\Lambda_M(\alpha,\beta).
\]
Then $\kappa_M>0$ because every distance included in its definition is
positive and there are finitely many of them. Similarly,
$\Lambda_{\min}>0$ by \eqref{eq:pooled}, and $\iota>0$ because every
non-attaining pair has a maximum strictly below $\Lambda_M(\alpha,\beta)$.

It remains to define the constants controlling curvature, interior margin and
endpoint perturbations. These constants are not defined directly from the
pooled profile. A maximum of smooth individual profiles need not be
differentiable at a point where two profiles have the same value but different
derivatives and the identity of the larger profile changes. We instead work
with each individual profile. For each node $(\alpha,\beta)$ with a non-empty
live set and each $(i,j)\in\mathcal{P}^*(\alpha,\beta)$, the individual absolute
profile $|\bar D_{\gamma_i\gamma_j}|$ has its unique maximum at
$c_M(\alpha,\beta)$. Indeed, a second maximizer would also maximize the pooled
profile, contrary to Assumption~\ref{ass:gnondeg}. The same assumption gives a
strictly negative second derivative there. Apply the construction in
Lemma~\ref{lem:conseq} to each such individual profile on
$(\gamma_i,\gamma_j)$. Define $\varkappa$, $\varrho_0$ and $\varsigma_0$ as,
respectively, the minima of the resulting positive curvature constants,
positive relative interior margins and positive endpoint-perturbation radii.
Define $\Lambda_{\max}$ and $C_L$ as, respectively, the maxima of the
individual profile amplitudes and the finite endpoint-Lipschitz constants.
All five extrema are over the finitely many node--pair combinations, so the
three minima are strictly positive and the two maxima are finite. These are
the constants used when Lemma~\ref{lem:conseq} is applied to an attaining pair
in Lemma~\ref{lem:gloc}.

The grid points must be included in $\kappa_M$ because they are endpoints of
the sub-windows on which the contrast is evaluated, even though they are not
endpoints of recursion windows.

\subsection{Consistency}
\label{sup:gridthm}

Throughout, $(s,e]$ tracks $(\alpha,\beta) \in \mathcal{T}_{M,R}(g)$ at scale
$C_\star$ in the sense of Definition~\ref{def:track}, and we work on the event $E_n$ of
Lemma~\ref{lem:noise}. Note first that the sample and population grids
correspond: since $|s - n\alpha| \le C_\star \bar r_n$ and
$|e - n\beta| \le C_\star\bar r_n$, rounding gives
\begin{equation}
\label{eq:gridtrack}
| g_i - n \gamma_i | \;\le\; C_\star \bar r_n + 1, \qquad i = 1,\ldots,M,
\end{equation}
so that each sample sub-window $(g_i,g_j]$ tracks the population sub-window
$(\gamma_i,\gamma_j)$ at scale $C_\star + 1$. Lemma~\ref{lem:track}, applied
with $\kappa_M$ in place of $\kappa$, then gives the dichotomy on every
sub-window: for $l \in V(\gamma_i,\gamma_j)$ one has
$x_l(g_i,g_j) \ge \tfrac12\kappa_M n$, and otherwise
$x_l(g_i,g_j) \le C_\star\bar r_n + 2$.

\begin{lemma}[No detection, grid]
\label{lem:gnodet}
Fix $C_\star$ and $M \ge 2$. There is $n_0$ such that for $n \ge n_0$, on
$E_n$: if $(s,e]$ is any window with $x_j(s,e) \le C_\star\bar r_n + 1$ for
every $j$, then $\max_b \mathcal{C}^{(M)}_{s,e}(b) < \lambda$. Thus the value
$W$ returned by $\textnormal{\textsc{GridSearch}}(X,s,e,M)$ fails the proposal condition
$W>\lambda$ in Algorithm~S1, and the call returns $\emptyset$ without making a
split.
\end{lemma}

\begin{proof}
Let $(i,j) \in \mathcal{P}_M(s,e)$ and let $\tau_l$ lie in the interior of
$(g_i,g_j]$. Since $g_i \ge s$ and $g_j \le e$, both $\tau_l - g_i \le \tau_l - s$ and $g_j + 1 - \tau_l \le e + 1 - \tau_l$, so
\begin{equation}
\label{eq:grid-edge-distance}
x_l(g_i,g_j) \;\le\; x_l(s,e) \;\le\; C_\star\bar r_n + 1 .
\end{equation}
By Corollary~\ref{cor:gap}(a), the signal strength of $\tau_l$ on $(g_i,g_j]$ is
$O\{n^{-1/4}(\log n)^{3/4}\}$. By Lemma~\ref{lem:rep}(b) the noiseless
contrast on that sub-window is at most $N$ times this, and adding
$Z_n$ and maximizing over the at most $\binom{M}{2}$ admissible pairs gives
$\max_b \mathcal{C}^{(M)}_{s,e}(b) \le o(1) + \sigma(8\log n)^{1/2} < \lambda$
for $n$ large, since $M$ is fixed and $\Theta > 1$.
\end{proof}

Formula~\eqref{eq:grid-edge-distance} is the step that extends the no-detection
argument from the original window to every grid sub-window: shrinking a
window can only move an interior kink closer to an endpoint. Thus a kink that
has insufficient signal strength on $(s,e]$ also has insufficient signal
strength on each sub-window on which the grid search evaluates the contrast.

\begin{lemma}[Detection and localization, grid]
\label{lem:gloc}
\mbox{}\par
Let $(\alpha,\beta)$ be a node of $\mathcal{T}_{M,R}(g)$ with a non-empty live
set, and let $(s,e]$ track it at scale $C_\star$. Under
Assumptions~\ref{ass:gterm} and~\ref{ass:gnondeg} there are $n_0$ and $C_1$
such that for $n \ge n_0$, on $E_n$,
\[
\max_b \mathcal{C}^{(M)}_{s,e}(b) \;>\; \lambda,
\qquad
\big| \hat b - n\, c_M(\alpha,\beta) \big| \;\le\; C_1 \bar r_n .
\]
\end{lemma}

\begin{proof}
Detection does not require Assumption~\ref{ass:gnondeg}. Let
$(i,j)\in\mathcal{P}^*(\alpha,\beta)$. Its population maximum is
$\Lambda_M(\alpha,\beta)\ge\Lambda_{\min}$ by the definition of
$\Lambda_{\min}$. The maximum in \eqref{eq:gridstat} is at least the maximum
over this pair, which by \eqref{eq:riemann} applied on $(g_i,g_j]$ is at least
$n^{1/2}\{\Lambda_M(\alpha,\beta)-o(1)\}-Z_n
\ge\tfrac12\Lambda_{\min}n^{1/2}$, and this exceeds $\lambda$ for large $n$.
The full-window amplitude $\Lambda(\alpha,\beta)$ is no larger than the pooled
amplitude at the same node and may be strictly smaller than the minimum pooled
amplitude across all nodes. It therefore cannot replace
$\Lambda_M(\alpha,\beta)$ in this lower bound while retaining the constant
$\Lambda_{\min}$. One could introduce a second, possibly smaller, positive
constant by taking the minimum of $\Lambda(\alpha,\beta)$ over the grid-tree
nodes, but it is unnecessary because the grid statistic searches the
attaining pairs directly.

For localization, write $\Lambda_M = \Lambda_M(\alpha,\beta)$ and abbreviate
$\mathcal{P}^*(\alpha,\beta)$ to $\mathcal{P}^*$. By
\eqref{eq:gridtrack} and \eqref{eq:riemann} applied on each sub-window,
\[
\max_b \big| \langle X, \hat\psi^{g_i,g_j}_b \rangle \big|
\;=\; n^{1/2} \Big\{ \max_v \big| \bar D_{\gamma_i\gamma_j}(v) \big|
+ o(1) \Big\} + O\{(\log n)^{1/2}\}
\]
uniformly over the finitely many pairs. Hence for $(i,j) \notin \mathcal{P}^*$
the left side is at most $n^{1/2}(\Lambda_M - \iota + o(1))$, while for
$(i,j) \in \mathcal{P}^*$ it is at least $n^{1/2}(\Lambda_M - o(1))$. For $n$
large the two ranges are disjoint, so the maximum in \eqref{eq:gridstat} is
attained at a pair in $\mathcal{P}^*$.

Let $(i,j) \in \mathcal{P}^*$ be that pair. If $v$ maximizes
$|\bar D_{\gamma_i\gamma_j}|$ then $\bar D^{(M)}_{\alpha\beta}(v) = \Lambda_M$,
so $v = c_M(\alpha,\beta)$ by the uniqueness in
Assumption~\ref{ass:gnondeg}: every pair in $\mathcal{P}^*$ attains its maximum
at $c_M(\alpha,\beta)$ and there only. Assumption~\ref{ass:gnondeg} guarantees the negative second derivative at that point, so Lemma~\ref{lem:conseq} holds for the sub-window $(\gamma_i,\gamma_j)$ with $V(\gamma_i,\gamma_j)$ in place of $V(\alpha,\beta)$, and Lemma~\ref{lem:transfer}, which is stated for any fixed rescaled window at which those conclusions hold and not only for nodes of $\mathcal{T}(g)$, applies to $(g_i,g_j]$. That the sample sub-window tracks the population one at scale $C_\star+1$ is \eqref{eq:gridtrack}, and the kink dichotomy on it is Lemma~\ref{lem:track} with $\kappa_M$ in place of $\kappa$. Since
$c_M(\alpha,\beta)$ lies in the interior of $(\gamma_i,\gamma_j)$ with margin
at least $\varrho_0(\gamma_j - \gamma_i)$, Lemma~\ref{lem:loc} gives
$|\hat b - n c_M(\alpha,\beta)| \le C_1\bar r_n$, with $C_1$ the largest of the
finitely many constants so obtained.
\end{proof}

\begin{lemma}[Grid re-test on a one-kink window]
\label{lem:grt}
Let $(\alpha',\beta')$ be one of the finitely many Case~B re-test windows arising in the proof of Proposition~\ref{prop:gpopexact}, so that $V(\alpha',\beta') = \{j^*\}$, and let $(s',e']$ track it at some fixed scale $K$. There is a constant $C^{M,R}_{\mathrm{rt}}$, not depending on $K$, and an $n_0 = n_0(K)$, such that for $n \ge n_0$, on $E_n$,
\[
\max_b \mathcal{C}^{(R)}_{s',e'}(b) \;\ge\; \tfrac12 \Lambda^{M,R}_{\mathrm{det}}\, n^{1/2} \;>\; \lambda,
\qquad
\big| \hat b' - \tau_{j^*} \big| \;\le\; C^{M,R}_{\mathrm{rt}} \bar r_n .
\]
\end{lemma}

\begin{proof}
The population re-test window $(\alpha',\beta')$ contains the single live kink
$q_{j^*}$. Hence, by the second step of the proof of
Proposition~\ref{prop:gpopexact}, its pooled $R$-grid profile is maximized at
$q_{j^*}$ and nowhere else, whatever the relative sizes of the individual
sub-window maxima. The pooled maximum equals
$|d_{j^*}|\,\|\Psi^{\alpha'\beta'}_{q_{j^*}}\|$, which is at least
$\Lambda^{M,R}_{\mathrm{det}}$ by \eqref{eq:LamdetMR}. Detection follows as in
Lemma~\ref{lem:gloc}, from \eqref{eq:riemann} applied to a pair attaining this
pooled maximum. For localization a positive gap between an attaining pair and
the other pairs is needed, and the detection constant, being an amplitude
rather than a difference, does not provide one. Here the gap follows from
monotonicity of the one-kink amplitude. Writing
$A=q_{j^*}-\alpha'$ and $B=\beta'-q_{j^*}$, the scaling relation of
Lemma~\ref{lem:gram} gives
\[
\big\|\Psi^{\alpha'\beta'}_{q_{j^*}}\big\|^2 \;=\; \frac{A^3B^3}{3(A+B)^3} \;=\; \tfrac13 H(A,B)^3,
\qquad H(A,B) := \frac{AB}{A+B},
\]
and $\partial_A H = B^2/(A+B)^2 > 0$, $\partial_B H = A^2/(A+B)^2 > 0$. The one-kink amplitude is therefore strictly increasing under enlargement of the window at either end, so among the $R$-grid sub-windows of $(\alpha',\beta')$ that contain $q_{j^*}$ the full window $(\gamma'_1,\gamma'_R) = (\alpha',\beta')$ is the unique maximizer, the others containing $q_{j^*}$ being proper sub-windows and those not containing it having amplitude zero. To state the gap, set
\[
A_{ij} \;:=\;
\begin{cases}
|d_{j^*}| \, \big\|\Psi^{\gamma'_i\gamma'_j}_{q_{j^*}}\big\| , & \gamma'_i < q_{j^*} < \gamma'_j, \\[2pt]
0, & \text{otherwise},
\end{cases}
\]
which is defined for every admissible pair, including those that do not contain the kink and on which the profile vanishes identically. Since there are finitely many pairs and finitely many re-test configurations,
\[
\iota^{M,R}_{\mathrm{rt}} \;:=\; \min_{\text{re-tests}} \Big\{ A_{1R} \;-\; \max_{(i,j) \ne (1,R)} A_{ij} \Big\} \;>\; 0 ,
\qquad \max\emptyset := 0 .
\]
When $R = 2$ there is only the pair $(1,2)$, the inner maximum is over the empty set, and $\iota^{M,2}_{\mathrm{rt}} = \min A_{12} > 0$, which is the case $R = 2$ recommended in Remark~\ref{rem:noR}. When $R > 2$, strict monotonicity of $H$ makes every proper sub-window containing the kink strictly smaller, and the pairs not containing it contribute zero, so the minimum is again positive. The gap argument of Lemma~\ref{lem:gloc} then applies with $\iota^{M,R}_{\mathrm{rt}}$ in place of $\iota$: for $n$ large the maximum in \eqref{eq:gridstat} is attained on the full pair. On it $q_{j^*}$ is interior with a fixed margin, and Lemma~\ref{lem:loc} applies with $U = \{j^*\}$, $b^* = \tau_{j^*}$ and $\varkappa' = c_0$, exactly as in the case $k=1$ of Proposition~\ref{prop:track_prop}. No node-level amplitude or curvature assumption is used: a window with one live kink has a unique and non-degenerate pooled maximizer for every grid size. The constant $C^{M,R}_{\mathrm{rt}}$ does not depend on the tracking scale $K$, because the parameters $a_0$, $\varkappa' = c_0$, $\varrho_1$ and $r_0$ fed into Lemma~\ref{lem:loc} are determined by $\Lambda^{M,R}_{\mathrm{det}}$, $\iota^{M,R}_{\mathrm{rt}}$, $\kappa_M$ and the fixed margins alone; $K$ enters only through the requirement that $2K\bar r_n/n$ be below $\varsigma_0$ and that the absorbed kinks have signal strength $o(1)$, both of which hold for all large $n$ at each fixed $K$. 
\end{proof}

\begin{defin}[Grid tracking radii]
\label{def:gradii}
The radii of Definition~\ref{def:radii} cannot be used unchanged, because the constant provided by Lemma~\ref{lem:gloc} at a node depends on the scale at which that node is tracked: its proof applies Lemma~\ref{lem:transfer} on grid sub-windows tracking population sub-windows at scale $K_v+1$. The construction is instead recursive in the same finite tree. Write $G_v(K)$ for the constant that Lemma~\ref{lem:gloc} provides at $v \in \mathcal{T}_{M,R}(g)$ when its window tracks $v$ at scale $K$. Set $K_{\mathrm{root}} := 0$ and, for each child $v'$ of $v$, $K_{v'} := \max\{K_v,\, G_v(K_v)\}$; then $C_\star := \max\{\max_v K_v,\, C^{M,R}_{\mathrm{rt}}\}$, with $C^{M,R}_{\mathrm{rt}}$ as in Lemma~\ref{lem:grt}. The constants need not decrease from parent to child. Assumption~\ref{ass:gterm} makes $\mathcal{T}_{M,R}(g)$ finite, so all recursively defined constants and their maximum are finite.
\end{defin}

\begin{prop}[Recursion on tracking windows, grid]
\label{prop:gtrack}
Assume \ref{ass:gterm} and \ref{ass:gnondeg}. There is $n_0$ such that for $n \ge n_0$, on $E_n$, the conclusion of Proposition~\ref{prop:track_prop} holds for LABS-Grid with parameters $M$ and $R$, for every $v \in \mathcal{T}_{M,R}(g)$ and every $(s,e]$ tracking it at scale $K_v$, the radii being those of Definition~\ref{def:gradii}.
\end{prop}

\begin{proof}
The induction is that of Proposition~\ref{prop:track_prop}, on the height of the subtree of $\mathcal{T}_{M,R}(g)$, with Lemma~\ref{lem:gnodet} in place of Lemma~\ref{lem:nodet} and Lemma~\ref{lem:gloc} in place of the combination of Lemmas~\ref{lem:transfer} and~\ref{lem:loc}. The tracking radii are those of Definition~\ref{def:gradii} above. Only the three places where the grid enters need comment.

\emph{Case $k = 0$.} Lemma~\ref{lem:track} gives
$x_j(s,e) \le C_\star\bar r_n + 1$ for every $j$, and
Lemma~\ref{lem:gnodet} returns $\emptyset$.

\emph{Case $k = 1$, $V(\alpha,\beta) = \{j\}$.} By
Proposition~\ref{prop:gpopexact}, $c_M(\alpha,\beta) = q_j$, and
Lemma~\ref{lem:gloc} gives $|\hat b - \tau_j| \le C_1\bar r_n$ without
an additional non-degeneracy condition: with one live kink the pooled profile
has its unique maximum at $q_j$ by Proposition~\ref{prop:gpopexact}, and
Lemma~\ref{lem:sep} gives negative curvature there. The children
track $(\alpha,q_j)$ and $(q_j,\beta)$, both with empty live sets, so both
return $\emptyset$ by the case $k=0$, and the call returns $\{\hat b\}$.

\emph{Case $k \ge 2$.} Lemma~\ref{lem:gloc} gives
$|\hat b - n c_M(\alpha,\beta)| \le C_1\bar r_n$, so the children track
$(\alpha, c_M)$ and $(c_M,\beta)$ at their assigned scales from
Definition~\ref{def:gradii}. These child nodes have smaller subtree height, so
the induction hypothesis says that each recursive call returns exactly one
estimate for every live kink in that child, each within
$C_\star\bar r_n$ of the corresponding kink. The
division into Cases~A and~B is by whether $c_M(\alpha,\beta)$ is a kink, as in
Proposition~\ref{prop:track_prop}. Proposition~\ref{prop:gpopexact} guarantees
that the two population children contain all current live kinks except, in
Case~B, the kink at the split. It also guarantees that the re-test window
contains no live kink in Case~A and exactly the omitted kink in Case~B. In
Case~A the re-test window therefore includes no live
kink and every kink lies within $C_\star\bar r_n + 1$ of one of its edges, so
Lemma~\ref{lem:gnodet}, applied with $R$ in place of $M$, shows that the
re-test statistic is below $\lambda$. In Case~B the re-test window carries the single live kink $\tau_{j^*}$ at distance at least $\tfrac12\kappa_M n$ from both edges. It is in general not a node of $\mathcal{T}_{M,R}(g)$, and its amplitude may be smaller than every node amplitude, which is why $\Lambda^{M,R}_{\mathrm{det}}$ was introduced; Lemma~\ref{lem:gloc}, which is stated at nodes, therefore cannot be applied to it. Lemma~\ref{lem:grt} is used instead; it shows that the re-test statistic exceeds $\lambda$ and that $|\hat b' - \tau_{j^*}| \le C^{M,R}_{\mathrm{rt}}\bar r_n \le C_\star\bar r_n$. Order preservation is unchanged.
\end{proof}

\begin{theorem}[Consistency of LABS-Grid]
\label{thm:grid}
Fix $M \ge 2$ and $R \ge 2$. Let $X_1,\ldots,X_n$ follow the model of Section~4 of the main paper, and suppose Assumptions~\ref{ass:gterm} and~\ref{ass:gnondeg} hold for $g$ with these $M$ and $R$. Let
$\hat{\mathcal{T}}_n$ be the set returned by LABS-Grid, Algorithm~S1, with the
contrast of the main paper and threshold $\lambda_n$. Then the conclusion of
\mainthm{} holds verbatim, with a constant $C$ depending on $g$, $\sigma$, $M$
and $R$.
\end{theorem}

\begin{proof}
As in Section~\ref{sup:thmproof}: the initial call is on $(0,n]$, which tracks
the root $(0,1)$ of $\mathcal{T}_{M,R}(g)$ exactly and has
$V(0,1) = \{1,\ldots,N\}$, so Proposition~\ref{prop:gtrack} applies at the
root. If $N = 0$ then no window contains a kink and Lemma~\ref{lem:gnodet}
gives $\hat{\mathcal{T}}_n = \emptyset$.
\end{proof}

\subsection{Remarks}
\label{sup:gridremarks}

\begin{remark}[No condition on $R$]
\label{rem:noR}
Assumptions~\ref{ass:gterm} and~\ref{ass:gnondeg} restrict $M$ only, and
Theorem~\ref{thm:grid} holds for every $R \ge 2$. The reason is that the
re-test window is never in the position that requires an assumption. In Case~A
it includes no live kink, and Lemma~\ref{lem:gnodet} needs no hypothesis beyond
the dichotomy; in Case~B it includes exactly one, and a window with one live kink has a unique and non-degenerate pooled maximizer, at the kink, for every grid size. The look-ahead has already reduced the re-test to a single-change-point problem, for which the conclusion holds at every grid size. In particular $R = 2$ suffices, which is the formal counterpart
of the observation in Section~3.2 of the main paper that a single-interval
evaluation should suffice at the re-test stage.
\end{remark}

\begin{remark}[The threshold does not depend on $M$]
\label{rem:thresholdM}
The same $\lambda_n$ serves every $M$. This is because the event $E_n$ of
Lemma~\ref{lem:noise} controls $|\langle\varepsilon,\hat\psi_b\rangle|$
simultaneously over all $O(n^3)$ triples $(s,e,b)$, which already includes
every sub-interval any grid can produce; enlarging the collection searched by a
fixed factor $\binom{M}{2}$ does not change the order of the maximum. The
observation in Section~3.2 of the main paper, that the null distribution of the
proposal statistic grows with $M$, concerns the sharpness of the constant in
$\lambda_n$ rather than its order, and remains relevant in finite samples.
\end{remark}

\begin{remark}[What the grid changes]
\label{rem:gridchanges}
The grid does not weaken the conditions; it changes the population recursion to
which they are applied. A window carrying several live kinks may have a
sub-window that isolates one of them, and if the contrast on that sub-window
exceeds the contrast on every other, then $c_M(\alpha,\beta)$ is a kink where
$c(\alpha,\beta)$ was not. In the language of
Proposition~\ref{prop:gpopexact} this moves a node from Case~A to Case~B: the
split is made at a change-point rather than between change-points, and the
look-ahead recovers the change-point at the re-test rather than discarding a
spurious estimate. This may reduce the recursion depth. It also means that
Assumptions~\ref{ass:gterm} and~\ref{ass:gnondeg} are specific to $M$:
they may hold for one $M$ and fail for another, and neither implies
\assterm{} or \assnondeg{}.
\end{remark}

\begin{remark}[Recovering \mainthm{}]
\label{rem:M2}
When $M = 2$ the only admissible pair is $(1,2)$, so
$\bar D^{(2)}_{\alpha\beta} = |\bar D_{\alpha\beta}|$,
$c_2 = c$, $\mathcal{T}_{2,R}(g) = \mathcal{T}(g)$ and
$\mathcal{G}_{2,R}(g)$ consists of the window endpoints, so $\kappa_2 = \kappa$
and $\mathcal{P}^*$ is a singleton with $\iota = 1$.
Assumptions~\ref{ass:gterm} and~\ref{ass:gnondeg} reduce to \assterm{} and
\assnondeg{}, and Theorem~\ref{thm:grid} reduces to \mainthm{}.
\end{remark}

\section{Consistency under strengthened SIC selection}
\label{sup:ssic}

\mainthm{} and Theorem~\ref{thm:grid} concern LABS run at a single threshold in
a range that depends on unknown quantities. The procedure used in Section~5 of
the main paper instead computes a solution path over a grid of thresholds and
selects from it by a strengthened Schwarz criterion. This section shows that
the selected model is consistent, with the same rate.

The argument has two parts. First, the path contains a candidate change-point
set that satisfies the exact-count and localization conclusions of
\mainthm{}. Second, the criterion does not prefer any other set on the path: a
set that misplaces a change-point pays in residual sum of squares, and a set
with superfluous change-points pays in penalty. The second comparison is what
requires the exponent in the penalty to exceed one.

\subsection{The procedure}
\label{sup:ssicproc}

For $\mathcal{T}\subset\{1,\ldots,n-1\}$, let
\[
\mathcal{M}(\mathcal{T})
:=\operatorname{span}\big\{\mathbf 1,t,(t-u)_+:u\in\mathcal{T}\big\}.
\]
This is the space of vectors obtained by restricting continuous
piecewise-linear functions with knots in $\mathcal{T}$ to the observation grid
$\{1,\ldots,n\}$. Here an affine function is one of the form $a+bt$. The
dimension is at most $|\mathcal{T}|+2$. If $1\notin\mathcal{T}$, the displayed
vectors are linearly independent: their second differences identify the
coefficients of the hinge functions $(t-u)_+$ separately, after which the
constant and linear coefficients must also vanish. If $1\in\mathcal{T}$, then
$(t-1)_+=t-1$ for every observed $t$, so this hinge vector already lies in
$\operatorname{span}\{\mathbf 1,t\}$ and does not add a dimension. For every
$u\ge2$, by contrast, $(t-u)_+$ has a change of slope on the observation grid
and is not affine on the whole grid. The proof
below uses only the upper bound $\dim\mathcal{M}(\mathcal{T})\le
|\mathcal{T}|+2$, in the union bound in Lemma~\ref{lem:proj}. Let
$P_{\mathcal{T}}$ be orthogonal projection onto
$\mathcal{M}(\mathcal{T})$, and let
\[
\mathrm{RSS}(\mathcal{T}) := \big\| (I - P_{\mathcal{T}}) X \big\|^2 ,
\qquad
\mathrm{sSIC}(\mathcal{T}) := n \log\{\mathrm{RSS}(\mathcal{T})/n\}
+ \big(2|\mathcal{T}|+3\big) \xi_n ,
\]
with $\xi_n := (\log n)^{\gamma}$ and $\gamma > 1$ fixed; Section~5 of the main
paper takes $\gamma = 1.01$. Write
\[
\hat\sigma := \mathrm{MAD}\big(\Delta^2 X / \sqrt6\big),
\qquad
\lambda_a := a\, \hat\sigma\, (2\log n)^{1/2},
\qquad
\Pi_n := \big\{ \hat{\mathcal{T}}(\lambda_a) : a \in A_n \big\},
\]
where $\hat{\mathcal{T}}(\lambda)$ is the output of LABS, or of LABS-Grid, at
threshold $\lambda$, and $A_n$ is the grid of multipliers. The selected set is
\[
\hat{\mathcal{T}}^{\,\mathrm{sSIC}} := \arg\min_{\mathcal{T} \in \Pi_n}
\mathrm{sSIC}(\mathcal{T}),
\]
ties being broken by taking the smallest $|\mathcal{T}|$ and then the smallest
set in lexicographic order.

For a candidate set write
\[
b(\mathcal{T}) := \big\| (I - P_{\mathcal{T}}) f \big\|^2 ,
\qquad
\delta_j(\mathcal{T}) := \min\big\{ |\tau_j - u| : u \in \mathcal{T} \cup
\{0,n\} \big\} ,
\]
for the approximation error and the distance from the $j$th change-point to the
nearest knot.

\subsection{Conditions}
\label{sup:ssiccond}

In addition to \assterm{} and \assnondeg{}, or to their grid counterparts, we
require the following.

\begin{assumption}[Noise]
\label{ass:ssicnoise}
$\varepsilon_1,\ldots,\varepsilon_n$ are independent and identically distributed with mean zero, variance $\sigma^2$, and $\sigma$-sub-Gaussian, and $\hat\sigma/\sigma \to 1$ in probability. The two roles of $\sigma$ have to coincide here. Section~4 of the main paper assumes a sub-Gaussian parameter $\sigma$, and the threshold band of Remark~\ref{rem:uniformlam} is expressed in terms of that parameter, whereas $\hat\sigma$ estimates the standard deviation; if the two differed by a factor $K > 1$, the band would have to start at $(1+\eta)K\sigma(8\log n)^{1/2}$ and the constants below would depend on $K$. Assuming both equal to $\sigma$, which holds for Gaussian errors, avoids carrying $K$ through the argument.
\end{assumption}

\begin{assumption}[Threshold grid]
\label{ass:ssicgrid}
Suppose first that $N \ge 1$. There are fixed $\eta > 0$ and $\varsigma \in (0, \tfrac12\Lambda_{\mathrm{det}})$, or $\varsigma \in (0,\tfrac12\Lambda^{M,R}_{\mathrm{det}})$ in the grid case, such that, with probability tending to one, $A_n$ contains a value $a$ with $(1+\eta)\,\sigma(8\log n)^{1/2} \le \lambda_a \le \varsigma n^{1/2}$, that is, with $\lambda_a$ in the band of Remark~\ref{rem:uniformlam}. If $N = 0$, in which case $\Lambda_{\mathrm{det}}$ is a minimum over an empty set and is undefined, there is a fixed $\eta > 0$ such that, with probability tending to one, $A_n$ contains a value $a$ with $\lambda_a \ge (1+\eta)\sigma(8\log n)^{1/2}$, at which LABS returns the empty set. The upper restriction is important as a threshold of order $L n^{1/2}$ with $L$ above the population contrast at the root satisfies the lower bound but causes LABS to return the empty set, so without it the path need not contain a consistent model at all. The factor $1+\eta$ gives a fixed multiplicative gap above the noise bound, meaning that $\eta$ is positive and does not depend on $n$. This gap is needed because $\lambda_a$ is random through $\hat\sigma$; by Remark~\ref{rem:uniformlam} the conclusion of \mainthm{} then holds simultaneously for all thresholds in that band, so it may be applied at the random $\lambda_a$.
\end{assumption}

\begin{assumption}[Bounded path]
\label{ass:ssicpath}
$\max_{\mathcal{T} \in \Pi_n} |\mathcal{T}| \le K_n$ for a sequence with
$K_n \log n = o(n)$.
\end{assumption}

Assumption~\ref{ass:ssicnoise} strengthens the noise condition of Section~4 of
the main paper from independent to independent and identically distributed with
a finite variance, which is what makes $\hat\sigma$ and $\|\varepsilon\|^2$
well-behaved; the convergence of $\hat\sigma$ holds for the median absolute deviation
with the Gaussian consistency factor when the errors are Gaussian, the $N$
contaminated second differences at the kinks being of size $O(n^{-1})$ and
finite in number.

Assumption~\ref{ass:ssicgrid} is discussed in Remark~\ref{rem:gridrange}.
Assumption~\ref{ass:ssicpath} excludes candidate models whose dimension is
comparable to $n$, for which $n\log\{\mathrm{RSS}(\mathcal{T})/n\}$ is not
informative. It is satisfied by any implementation that caps the number of
estimated change-points, and it holds automatically for those $\lambda_a$ that
exceed $\sigma(8\log n)^{1/2}$, for which \mainthm{} gives
$|\hat{\mathcal{T}}(\lambda_a)| = N$.

\subsection{Two lemmas}
\label{sup:ssiclem}

\begin{lemma}[Approximation error]
\label{lem:approx}
Let $\mathcal{T} \subset \{1,\ldots,n-1\}$ and write $\delta_j =
\delta_j(\mathcal{T})$. Then:
\begin{enumerate}
\item[\emph{(a)}] $b(\mathcal{T}) \le N \bar d^{\,2} n^{-1}
\sum_{j=1}^N \delta_j^2$;
\item[\emph{(b)}] if $\varrho < \delta_{\min}/2$, so that the interval $I_j := (\tau_j - \lfloor\varrho n\rfloor,\ \tau_j + \lfloor\varrho n\rfloor]$, whose endpoints are integers, contains no change-point other than $\tau_j$, and if $I_j$ contains exactly one point of $\mathcal{T} \cup \{0,n\}$ and $\delta_j \le \varrho n / 2$, then
$b(\mathcal{T}) \ge \tfrac18 c_\ell^2 c_0\, \underline d^{\,2}\, \varrho\, \delta_j^2 / n$, with $c_0$ and $c_\ell$ the constants of Lemmas~\ref{lem:sep} and~\ref{lem:lev};
\item[\emph{(c)}] if $\varrho < \delta_{\min}/2$ and $I_j$ contains no point of $\mathcal{T} \cup \{0,n\}$, then $b(\mathcal{T}) \ge \tfrac18 c_\ell^2 \underline d^{\,2} \varrho^3 n$.
\end{enumerate}
\end{lemma}

\begin{proof}
(a) Let $\hat\tau_j \in \mathcal{T} \cup \{0,n\}$ attain $\delta_j$ and take $g \in \mathcal{M}(\mathcal{T})$ with the same affine part as $f$ and knot coefficients $+\Delta_j$ at $\hat\tau_j$, other coefficients zero. Then
$h := f - g = \sum_j \Delta_j\{(t-\tau_j)_+ - (t-\hat\tau_j)_+\}$ up to an
affine function, and each bracket is bounded in modulus by $\delta_j$, so
$\|h\|_\infty \le \sum_j |\Delta_j| \delta_j$ and, by Cauchy--Schwarz and $|\Delta_j| \le \bar d/n$ from Section~4 of the main paper, $b(\mathcal{T}) \le \|h\|^2 \le n (\sum_j |\Delta_j|\delta_j)^2 \le
N \bar d^{\,2} n^{-1} \sum_j \delta_j^2$.

(b) Let $\hat\tau_j$ be the unique point of
$\mathcal{T}\cup\{0,n\}$ in $I_j$. To lower-bound $b(\mathcal{T})$, restrict
the squared approximation error to the coordinates in $I_j$. The signal $f$
has $\tau_j$ as its only kink there because
$\varrho<\delta_{\min}/2$. Every function in
$\mathcal{M}(\mathcal{T})$, when restricted to $I_j$, has
$\hat\tau_j$ as its only possible knot, since knots outside $I_j$ contribute
only an affine function on this interval. Therefore
$b(\mathcal{T}) \ge |\Delta_j|^2 \|\psi^{I_j}_{\tau_j}\|^2
\{1 - \rho(\tau_j,\hat\tau_j)^2\}$, where $\psi^{I_j}$ and $\rho$ are formed on
$I_j$. Writing $L := \lfloor \varrho n\rfloor$, so that $\tau_j$ is at distance $L$ from each end of $I_j$, Lemma~\ref{lem:lev} gives $\|\psi^{I_j}_{\tau_j}\|^2 \ge c_\ell^2 L^3 \ge c_\ell^2 (\varrho n)^3/8$ once $\varrho n \ge 2$, and
Lemma~\ref{lem:sep} gives $1-\rho^2 \ge 1-\rho \ge
c_0 \{\delta_j/(\varrho n)\}^2$. Multiplying the lower bounds for
$\|\psi^{I_j}_{\tau_j}\|^2$ and $1-\rho^2$, and then using
$|\Delta_j| \ge \underline d/n$, gives the bound.

(c) As in (b), but now $f$ restricted to $I_j$ is compared with affine functions only, so $b(\mathcal{T}) \ge |\Delta_j|^2\|\psi^{I_j}_{\tau_j}\|^2 \ge c_\ell^2\underline d^{\,2} n^{-2} L^3 \ge c_\ell^2\underline d^{\,2}\varrho^3 n/8$. Here and in (b) the factors $1/8$ and, in (b), $1/2$ in the separation ratio $\delta_j/L \le 2\delta_j/(\varrho n)$ are the price of using integer endpoints; no rate is affected.
\end{proof}

Note the two-sided form of Lemma~\ref{lem:approx}: part (a) says that a set
localized to within $\eta$ has approximation error at most of order $\eta^2/n$,
so a set localized to within $\bar r_n$ has error $O(\log n)$; part (b) says
that this is the right order, so that error $O(\log n)$ forces localization
$O(\bar r_n)$. 

\begin{lemma}[Uniform control of projected noise]
\label{lem:proj}
Under Assumption~\ref{ass:ssicnoise}, $\varepsilon_t/\sigma$ is
$1$-sub-Gaussian. The Hanson--Wright inequality
\citep[Theorem~1.1]{rv13} therefore has an absolute positive constant that is
independent of $n$, of the projection and of the error distribution within
this normalized sub-Gaussian class. Consequently, there is a universal
constant $C_P$ and an event $F_n$ with $\mathbb{P}(F_n)\to1$ on which
\[
\big\| P_{\mathcal{S}} \varepsilon \big\|^2 \;\le\;
C_P\, \sigma^2\, (d+2) \log n
\quad \text{whenever } |\mathcal{S}| = d \le n/2,
\]
and
\[
\Big| \|\varepsilon\|^2 - n\sigma^2 \Big| \;\le\; \tfrac12 n \sigma^2 .
\]
\end{lemma}

\begin{proof}
For fixed $\mathcal{S}$, $\|P_{\mathcal{S}}\varepsilon\|^2$ is a quadratic form in $\varepsilon$ whose matrix has rank at most $d+2$ and operator norm one, so the Hanson--Wright inequality gives $\mathbb{P}(\|P_{\mathcal{S}}\varepsilon\|^2 > \sigma^2(d+2) + t) \le 2\exp[-c\min\{t^2/\{\sigma^4 (d+2)\},\, t/\sigma^2\}]$. Taking $t = C_P\sigma^2 (d+2)\log n$ with $C_P$ large enough makes this at most $2\exp\{-3(d+2)\log n\}$, and there are at most $\binom{n}{d} \le \exp(d\log n)$ sets
of size $d$, so the first event fails with probability at most $2\sum_{d\le n/2} \exp\{-(2d+6)\log n\} \le 4n^{-6}$. The second is the law of large
numbers for $\varepsilon_t^2$, whose mean is $\sigma^2$ and whose variance is
finite under Assumption~\ref{ass:ssicnoise}.
\end{proof}

\subsection{The theorem}
\label{sup:ssicthm}

\begin{theorem}[Consistency of sSIC-selected LABS]
\label{thm:ssic}
Let the model and assumptions of \mainthm{} hold, together with
Assumptions~\ref{ass:ssicnoise} to~\ref{ass:ssicpath}, and let $\gamma > 1$.
Write $\hat N = |\hat{\mathcal{T}}^{\,\mathrm{sSIC}}|$ and let $\hat\tau_{(1)} < \cdots < \hat\tau_{(\hat N)}$ be the selected change-points. If $N \ge 1$ there is $C < \infty$, depending only on $g$, $\sigma$ and $\gamma$, such that \[ \mathbb{P}\Big( \hat N = N \ \text{ and } \ \max_{1 \le j \le N} \big| \hat\tau_{(j)} - \tau_j \big| \le C (n\log n)^{1/2} \Big) \;\longrightarrow\; 1 , \] and if $N = 0$ then $\mathbb{P}(\hat N = 0) \to 1$.
The same holds for LABS-Grid with fixed $M, R \ge 2$, with \mainthm{} replaced
by Theorem~\ref{thm:grid} and its assumptions in
Assumption~\ref{ass:ssicgrid}.
\end{theorem}

\begin{proof}
Let $G_n$ be the intersection of three events. The first is the uniform
contrast-noise event $E_n$ from Lemma~\ref{lem:noise}. The second is the
projected-noise event $F_n$ from Lemma~\ref{lem:proj}. The third is the
threshold-grid event asserted in Assumption~\ref{ass:ssicgrid}. Work on $G_n$,
whose probability tends to one. Write
$\mathcal{T}_0 = \{\tau_1,\ldots,\tau_N\}$, $q = |\mathcal{T}|$ and
$\Phi_q^2 := C_P\sigma^2(q+N+2)\log n$, and recall $\bar r_n = (n\log n)^{1/2}$.
Constants $c_1, c_2, \ldots$ depend only on $g$ and $\sigma$.

\medskip
\noindent\emph{Step 1: the path contains a good set.} By
Assumption~\ref{ass:ssicgrid} there is $a \in A_n$ with $\lambda_a$ admissible,
so \mainthm{} applies at that threshold and
$\mathcal{T}^\star := \hat{\mathcal{T}}(\lambda_a) \in \Pi_n$ satisfies
$|\mathcal{T}^\star| = N$ and $\delta_j(\mathcal{T}^\star) \le C_\star \bar r_n$
for every $j$. Lemma~\ref{lem:approx}(a) gives
$b(\mathcal{T}^\star) \le N^2 \bar d^{\,2} C_\star^2 \log n =: c_1 \log n$.

\medskip
\noindent\emph{Step 2: an upper bound for $\mathrm{RSS}(\mathcal{T}^\star)$.}
Put $\mathcal{A}^\star = \mathcal{T}^\star \cup \mathcal{T}_0$, so that
$f \in \mathcal{M}(\mathcal{A}^\star)$ and
$(I-P_{\mathcal{T}^\star})f \in \mathcal{M}(\mathcal{A}^\star)$. Hence
$|\langle (I-P_{\mathcal{T}^\star})f, \varepsilon\rangle| \le
b(\mathcal{T}^\star)^{1/2} \|P_{\mathcal{A}^\star}\varepsilon\|$, and
$|\mathcal{A}^\star| \le 2N$, so Lemma~\ref{lem:proj} gives
$\|P_{\mathcal{A}^\star}\varepsilon\|^2 \le C_P\sigma^2 (2N+2)\log n$. Therefore
\begin{equation}
\label{eq:rssstar}
\mathrm{RSS}(\mathcal{T}^\star) \;\le\; \|\varepsilon\|^2 +
b(\mathcal{T}^\star) + 2 b(\mathcal{T}^\star)^{1/2}
\|P_{\mathcal{A}^\star}\varepsilon\| \;\le\; \|\varepsilon\|^2 + c_2 \log n .
\end{equation}

\medskip
\noindent\emph{Step 3: a lower bound for $\mathrm{RSS}(\mathcal{T})$.} Let
$\mathcal{T} \in \Pi_n$ and $\mathcal{A} = \mathcal{T} \cup \mathcal{T}_0$.
Since $\mathcal{M}(\mathcal{T}) \subseteq \mathcal{M}(\mathcal{A})$ and
$P_{\mathcal{A}}f = f$,
\[
\mathrm{RSS}(\mathcal{T}) = \mathrm{RSS}(\mathcal{A}) +
\big\| (P_{\mathcal{A}} - P_{\mathcal{T}})X \big\|^2
\;\ge\; \|\varepsilon\|^2 - \|P_{\mathcal{A}}\varepsilon\|^2 +
\Big( b(\mathcal{T})^{1/2} - \|P_{\mathcal{A}}\varepsilon\| \Big)_+^2 ,
\]
using $(P_{\mathcal{A}}-P_{\mathcal{T}})X = (I-P_{\mathcal{T}})f +
(P_{\mathcal{A}}-P_{\mathcal{T}})\varepsilon$ and
$\mathrm{RSS}(\mathcal{A}) = \|(I-P_{\mathcal{A}})\varepsilon\|^2$. By
Lemma~\ref{lem:proj} and Assumption~\ref{ass:ssicpath},
$\|P_{\mathcal{A}}\varepsilon\|^2 \le \Phi_q^2 \le C_P\sigma^2(K_n+N+2)\log n = o(n)$, so
\begin{equation}
\label{eq:rsslower}
\mathrm{RSS}(\mathcal{T}) \;\ge\; \|\varepsilon\|^2 - \Phi_q^2 +
\big( b(\mathcal{T})^{1/2} - \Phi_q \big)_+^2 .
\end{equation}

\medskip
\noindent\emph{Step 4: sets with large approximation error.} Suppose
$b(\mathcal{T}) \ge \epsilon n$ for a fixed $\epsilon > 0$. By
\eqref{eq:rsslower} and $\Phi_q^2 = o(n)$,
$\mathrm{RSS}(\mathcal{T}) \ge \|\varepsilon\|^2 + \tfrac12\epsilon n$ for $n$
large, while \eqref{eq:rssstar} bounds
$\mathrm{RSS}(\mathcal{T}^\star)$ by $\|\varepsilon\|^2 + c_2\log n$. Since
$\|\varepsilon\|^2 \le \tfrac32 n\sigma^2$ on $F_n$,
\[ n \log \frac{\mathrm{RSS}(\mathcal{T})}{\mathrm{RSS}(\mathcal{T}^\star)} \;\ge\; n \log\Big( 1 + \frac{\epsilon n/2 - c_2\log n}{\|\varepsilon\|^2 + c_2\log n} \Big) \;\ge\; n \log\Big(1 + \frac{\epsilon}{5\sigma^2}\Big) \]
for $n$ large, the denominator being at most $2n\sigma^2$ on $F_n$,
which exceeds $2N\xi_n$ for $n$ large, and the penalty difference is at least
$-2N\xi_n$. So $\mathrm{sSIC}(\mathcal{T}) > \mathrm{sSIC}(\mathcal{T}^\star)$
and $\mathcal{T}$ is not selected.

In particular no $\mathcal{T}$ with $q < N$ is selected: by the pigeonhole
principle some $\tau_j$ has no point of $\mathcal{T}\cup\{0,n\}$ within
$\delta_{\min}n/3$, since the $N+1$ gaps between consecutive elements of
$\mathcal{T}_0 \cup \{0,n\}$ have length at least $\delta_{\min}n$, so
Lemma~\ref{lem:approx}(c) gives $b(\mathcal{T}) \ge c_3 n$.

\medskip
\noindent\emph{Step 5: sets with small approximation error.} Suppose
$b(\mathcal{T}) < \epsilon n$. Then by \eqref{eq:rssstar},
\eqref{eq:rsslower} and Lemma~\ref{lem:proj}, both
$\mathrm{RSS}(\mathcal{T})$ and $\mathrm{RSS}(\mathcal{T}^\star)$ lie between
$\tfrac13 n\sigma^2$ and $2n\sigma^2$ for $n$ large and $\epsilon$ small.
Writing $D := \mathrm{RSS}(\mathcal{T}) - \mathrm{RSS}(\mathcal{T}^\star)$ and
using $\log x \ge 1 - 1/x$,
\[
n \log \frac{\mathrm{RSS}(\mathcal{T})}{\mathrm{RSS}(\mathcal{T}^\star)}
\;\ge\; \frac{n D}{\mathrm{RSS}(\mathcal{T})}
\;\ge\; \frac{3D}{\sigma^2} \ \ \text{if } D \le 0, \qquad \ge\; \frac{D}{2\sigma^2} \ \ \text{if } D > 0 ,
\]
and by \eqref{eq:rssstar} and \eqref{eq:rsslower},
\begin{equation}
\label{eq:Dbound}
D \;\ge\; \big( b(\mathcal{T})^{1/2} - \Phi_q \big)_+^2 - \Phi_q^2
- c_2 \log n .
\end{equation}

\medskip
\noindent\emph{Step 6: no over-fitting.} Suppose $q > N$ and $b(\mathcal{T}) < \epsilon n$. If $D > 0$ then the logarithmic term is non-negative by the second branch of Step~5, while the penalty difference $2(q-N)\xi_n$ is strictly positive, and $\mathcal{T}$ is not selected. Assume therefore $D \le 0$; dropping the non-negative term in \eqref{eq:Dbound} and using the first branch of Step~5,
\begin{align*}
\mathrm{sSIC}(\mathcal{T}) - \mathrm{sSIC}(\mathcal{T}^\star)
&\;\ge\; -\frac{3}{\sigma^2}\big( \Phi_q^2 + c_2\log n \big) + 2(q-N)\xi_n \\
&\;\ge\; -3C_P(q+N+2)\log n - \frac{3c_2\log n}{\sigma^2} + 2(q-N)\xi_n .
\end{align*}
Since $q \ge N+1$ gives $q + N + 2 \le (2N+3)(q-N)$, the right side is at least $(q-N)\{2\xi_n - 3C_P(2N+3)\log n - 3c_2\sigma^{-2}\log n\}$, which is positive
for $n$ large because $\xi_n/\log n = (\log n)^{\gamma-1} \to \infty$. So
$\mathcal{T}$ is not selected.

\medskip
\noindent\emph{Step 7: localization.} By Steps~4 and~6 the selected set
$\hat{\mathcal{T}} := \hat{\mathcal{T}}^{\,\mathrm{sSIC}}$ has
$|\hat{\mathcal{T}}| = N$, which is the first assertion. Its penalty equals
that of $\mathcal{T}^\star$, so
$\mathrm{sSIC}(\hat{\mathcal{T}}) \le \mathrm{sSIC}(\mathcal{T}^\star)$ forces
$D \le 0$ in Step~5, and \eqref{eq:Dbound} with $q = N$ gives
\[
\begin{aligned}
\big( b(\hat{\mathcal{T}})^{1/2} - \Phi_N \big)_+^2
&\le \Phi_N^2 + c_2\log n, \\
b(\hat{\mathcal{T}})
&\le \big( 2\Phi_N + (c_2\log n)^{1/2} \big)^2
\le c_4 \log n.
\end{aligned}
\]
since $\Phi_N^2 = C_P\sigma^2(2N+2)\log n$. Lemma~\ref{lem:approx}(c), whose constant we write as $c_5 := \tfrac18 c_\ell^2 \underline d^{\,2}$, applied with any $\varrho$ exceeding $(c_4\log n/(c_5n))^{1/3}$, then shows that every $\tau_j$ has a point of $\hat{\mathcal{T}} \cup \{0,n\}$ within $\Gamma_n' := C_6 (n^2\log n)^{1/3}$ of it, where $C_6 > (c_4/c_5)^{1/3}$ is any fixed constant. Since $\Gamma_n' = o(n)$ and the $\tau_j$ are $\delta_{\min}n$ apart, the
$N$ points of $\hat{\mathcal{T}}$ are in one-to-one correspondence with the
$\tau_j$, and each interval $I_j = (\tau_j - \delta_{\min}n/3,\ \tau_j +
\delta_{\min}n/3)$ contains exactly one of them. Lemma~\ref{lem:approx}(b) with
$\varrho = \delta_{\min}/3$ now gives
\[
\tfrac{1}{24} c_\ell^2 c_0\, \underline d^{\,2}\delta_{\min}\, \delta_j(\hat{\mathcal{T}})^2 / n \;\le\; b(\hat{\mathcal{T}}) \;\le\; c_4 \log n ,
\]
so $\delta_j(\hat{\mathcal{T}}) \le C (n\log n)^{1/2}$ for every $j$, which is
the second assertion. If $N = 0$ then $\mathcal{T}^\star = \emptyset$ by \mainthm{}, $b(\mathcal{T}) = 0$ for every candidate, and Step~6 applies to every $\mathcal{T}$ with $q \ge 1$, so the empty set is selected. The statement for LABS-Grid is identical, Step~1 using Theorem~\ref{thm:grid} in place of \mainthm{}.
\end{proof}

\subsection{Remarks}
\label{sup:ssicrem}

\begin{remark}[Why the penalty exponent must exceed one]
\label{rem:gamma}
Step~6 is the only place where $\gamma > 1$ is used, and it cannot be dispensed with there. Adding $q - N$ superfluous knots, chosen from $n$ positions in a
data-dependent way, reduces the residual sum of squares by an amount that can
be as large as a constant multiple of $\sigma^2 (q-N)\log n$, the $\log n$
arising from the $\binom{n}{q}$ possible choices in Lemma~\ref{lem:proj}. The
penalty must therefore charge more than a constant multiple of $\log n$ for
each knot. The ordinary Schwarz criterion, with $\xi_n = \log n$, charges
exactly of that order and does not separate the two, whereas
$\xi_n = (\log n)^\gamma$ with $\gamma > 1$ does, by the factor
$(\log n)^{\gamma-1} \to \infty$. This is the reason for the strengthened form
used in Section~5 of the main paper.
\end{remark}

\begin{remark}[The range of the threshold grid]
\label{rem:gridrange}
The lower restriction in Assumption~\ref{ass:ssicgrid} is
\[
\lambda_a = a\hat\sigma(2\log n)^{1/2}
> \sigma(8\log n)^{1/2}.
\]
Since $\hat\sigma/\sigma \to 1$, this requires some $a \in A_n$ with $a > 2$.
The constant $8$ comes from the
union bound over all $O(n^3)$ triples $(s,e,b)$ in Lemma~\ref{lem:noise} and is
not sharp: LABS evaluates the contrast on $O(\hat N)$ windows, not on all of
them, as noted in Section~3.1 of the main paper, so the maximum that
$\lambda$ must exceed is over a much smaller collection. What
Theorem~\ref{thm:ssic} establishes is that the selection is consistent once the
grid reaches into the range in which the single-threshold theorem applies;
identifying the smallest such range for this contrast is a separate question,
and the simulation evidence indicates that the values used in practice are
adequate.
\end{remark}

\begin{remark}[Rate and the required sample size]
\label{rem:ssicrate}
The rate is the same as in \mainthm{}: selection by the criterion does not
alter its order. This is because the penalty difference between two
candidate sets of size $N$ is zero, so the localization in Step~7 is governed
by the $O(\log n)$ terms alone and not by $\xi_n$. The value of $n$ from which
the proof excludes over-fitting may nevertheless be large, because Step~6
requires $(\log n)^{\gamma-1}$ to exceed a fixed multiple of $2N+3$.
\end{remark}

\begin{remark}[What is not covered]
\label{rem:ssicnot}
Theorem~\ref{thm:ssic} treats selection over a path generated by thresholding.
It does not justify the use of $\hat\sigma$ inside the criterion itself, since
$\mathrm{sSIC}$ as defined does not involve $\hat\sigma$; nor does it cover
selection over a path generated in any other way, or the choice of $\gamma$.
Assumption~\ref{ass:ssicpath} is a restriction on the path rather than a
property of the algorithm. A dimension bound of this order is needed in the
proof of Theorem~\ref{thm:ssic} to keep the projected-noise term in
Lemma~\ref{lem:proj} of order $o(n)$. Section~\ref{sup:sizepath} imposes the
bound by construction rather than assuming it, and in doing so also removes
Assumption~\ref{ass:ssicgrid}.
\end{remark}

\subsection{A solution path indexed by the number of change-points}
\label{sup:sizepath}

Assumptions~\ref{ass:ssicgrid} and~\ref{ass:ssicpath} are auxiliary: the first
asks that the grid of multipliers reach far enough, the second that the path
not contain models whose dimension is comparable to $n$. Both can be removed,
without altering the algorithm, by generating the path over the whole range of
thresholds rather than over a grid, indexing it by the number of change-points
rather than by the threshold, and truncating it at a cap. This is the usual
setting for information criteria, which are ordinarily shown to be consistent
over a finite collection of candidate models of bounded size. The threshold
then plays no role beyond that of an internal device for generating candidates.
Nothing in Sections~\ref{sup:ssicproc} to~\ref{sup:ssicrem} is changed; what
follows is an alternative route to the same conclusion.

\subsubsection*{Reduction to one candidate per model size}

\begin{lemma}[Reduction to one model per size]
\label{lem:reduce}
Let $\mathcal{C}$ be a finite collection of candidate sets and for each $k$ let
$\mathcal{T}_k \in \arg\min\{\mathrm{RSS}(\mathcal{T}) : \mathcal{T} \in
\mathcal{C}, |\mathcal{T}| = k\}$ when that set is non-empty. Then
\[
\min_{\mathcal{T} \in \mathcal{C}} \mathrm{sSIC}(\mathcal{T})
\;=\; \min_{k} \mathrm{sSIC}(\mathcal{T}_k),
\]
and any minimizer of the right side is a minimizer of the left.
\end{lemma}

\begin{proof}
$\mathrm{sSIC}(\mathcal{T})$ depends on $\mathcal{T}$ only through
$\mathrm{RSS}(\mathcal{T})$ and $|\mathcal{T}|$, and is increasing in the
former for fixed $|\mathcal{T}|$.
\end{proof}

Lemma~\ref{lem:reduce} is elementary but it is what makes the reindexing
legitimate: selecting by $\mathrm{sSIC}$ over a collection of candidate sets is
the same as retaining the best set of each size and then choosing the size. No
information is lost by presenting the path as a sequence indexed by $k$.

\subsubsection*{The path}

\begin{defin}[Size-indexed path]
\label{def:sizepath}
Fix a cap $K$. Let
\[
\Pi_n^{K} := \big\{ \hat{\mathcal{T}}(\lambda) : \lambda \ge 0 \big\}
\cap \big\{ \mathcal{T} : |\mathcal{T}| \le K \big\},
\]
where $\hat{\mathcal{T}}(\lambda)$ is the output of LABS, or of LABS-Grid with
fixed $M$ and $R$, at threshold $\lambda$. Let $\mathcal{K}_n :=
\{|\mathcal{T}| : \mathcal{T} \in \Pi_n^K\}$, let $\mathcal{T}_k$ be a
minimum-$\mathrm{RSS}$ element of $\Pi_n^K$ of size $k$, and set
\[
\hat{\mathcal{T}}^{\,K} := \arg\min_{k \in \mathcal{K}_n}
\mathrm{sSIC}(\mathcal{T}_k) .
\]
\end{defin}

Since $\hat{\mathcal{T}}(\lambda)$ is piecewise constant in $\lambda$, with
breakpoints among the finitely many contrast values the algorithm computes,
$\Pi_n^K$ is a finite collection and $\hat{\mathcal{T}}^{\,K}$ is well defined.

\subsubsection*{Computing the path}

The proposal made at a node does not depend on the threshold. At a node with
interval $(s,e]$ the proposal $\hat b$ and its value $W$ are determined by the
data on that interval alone, and the two children are $(s,\hat b]$ and
$(\hat b,e]$; the threshold decides only which nodes are reached and what each
returns. The intervals and proposals therefore form a single binary tree, the
\emph{proposal tree} $\mathcal{F}$, and the same tree serves every $\lambda$.

The look-ahead takes no part in building $\mathcal{F}$. Were
\algmain{} itself run at $\lambda = 0$, every call would recurse to the
shortest evaluable intervals, so each child would return a point immediately
adjacent to its parent's proposal, each re-test interval would have length
$O(1)$, and the re-test would carry no information; those of length less than
three would not admit the contrast at all. The construction below therefore separates the two, building
$\mathcal{F}$ from the proposals alone and applying the look-ahead afterwards,
at each threshold in turn, where the re-test intervals are the ones \algmain{}
actually uses.

\begin{enumerate}
\item \emph{Proposal tree.} Starting from $(0,n]$ and recursing on
$(s,\hat b]$ and $(\hat b,e]$ whenever $B_{s,e} \ne \emptyset$, record at each
node its interval, its proposal and the proposal value. No threshold is used
and no re-test is performed. Nodes are expanded lazily, only as far as step~3
requires.

\item \emph{Output at a given $\lambda$.} Evaluate bottom-up on $\mathcal{F}$:
$\mathrm{out}(v) = \emptyset$ if $W_v \le \lambda$, and otherwise the return of
\algmain{} at $v$, formed from the children's outputs $\mathcal{L}$ and
$\mathcal{R}$, with $s'$, $e'$ and the re-test as there. This reproduces
$\hat{\mathcal{T}}(\lambda)$ exactly. The look-ahead enters here and only here,
and since $(s',e']$ depends on $\lambda$ through $\mathcal{L}$ and
$\mathcal{R}$, the re-test contrasts are recomputed at each threshold.

\item \emph{Threshold sweep.} Here a sweep means evaluating the output as
$\lambda$ decreases through its breakpoints down to zero. Record the output
after each pass and expand $\mathcal{F}$ as deeper nodes become active. Early
stopping is discussed below.

\item \emph{Selection.} Retain, for each $k \le K$, the recorded set of size
$k$ with the smallest $\mathrm{RSS}$, and minimize $\mathrm{sSIC}$ over $k$.
\end{enumerate}

The breakpoints of $\lambda \mapsto \hat{\mathcal{T}}(\lambda)$ lie among the
proposal values and the re-test values, and may be assembled exactly, bottom-up:
those at a node are the breakpoints of its two children together with $W_v$ and
one re-test value for each of the resulting sub-intervals of $\lambda$. In an
implementation it is simpler to sweep a grid of thresholds, refining it until no
new set of size at most $K$ appears.

Step~1 expands only the part of $\mathcal{F}$ reached by the smallest threshold
used, so its cost is that of one run of LABS at that threshold. Each pass in step~2 costs $O(n\log n)$ under balanced splits, the re-test intervals at a given depth being disjoint, so $P$ passes cost $O(Pn\log n)$ typically and $O(Pn^2)$ in the worst case, matching the corresponding costs for LABS itself. The tree must be expanded far
enough for the thresholds being swept. Truncating $\mathcal{F}$ by depth is not a substitute
for truncating the path by size.

Stopping the sweep once the output size first exceeds $K$ as $\lambda$
decreases is a heuristic rule. It would be valid if the map
$\lambda\mapsto|\hat{\mathcal{T}}(\lambda)|$ were non-increasing, equivalently
if the output size could only increase as the threshold is lowered. This
monotonicity can fail when the two children return sets of equal size but in
different positions: the re-test interval can then differ between two
thresholds, and the parent estimate may be retained at the higher threshold
but discarded at the lower one. In that case early stopping can omit members
of $\Pi_n^K$ occurring at smaller $\lambda$. Sweeping to $\lambda=0$ is always
valid and is what Definition~\ref{def:sizepath} and
Theorem~\ref{thm:ssicK} require. If the stated monotonicity does hold, once the
output size exceeds $K$ it cannot later return to $K$ or below, so early
stopping omits no member of $\Pi_n^K$.

\subsubsection*{Consistency}

\begin{theorem}[sSIC over the size-indexed path]
\label{thm:ssicK}
Let the model and assumptions of \mainthm{} hold, together with
Assumption~\ref{ass:ssicnoise}, and let $\gamma > 1$. Let $K = K_n \to \infty$
with $K_n \log n = o(n)$. Then $\hat{\mathcal{T}}^{\,K_n}$ of
Definition~\ref{def:sizepath} satisfies the conclusion of
Theorem~\ref{thm:ssic}. No analogue of Assumption~\ref{ass:ssicgrid} or
Assumption~\ref{ass:ssicpath} is required. The same holds for LABS-Grid with
fixed $M, R \ge 2$, with \mainthm{} replaced by Theorem~\ref{thm:grid}.
\end{theorem}

\begin{proof}
Assumption~\ref{ass:ssicpath} holds with the stated $K_n$ by construction,
since every element of $\Pi_n^{K_n}$ has at most $K_n$ points.

For Step~1 of the proof of Theorem~\ref{thm:ssic}, let $\lambda$ be any
threshold in the range admitted by \mainthm{}, which is non-empty. Then
$\mathcal{T}^\star := \hat{\mathcal{T}}(\lambda)$ has $|\mathcal{T}^\star| = N$
and $\delta_j(\mathcal{T}^\star) \le C_\star\bar r_n$ for every $j$, on an
event of probability tending to one. Since $K_n \to \infty$ we have
$N \le K_n$ for $n$ large, so $\mathcal{T}^\star \in \Pi_n^{K_n}$; no condition
on a grid of multipliers is needed, because the path is generated over all
$\lambda \ge 0$. By Lemma~\ref{lem:reduce},
\[
\mathrm{sSIC}\big( \hat{\mathcal{T}}^{\,K_n} \big) \;\le\;
\mathrm{sSIC}(\mathcal{T}_N) \;\le\; \mathrm{sSIC}(\mathcal{T}^\star) .
\]
Steps~2 to~7 of the proof of Theorem~\ref{thm:ssic} are unchanged: they compare
an arbitrary candidate with $\mathcal{T}^\star$, and use only that the selected
set has $\mathrm{sSIC}$ no larger than that of $\mathcal{T}^\star$ and that
every candidate has at most $K_n$ points.
\end{proof}

\subsubsection*{Remarks}

\begin{remark}[Choice of the cap]
\label{rem:capchoice}
The conditions $K_n\to\infty$ and $K_n\log n=o(n)$ make the cap
asymptotically non-binding while requiring $K_n=o(n/\log n)$. For example,
$K_n=\lceil n/(\log n)^2\rceil$ is admissible. In practice a fixed
$K$ exceeding any plausible number of change-points is what is used, and
Theorem~\ref{thm:ssicK} then holds for every signal with $N \le K$. This is the
same convention under which the Schwarz criterion is ordinarily shown to be
consistent, namely over a finite collection of models of bounded size.
\end{remark}

\begin{remark}[What is gained, and what is lost]
\label{rem:sizegain}
Relative to Theorem~\ref{thm:ssic}, both auxiliary conditions are removed:
Assumption~\ref{ass:ssicgrid} because the threshold sweep covers every
threshold, so the admissible range of \mainthm{} is reached by construction,
and
Assumption~\ref{ass:ssicpath} because the cap is imposed rather than assumed.
What is not removed is the requirement that the admissible range be non-empty,
which is \assterm{} and \assnondeg{} together with
Lemma~\ref{lem:noise}; the threshold sweep locates a suitable threshold
without knowing where it is, but it does not create one. Its computational
cost is a factor $P$ on the running time, in place of the $|A_n|$ runs of the
procedure of
Section~\ref{sup:ssicproc}.
\end{remark}

\section{Connections with other look-ahead methods}
\label{sup:lookahead}

This section expands on the brief discussion in the Introduction of the main paper, giving a fuller account of the connections between the look-ahead mechanism in LABS and related ideas in other areas of algorithm design.

\subsection{Constraint satisfaction and SAT solvers}

In constraint satisfaction problems (CSPs), a backtracking algorithm assigns values to variables one at a time, recursing when an assignment is made and backtracking when a dead-end is reached. Techniques for improving backtracking algorithms are traditionally classified into two categories \citep{kondrak1997, dechter2003}:

\begin{itemize}
    \item Look-ahead schemes attempt to foresee the effects of the current assignment on future, not yet assigned, variables. The simplest look-ahead technique is forward checking, which removes from the domains of unassigned variables any values that would be inconsistent with the current partial assignment. More sophisticated techniques such as maintaining arc consistency propagate constraints more extensively, potentially detecting dead-ends earlier. A prominent special case is the Boolean satisfiability problem (SAT), where look-ahead solvers \citep{hvm09} tentatively propagate the consequences of candidate variable assignments, detecting unit propagations and failures, before committing to a branching decision.

    \item Look-back schemes analyze dead-ends when they occur to extract useful information. The simplest look-back is chronological backtracking; more sophisticated techniques such as conflict-directed backjumping identify the variables responsible for the failure and backtrack directly to them, skipping irrelevant intermediate variables.
\end{itemize}

The distinction drawn in this literature is that look-ahead uses information about future parts of the search to improve current decisions, while look-back uses information about past failures to improve backtracking decisions. Empirical studies have found that stronger look-ahead can reduce the benefit of look-back techniques \citep{kondrak1997}.

LABS employs a mechanism with features of both. Like look-ahead, it uses information from child subproblems to refine decisions at the parent level: the child recursions can be viewed as exploring future parts of the search tree, and the change-points they detect provide information that improves the parent's decision. Like look-back, it uses information gathered after an initial decision, in the proposal phase, to revise that decision: the refinement phase re-evaluates the parent change-point in the light of what was discovered in the children.

There are also differences. In CSP and SAT look-ahead, information flows from the current assignment to future variables through constraint propagation; the algorithm does not actually recurse into future subproblems. In LABS, by contrast, we do recurse into the child subproblems and use the actual results of those recursions, not just propagated constraints, to refine the parent.

The connection to CSP also suggests a broader perspective: LABS can be viewed as a form of intelligent backtracking for the change-point detection problem. Just as conflict-directed backjumping in CSPs avoids exploring irrelevant parts of the search tree by identifying culprit variables, LABS avoids detecting spurious change-points by using child information to identify narrower, more relevant intervals. The recursive structure handles the bookkeeping automatically, much as the call stack in a recursive backtracking algorithm handles the state management for backjumping.

\subsection{Decision tree induction}

In the decision tree literature the connection to LABS is particularly direct. Standard CART \citep{bfso84} selects splits greedily: at each node, the best split is chosen without considering what splits will follow in the child nodes. \cite{ms95} study what happens when the algorithm instead looks one or more levels ahead, tentatively making a split, building the child trees, and using the quality of the resulting subtree to evaluate the parent split. They find that look-ahead can degrade performance in classification trees, a phenomenon they term pathology, and characterize conditions under which it helps and under which it hurts. \cite{em07} develop this idea further with anytime algorithms that invest varying amounts of look-ahead depending on the available computation time.

The structural parallel with LABS is that in both settings a recursive partitioning algorithm tentatively makes a split, or proposes a change-point, recurses into the children, and uses the results to refine or validate the parent decision. A difference is that in tree look-ahead the purpose is to choose a better split variable or threshold from scratch, whereas in LABS the child recursions serve to narrow the interval on which the parent change-point is re-evaluated, exploiting the specific geometry of the change-point detection problem.

\subsection{Numerical linear algebra}

In numerical linear algebra, \cite{pt85} propose a look-ahead modification of the Lanczos algorithm for unsymmetric matrices, in which the procedure tentatively explores further Krylov vectors before committing to a basis selection, thereby circumventing potential breakdowns. The look-ahead similarly involves tentative exploration before commitment, but the structure being explored, namely Krylov subspaces, and the purpose, namely avoiding numerical breakdown, are different from ours.

\subsection{Bayesian sequential analysis}

In Bayesian sequential analysis, the term look-ahead refers to procedures that, at each stage of sequential sampling, evaluate whether to stop collecting data or to continue by examining the expected Bayes risk a fixed number of steps into the future; see \cite{b85}, Section 7.4.6. Our use of the term is unrelated to this sequential-sampling context: in LABS, looking ahead refers to examining the results of deeper recursion before confirming a candidate change-point, not to deciding when to stop gathering observations.

\subsection{Further research directions}

Several questions remain open. First, one could seek sufficient conditions on the signal for \assterm{} and \assnondeg{}, or a formulation that does not require them. The case in which $N$ grows with $n$ also requires separate analysis because the population recursion is then no longer a fixed finite object. Second, LABS could be extended to other signal models, including piecewise-polynomial signals of higher degree and signals with both level and slope changes. Third, it would be useful to determine the smallest range of thresholds for which a solution path computed on a fixed grid is guaranteed to contain a consistent model, and to develop more efficient path-construction algorithms. Finally, further work could establish conditions under which increasing $M$ makes the assumptions of Section~\ref{sup:gridass} easier to satisfy and study how a suitable finite-sample threshold depends on $M$. The present consistency result does not address either question.

\subsection{Summary}

The strategy of full recursion followed by refinement is what distinguishes LABS among these uses of look-ahead. Rather than propagating constraints or information forward, we solve the child subproblems completely and then use the solutions to define a narrower interval for re-evaluating the parent. This is possible because the child recursions are computationally cheap, taking at most as long as the parent, and because the structure of the change-point problem allows the child solutions to inform the parent refinement directly.

\end{document}